\newif\ifsubmission\submissionfalse
\documentclass[sigconf, nonacm]{acmart}

\newcommand\vldbdoi{XX.XX/XXX.XX}

\newcommand\vldbavailabilityurl{URL_TO_YOUR_ARTIFACTS}
\newcommand\vldbpagestyle{plain} 

\usepackage{graphicx}
\usepackage{caption}
\usepackage{subcaption}
\usepackage{graphics}
\usepackage{booktabs}

\usepackage{amsthm}
\usepackage{amsmath}
\usepackage{bbm}

\usepackage{relsize}
\usepackage[noend]{algpseudocode}
\usepackage[ruled,vlined, linesnumbered]{algorithm2e}
\SetAlCapNameFnt{\footnotesize}
\SetAlCapFnt{\footnotesize}

\usepackage{multicol}
\usepackage[most]{tcolorbox}
\usepackage{color}
\usepackage{pgfplots}
\usepackage{colortbl,hhline}
\usetikzlibrary{patterns}
\pgfplotsset{compat=1.17}
\definecolor{mypurple}{HTML}{f875aa}
\definecolor{mypink}{HTML}{c6ffc1}
\definecolor{myred}{HTML}{3edbf0}
\definecolor{myorange}{HTML}{77acf1}
\definecolor{myblue}{HTML}{04009a}
\definecolor{mydarkpurple}{HTML}{E75480}
\definecolor{mydarkpink}{HTML}{22bc22}
\definecolor{mydarkred}{HTML}{8B0000}
\definecolor{mydarkorange}{HTML}{FF8C00}
\definecolor{mydarkblue}{HTML}{04009a}
\definecolor{mybluegreen}{HTML}{30BFBF}
\definecolor{mydarkgreen}{HTML}{0D47A1}
\usepackage{caption}
\usepackage{placeins}
\pgfplotsset{
    compat=1.3,
    legend image code/.code={
        \draw [#1] (0cm,-0.1cm) rectangle (0.6cm,0.1cm);
    },
}

\usepackage{thm-restate}
\usepackage{enumerate}
\usepackage{enumitem}
\setlist{noitemsep,topsep=0pt,parsep=0pt,partopsep=0pt}
\usepackage{mathtools}
\usepackage{xspace}
\usepackage{booktabs}
\usepackage{pifont}
\makeatletter
\def\thm@space@setup{\thm@preskip=2pt
\thm@postskip=2pt}
\makeatother

\theoremstyle{plain}
\newtheorem{theorem}{Theorem}[section]

\newtheorem{lemma}[theorem]{Lemma}

\newtheorem{definition}[theorem]{Definition}

\newtheorem{example}[theorem]{Example}

\newcommand{\defn}[1]{\textbf{\textit{#1}}}

\algrenewcommand\algorithmicindent{1em}%

\newcommand{\var}[1]{\mathrm{Var} \left[ #1 \right]}

\usepackage{thm-restate}

\DeclarePairedDelimiter\ceil{\lceil}{\rceil}
\DeclarePairedDelimiter\floor{\lfloor}{\rfloor}

\definecolor{mygreen}{RGB}{20,140,80}
\definecolor{linkcolor}{RGB}{0,0,230}
\definecolor{mylightgray}{RGB}{230,230,230}
\definecolor{verylightgray}{RGB}{245,245,245}

\newcommand{\etal}[0]{et al.\xspace}

\newcounter{myalgctr}

\newtcolorbox{OuterBox}[1][]{%
    breakable,
    enhanced,
    frame hidden,
    interior hidden,
    left=-5pt,
    right=-5pt,
    top=-5pt,
    float=p,
    boxsep=0pt,
    arc=0pt
#1}%

\newtcolorbox{InnerBox}[1][]{%
    enforce breakable,
    enhanced,
    colback=gray,
    colframe=white,
#1}%

\newenvironment{tbox}{
\vspace{0.2cm}
\begin{tcolorbox}[width=\columnwidth,
                  enhanced,
                  boxsep=2pt,
                  left=1pt,
                  right=1pt,
                  top=4pt,
                  boxrule=1pt,
                  arc=0pt,
                  colback=white,
                  colframe=black,
	              breakable
                  ]
}{
\end{tcolorbox}
}

\newcommand{\tboxhrule}[0]{\vspace{0.1cm} {\color{black} \hrule} \vspace{0.2cm}}

\newenvironment{titledtbox}[1]{\begin{tbox}#1 \tboxhrule}{\end{tbox}}

\newcommand{\core}{k}

\newcommand{\eps}{\varepsilon}

\newcommand{\gn}{\mathcal{F}}

\newcommand{\kest}{\hat{\core}}

\newcommand{\whp}{whp\xspace}

\newcommand{\prob}{\mathsf{Pr}}
\newcommand{\geom}{\mathsf{Geom}}
\newcommand{\lf}{\eta/5}
\newcommand{\upexp}{(1+\lf)}
\newcommand{\upexpold}{(2+\lambda)(1+\lf)}

\newcommand{\hnup}{\widehat{\nup}}
\newcommand{\expect}{\mathbb{E}}

\newcommand{\nparam}{\nparaminside}

\newcommand{\lcur}{r}
\newcommand{\nparaminside}{\eps/(8\log^2n)}
\newcommand{\domain}{\mathcal{D}}

\newcommand{\dedge}{X}

\newcommand{\degen}{d}

\newcommand{\multfactor}{\phi}

\newcommand{\addfactor}{\zeta}
\newcommand{\order}{D}
\newcommand{\df}{\gs_f}
\newcommand{\gs}{GS}

\newcommand{\return}{\textbf{Return}\xspace}
\newcommand{\draw}{\geom(\nparam)}

\newcommand{\reals}{\mathbb{R}}
\newcommand{\norm}[1]{\left\Vert#1\right\Vert}

\newcommand{\range}{Range}
\newcommand{\rangeout}{\mathcal{Y}}
\newcommand{\adj}{\mathbf{a}}

\newcommand{\lap}{\text{Lap}}

\newcommand{\rout}{S}

\newcommand{\integers}{\mathbb{Z}}

\newcommand{\numgrouplevels}{2\log n}

\newcommand{\mech}{\mathcal{M}}
\newcommand{\nup}{\mathcal{U}}

\newcommand{\alg}{\mathcal{A}}

\newcommand{\tcountalgo}{$\text{EdgeOrient}_\Delta$}
\newcommand{\kcorealgo}{$k$-CoreD}

\ifsubmission
\newcommand{\julian}[1]{}
\newcommand{\pranay}[1]{}
\newcommand{\qq}[1]{}
\newcommand{\babis}[1]{}
\else
\newcommand{\julian}[1]{{\color{cyan} Julian: #1}}
\newcommand{\pranay}[1]{{\color{brown}#1}}
\newcommand{\qq}[1]{{\color{blue} Quanquan: #1}}
\newcommand{\babis}[1]{{\color{red} Babis: #1}}
\fi

\newcommand{\revision}[1]{#1}

\newcommand{\privatefraction}{f}

\newcommand{\myparagraph}[1]{\vspace{1pt}\noindent {\bf #1.}}
\newcommand{\eat}[1]{}
\newcommand{\squishlist}{
 \begin{list}{$\bullet$}
  { \setlength{\itemsep}{0pt}
     \setlength{\parsep}{1pt}
     \setlength{\topsep}{1pt}
     \setlength{\partopsep}{0pt}
     \setlength{\leftmargin}{1em}
     \setlength{\labelwidth}{1em}
     \setlength{\labelsep}{0.5em} } }

\newcommand{\squishend}{
  \end{list}
}

\usepackage[font=small,aboveskip=0pt,belowskip=0pt,tableposition=above]{caption}

\usepackage{footnote}
\makesavenoteenv{tabular}
\makesavenoteenv{table}
\usepackage{tablefootnote}

\usepackage[capitalize,nameinlink]{cleveref}
\crefname{theorem}{Theorem}{Theorems}
\Crefname{lemma}{Lemma}{Lemmas}
\Crefname{claim}{Claim}{Claims}
\Crefname{observation}{Observation}{Observations}
\Crefname{algorithm}{Algorithm}{Algorithms}
\Crefname{myalgctr}{Algorithm}{Algorithms}
\Crefname{invariant}{Invariant}{Invariant}
\Crefname{challenge}{Challenge}{Challenges}

\makeatletter
\renewcommand\theHALG@line{\thealgorithm.\arabic{ALG@line}}
\makeatother

\algblock{Input}{EndInput}
\algblock{Output}{EndOutput}
\algblock{ReduceAdd}{EndReduceAdd}

\algnewcommand\algorithmicinput{\textbf{Input:}}
\algnewcommand\algorithmicoutput{\textbf{Output:}}
\algnewcommand\algorithmicreduceadd{\textbf{ReduceAdd}}
\algnewcommand\algorithmicpardo{\textbf{do}}
\algrenewtext{Input}[1]{\algorithmicinput\ #1}
\algrenewtext{Output}[1]{\algorithmicoutput\ #1}
\algrenewtext{ReduceAdd}[2]{#1 $\leftarrow$ \algorithmicreduceadd(#2)}
\algtext*{EndInput}
\algtext*{EndOutput}
\algtext*{EndIf}
\algtext*{EndFor}
\algtext*{EndWhile}
\algtext*{EndParFor}
\algtext*{EndReduceAdd}

\SetKwFor{ParFor}{parfor}{\textbf{do}}{}%

\copyrightyear{2024}
\acmYear{2024}
\setcopyright{rightsretained}

\usepackage{etoolbox}
\makeatletter
\patchcmd{\maketitle}{\@copyrightpermission}{
   \begin{minipage}{0.3\columnwidth}
     \href{https://creativecommons.org/licenses/by-nc-nd/4.0/}{\includegraphics[width=0.90\textwidth]{figures/cc-by-nc-nd4acm.png}}
   \end{minipage}\hfill
   \begin{minipage}{0.7\columnwidth}
     \href{https://creativecommons.org/licenses/by-nc-nd/4.0/}{}
   \end{minipage}

   \vspace{5pt}
}{}{}

\makeatother

\begin{document}
\begin{abstract}
The rise of massive networks across diverse domains necessitates sophisticated graph analytics, often involving sensitive data and raising privacy concerns. This paper addresses these challenges using \emph{local differential privacy (LDP)}, which enforces privacy at the individual level, where \emph{no third-party entity is trusted}, unlike centralized models that assume a trusted curator.

We introduce novel LDP algorithms for two fundamental graph statistics: $k$-core decomposition and triangle counting. Our approach leverages previously unexplored input-dependent private graph properties, specifically the degeneracy and maximum degree of the graph, to improve theoretical utility. Unlike prior methods, our error bounds are determined by the maximum degree rather than the total number of edges, resulting in significantly tighter theoretical guarantees. For triangle counting, we improve upon the previous work of Imola, Murakami, and Chaudhury~\cite{IMC21locally, IMC21communication},
which bounds error in terms of the total number of edges. Instead, our algorithm achieves error bounds based on the graph's degeneracy by leveraging a differentially private out-degree orientation, a refined variant of Eden et al.'s \eat{[ICALP `23]} randomized response technique~\cite{ELRS23}, and a novel, intricate analysis, yielding improved theoretical guarantees over prior state-of-the-art.

Beyond theoretical improvements, we are the first to evaluate the practicality of local DP algorithms in a distributed simulation environment,
unlike previous works that tested on a single processor. Our experiments on real-world datasets demonstrate substantial accuracy improvements, with our $k$-core decomposition achieving errors within \(\mathbf{3}\)\textbf{x} the exact values—far outperforming the \(\mathbf{131}\)\textbf{x} error in the baseline of Dhulipala et al.~\cite{DLRSSY22} \eat{[FOCS `22]}. Additionally, our triangle counting algorithm reduces multiplicative approximation errors by up to \textbf{six orders of magnitude} compared to prior methods, all while maintaining competitive runtime performance.
\end{abstract}

\title{Practical and Accurate Local Edge Differentially Private Graph Algorithms}

\author{Pranay Mundra}
\affiliation{\institution{Yale University} \city{New Haven} \state{CT}\country{USA}}
\email{pranay.mundra@yale.edu}
\author{Charalampos Papamanthou} 
\affiliation{\institution{Yale University} \city{New Haven} \state{CT}\country{USA}}
\email{charalampos.papamanthou@yale.edu}
\author{Julian Shun} 
\affiliation{\institution{MIT CSAIL}\city{Cambridge} \state{MA}\country{USA}}
\email{jshun@mit.edu}
\author{Quanquan C. Liu} 
\affiliation{\institution{Yale University} \city{New Haven} \state{CT}\country{USA}}
\email{quanquan.liu@yale.edu}

\renewcommand{\shortauthors}{Pranay Mundra, Charalampos Papamanthou, Julian Shun, Quanquan C. Liu}



\settopmatter{printfolios=true}
\maketitle

\pagestyle{\vldbpagestyle}

\ifdefempty{\vldbavailabilityurl}{}{
\begingroup\small\noindent\raggedright\textbf{Artifact Availability:}\\
The source code, data, and/or other artifacts have been made available at \url{https://doi.org/10.5281/zenodo.15741879}. 
\endgroup
}


\section{Introduction}
\begin{figure}[!t]
    \centering
    \includegraphics[width=0.65\linewidth]{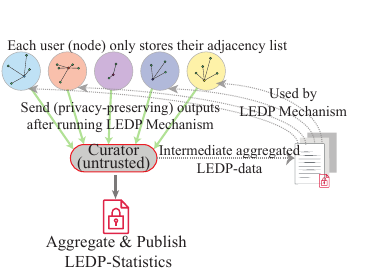}
    \caption{Local edge differential privacy (LEDP) Model}
    \label{fig:ledp-model}
\end{figure}
Graph statistics such as the $k$-core decomposition and triangle count provide important characteristics about the underlying graph, such as its well-connected communities.
These analytics, often performed on sensitive and private graphs, are regularly shared with a wide audience, including researchers, companies, governments, and the public. As such, it is vital to investigate techniques that can safeguard these published graph statistics from privacy attacks.

\eat{The \emph{$k$-core decomposition} assigns an ``importance'' value to each individual in a graph, representing roughly 
its influence on the rest of the graph.; such graphs
include social, disease transmission, email, or 
any other type of real-world network. The $k$-core decomposition is a powerful tool 
for informing researchers about the composition of 
real-world graphs, such as social, disease transmission, or email graphs. It is widely used to analyze the structure of real-world graphs, including social, email, and disease transmission networks.}
The \emph{$k$-core decomposition} assigns an ``importance'' value to each node, roughly representing its influence within the graph. It is widely used to analyze the structure of real-world graphs, including social, email, and disease transmission networks. Formally, a $k$-core of a graph is a maximal subgraph
where the induced degree of every vertex in the subgraph is at least $k$. 
The $k$-core decomposition (see~\Cref{def:kcore} and~\Cref{fig:kcore}) assigns a number, denoted as $core(v)$, to each vertex $v$. 
This number, $core(v)$, represents the largest value of $k$ for which the $k$-core still includes vertex $v$.
Unfortunately, these values pose 
privacy risks. 


Consider the application of $k$-core decomposition to COVID transmission data~\cite{qin2020analysis,s2021superspreading,serafino2022digital,guo2022integrative,chang2022comparative}
and other disease transmission networks~\cite{ciaperoni2020relevance} such as HIV~\cite{fujimoto2022integrated,holme2017cost}.
The \eat{$k$-core decomposition is} core numbers are generated and published, sometimes even with location data~\cite{covid19kcore}.
Revealing the precise core numbers of every individual can lead to privacy breaches.
Consider a scenario where exactly $c$ individuals have a core number of $c-1$. 
This implies they form a clique of $c$ vertices, all connected. 
In disease transmission graphs, this directly exposes a cluster of sensitive disease transmissions! 
Therefore, it's essential to release privacy-preserving core numbers. 

\eat{Like $k$-core decomposition, \emph{triangle counting} also has broad applications in various problems that use sensitive data.
The triangle count of a graph is the number of cycles with exactly three vertices.
Applications include community detection in social networks~\cite{palla2005uncovering,prat2012shaping,ST21,MSEW18},
recommendation systems~\cite{ZZLSZ14,LZZY18}, 
social and network science measurements~\cite{granovetter1977strength,
watts1998collective, newman2004finding,farkas2001spectra,foucault2010friend,leskovec2008microscopic}, 
and clustering~\cite{tseng2021parallel}.
There is much recent interest in the security community~\cite{IMC21communication,IMC21locally,liu2024edge,liu2022collecting,liu2023cargo}
in protecting user privacy when outputting triangle counts because of the potential of privacy attacks.}

Similarly, \emph{triangle counting} is widely used in applications that process sensitive data. The triangle count, which measures the number of three-node cycles in a graph, is a fundamental metric in community detection~\cite{palla2005uncovering,prat2012shaping,ST21,MSEW18}, recommendation systems~\cite{ZZLSZ14,LZZY18}, and clustering~\cite{tseng2021parallel}. In databases, triangle counting is essential for graph analytics frameworks~\cite{noauthororeditorneo4j} and is leveraged in query optimization~\cite{al2018triangle, bishnu2025arboricityrandomedgequeries} and fraud detection~\cite{computers10100121}. However, exposing triangle counts without privacy guarantees can lead to inference attacks that compromise user confidentiality. Recent works in security and privacy~\cite{IMC21communication,IMC21locally,liu2024edge,liu2022collecting,liu2023cargo} have highlighted the risks associated with publishing triangle counts, reinforcing the importance of privacy-preserving graph analytics.

\eat{As databases increasingly incorporate graph-based queries and analytics, ensuring privacy in fundamental graph statistics such as $k$-core decomposition and triangle counting becomes critical. This motivates the need for rigorous privacy mechanisms that protect structural graph information while preserving utility.}
\revision{

Approximate values for such statistics are widely used to improve scalability and efficiency with minimal utility impact. In graph databases and uncertain networks, approximate cores enable analysis under probabilistic edges~\cite{Bonchi2014Core}; in social and recommendation systems, they support influence detection and improve accuracy~\cite{kitsak2010,doi:10.1142/S012918312150087X}. Often, approximate core numbers are used for preprocessing other algorithms such as clustering~\cite{MTXP21,giatsidis2016k}. Triangle counts are used in clustering, fraud detection, and query optimization, where small errors are tolerable~\cite{Alon1997Triangle,liu2023cargo,DBLP:journals/corr/abs-2502-15379}. In dynamic domains like cybersecurity and biology, approximations allow timely insights from noisy data~\cite{SL13,WLT25,liu2023some}. Approximate statistics with strong privacy guarantees are often practical and effective when these applications use sensitive data~\cite{liu2023cargo,DLRSSY22,IMC21locally,IMC21communication,ELRS23,liu2023some}.
}



\revision{
Differential privacy (DP)~\cite{DMNS06} is often considered the ``gold standard'' in protecting individual privacy. Traditionally, DP has been studied in the \textbf{central model}, where a trusted curator has access to raw data and applies DP mechanisms before releasing the results. However, this assumption is often impractical, especially in modern systems that rely on decentralized or federated architectures. This motivates the \textbf{local model} of differential privacy, introduced by the seminal works of 
Evfimievski \etal~\cite{evfimievski2003limiting} and Kasiviswanathan \etal~\cite{kasiviswanathan2011can}, which
recently gained much attention in the theoretical computer science~\cite{DLRSSY22,ELRS23,dinitz2024tight,HSZ24},
cryptography and security~\cite{IMC21communication,IMC21locally,fu2023gc,li2023fine,liu2022crypto,wang2023mago,zhong2023disparate,liu2023cargo}, 
data mining~\cite{hillebrand2023communication,hillebrand2023unbiased,liu2024edge,liu2022collecting,
murakami2023automatic,brito2023global,nguyen2024faster}, query answering~\cite{10.1145/3299869.3300102, 10415644, 10.1145/3299869.3319891,farias2020local, roth2021mycelium},
and machine learning~\cite{wang2023differential,lin2022towards,mueller2022sok,hidano2022degree,zhu2023blink,xiang2023preserving,
seif2022differentially,guo2023community,hehir2022consistent,jiang2023personalized} communities. In this model, each user independently applies DP mechanisms before sharing their privacy-preserving outputs with an untrusted curator. It offers stronger privacy guarantees by never exposing raw data and is naturally suited to distributed settings; it has been used in prominent cases including federated learning~\cite{mahawaga2022local,NHC22,KGS21}, 
the 2020 U.S.\ Census \cite{abowd20222020}, and by companies such as Apple~\cite{appleLearningWith}.

While DP in the traditional database setting focuses on protecting individual records, many real-world datasets are inherently relational, represented as graphs. In these cases, the sensitive information are the edges, i.e., the connections between entities. This motivates the need for \textit{local differentially private (LDP) graph algorithms}. The \emph{local edge differential privacy} model (LEDP), as introduced in recent works~\cite{Qin17,IMC21locally,DLRSSY22}, represents a novel approach designed to ensure local privacy for \emph{graph outputs} (\cref{fig:ledp-model}). Graph data is increasingly integral to modern database systems, underpinning applications in knowledge graphs, social networks, cybersecurity, and financial fraud detection. Many relational databases now support graph extensions (e.g., SQL-based graph queries), while specialized graph databases~\cite{noauthororeditorneo4j, bebee2018amazon} are widely deployed in industry. However, applying LEDP algorithms in these settings is challenging: unlike in tabular data, where individual data points can be perturbed, perturbing edges within a graph introduces structural dependencies and high computational cost. 

All previous implementations of local differentially private graph algorithms~\cite{Qin17,imola2023differentially,IMC21communication}
use \emph{Randomized Response (RR)}~\cite{warner1965randomized}, which independently flips the presence or absence of each edge with a probability based on the privacy parameter $\eps$. While simple and composable, RR introduces substantial noise, especially for small $\eps \in (0, 1]$, increasing the \emph{density}\footnote{The density is the ratio of the 
number of edges to the number of nodes.} of the input graph and limiting both the accuracy and scalability of the algorithm. Dhulipala et al.~\cite{DLRSSY22} proposed the first LEDP algorithm that goes beyond RR, leveraging the geometric mechanism for $k$-core decomposition. However, their algorithm is purely theoretical, and their error bounds scale poorly with graph size. In particular, they give additive error bounds $\geq \frac{\log_2^3(n)}{\eps}$ (where $n$ is the number of vertices);
on a graph with $n = 10^5$ nodes and $\eps = 0.5$, this translates to an additive error of $9164$. Most real-world graphs of this size 
have max core numbers of $10^2$ magnitude, so the additive error itself leads to a $\geq 91$-multiplicative
approximation factor---much too large for any practical use. 

This work simultaneously develops \emph{both} new theoretical and implementation techniques that, together with Randomized Response, enable provably private, accurate, and computationally efficient LEDP algorithms. We make the following contributions:

\squishlist
    \item We design a \emph{novel} LEDP $k$-core decomposition algorithm that doesn't use Randomized Response and provides provable privacy and error guarantees. Leveraging the input-dependent maximum degree property of the graph, we achieve improved theoretical bounds over the LEDP $k$-core decomposition algorithm of Dhulipala et al.~\cite{DLRSSY22}~(see~\cref{table:additive_error}). Two key innovations lie in thresholding the maximum number of levels a node can move up based on its noisy (private) degree
    and the use of bias factors to reduce the impact of noise. Since a node's core number is upper bounded by its degree, our algorithm offers stronger theoretical guarantees for most real-world graphs, where the maximum degree is significantly smaller than the number of nodes.  

    \item We present the \emph{first} implementation of a private $k$-core decomposition algorithm and demonstrate through empirical evaluation that it achieves an average error of $\mathbf{3}$\textbf{x} the exact values, markedly improving upon prior approaches. Furthermore, our LEDP implementation attains error rates that closely align with the theoretical approximation bounds of the best \emph{non-private} algorithms, underscoring its practical efficiency and accuracy.  

    \item We design a \emph{novel} LEDP triangle counting algorithm that modifies our $k$-core decomposition to construct a low out-degree ordering, minimizing each node’s out-degree. Leveraging this ordering, our algorithm achieves improved theoretical error bounds over the best-known methods of Imola et al.~\cite{IMC21communication,IMC21locally} and Eden et al.~\cite{ELRS23} for bounded degeneracy graphs, common in real-world networks~(see~\cref{table:additive_error}). We present a novel analysis to analyze the non-trivial error bounds based on the Law of Total Expectation/Variance. Our implementation reduces relative error by up to \textbf{89x} and improves the multiplicative approximation by up to \textbf{six orders of magnitude} over the best previous implementation~\cite{IMC21communication}, while maintaining competitive runtime.

    \item Recognizing that the LEDP model (see~\cref{sec:ledp-dist}) is inherently decentralized, we present the first evaluation of LEDP graph algorithms in a simulated distributed environment with actual message passing. Unlike prior studies that relied on a single processor, we simulate a distributed environment by partitioning the graph across multiple processors. This approach provides a more realistic assessment of both computational and communication overhead in large-scale distributed scenarios. We demonstrate the practicality of this evaluation by applying it to our $k$-core decomposition and triangle counting algorithms.

    \item We present the first LEDP graph algorithm implementations that scale to \defn{billion}-edge graphs, whereas prior implementations were tested on graphs with millions of edges~\cite{IMC21locally,IMC21communication}. Our evaluation framework serves as a valuable tool for designing and testing other LEDP algorithms. Our source code is available at~\cite{githubCode}.   
\squishend
}


\begin{table}[!t]
    \centering
    \resizebox{\columnwidth}{!}{%
    \begin{tabular}{|c|c|c|}
    \hline
    & Previous Work & This Work \\
    \hline 
    $k$-Core & $O\left(\frac{\log^3(n)}{\eps}\right)$~\cite{DLRSSY22} & $O\left(\frac{\log(D_{\max})\log^2(n)}{\epsilon}\right)$ \\
    \hline 
    Triangle Counting & $O\left(\frac{n^{3/2}}{\eps^3} + \frac{\sqrt{C_4}}{\eps^2}\right)$~\cite{ELRS23} & $O\left(\frac{\sqrt{n}d\log^3 n}{\eps^{2}} + \sqrt{\overrightarrow{C}_4}\right)$ \\
    \hline 
    \end{tabular}
    }
    \caption{\revision{\footnotesize{Additive error bounds compared to previous work. Here $n:$ number of nodes, $D_{max}:$ maximum degree, $\eps:$ privacy parameter, $d:$ degeneracy, and $C_4$ and $\overrightarrow{C}_4$ is the number of $4$-cycles and oriented $4$-cycles, respectively.}}}
    \label{table:additive_error}
\end{table}

\eat{
\eat{\smallskip}\noindent\textbf{Central and local differential privacy.}
Differential privacy~\cite{DMNS06} is often considered the ``gold standard'' in protecting the privacy of 
individuals. Differential privacy has traditionally been studied in the \emph{central}
model where \emph{all} of the private information is revealed to a \emph{trusted curator} who is 
responsible for parsing the data and producing privacy-preserving outputs. But such an assumption is often unrealistic.
Today's world requires computation over \emph{decentralized}
networks involving the transfer of sensitive personal information; hence,
preserving \emph{local} privacy is crucial. 

The \emph{local model of differential privacy}, introduced by the seminal works of 
Evfimievski \etal~\cite{evfimievski2003limiting} and Kasiviswanathan \etal~\cite{kasiviswanathan2011can},
has recently gained much attention in the theoretical computer science~\cite{DLRSSY22,ELRS23,dinitz2024tight,HSZ24},
cryptography and security~\cite{IMC21communication,IMC21locally,fu2023gc,li2023fine,liu2022crypto,wang2023mago,zhong2023disparate,liu2023cargo}, 
data mining~\cite{hillebrand2023communication,hillebrand2023unbiased,liu2024edge,liu2022collecting,
murakami2023automatic,brito2023global,nguyen2024faster}, query answering~\cite{10.1145/3299869.3300102, 10415644, 10.1145/3299869.3319891,farias2020local, roth2021mycelium},
and machine learning~\cite{wang2023differential,lin2022towards,mueller2022sok,hidano2022degree,zhu2023blink,xiang2023preserving,
seif2022differentially,guo2023community,hehir2022consistent,jiang2023personalized} communities. 
\eat{In this model, each user independently applies differential privacy mechanisms to their data before sharing their (privacy-preserving)
outputs with an \emph{untrusted} curator. The curator's sole responsibility is to aggregate these outputs and compute relevant statistics, without
accessing the original, sensitive, private data. Unlike the central model, where a trusted curator collects raw and sensitive data, the local model offers the 
\emph{strongest privacy guarantees} by preventing data sharing with any third party.}
In this model, users independently apply differential privacy mechanisms before sharing their (privacy-preserving) outputs with an \emph{untrusted} curator, who aggregates the data to compute statistics without accessing the original sensitive information. \eat{Unlike the central model, where a trusted curator collects raw data, the} The local model ensures \emph{stronger privacy guarantees} by preventing direct data sharing.  
Because of both its strong privacy guarantees
and its easy adaptation to distributed applications, the local model of differential privacy
has been used in a variety 
of prominent cases including federated learning~\cite{mahawaga2022local,NHC22,KGS21}, 
the 2020 U.S.\ Census \cite{abowd20222020}, and by companies such as Apple~\cite{appleLearningWith}.

\eat{\noindent\textbf{Local Edge Differential Privacy and Randomized Response.} Despite interest in local privacy, local differentially private \emph{graph algorithms} have only recently received much interest. The \emph{local edge differential privacy} model (LEDP), as introduced in recent works~\cite{Qin17,IMC21locally,DLRSSY22}, represents a novel approach designed to ensure local privacy for \emph{graph outputs} (\cref{fig:ledp-model}). In contrast to traditional databases, where the focus is on protecting individual entries, the primary concern in graph data is protecting the sensitive relationships between individuals. Graph data is increasingly central to modern database systems, powering applications in knowledge graphs, social network analytics, cybersecurity, and financial fraud detection. Many relational databases now support graph extensions (e.g., SQL-based graph queries), and specialized graph databases (e.g., Neo4j, Amazon Neptune) are widely used in industry. Implementing LEDP graph algorithms presents significant challenges, primarily due to the large errors introduced by the strong privacy guarantees and the considerable computational resources required. Unlike the setting where individual data points can be perturbed using noise, LEDP graph algorithms involve perturbing the edges within a graph, which is inherently more complex.}

\noindent\textbf{Local Edge Differential Privacy and Randomized Response.} Despite growing interest in local privacy, \emph{local differentially private (LDP) graph algorithms} have only recently gained significant attention. The \emph{local edge differential privacy (LEDP) model}, introduced in recent works~\cite{Qin17,IMC21locally,DLRSSY22}, provides a novel framework for ensuring local privacy in \emph{graph outputs} (\cref{fig:ledp-model}). Unlike traditional databases, where privacy mechanisms focus on protecting individual records, graph data introduces the additional challenge of safeguarding \emph{sensitive relationships} between entities. Graph data is increasingly integral to modern database systems, underpinning applications in knowledge graphs, social network analysis, cybersecurity, and financial fraud detection. Many relational databases now support graph extensions (e.g., SQL-based graph queries), while specialized graph databases~\cite{noauthororeditorneo4j, bebee2018amazon} are widely deployed in industry. These systems rely on efficient graph analytics, yet integrating LEDP mechanisms remains challenging due to the high computational overhead and large errors induced by strict privacy constraints. Unlike the setting where individual data points can be perturbed using noise, LEDP graph algorithms involve perturbing the edges within a graph, which is inherently more complex.

All previous implementations of local differentially private graph algorithms~\cite{Qin17,imola2023differentially,IMC21communication}
use \emph{Randomized Response (RR)}~\cite{warner1965randomized} to perturb graph edges. 
This mechanism systematically examines every pair of nodes in the input graph: for each existing edge, it deletes it with probability $1/(e^{\eps} + 1)$, and for each non-existent edge, it inserts an edge with the same probability, where $\eps$ is the privacy parameter.
For small $\eps$ (in particular $0 < \eps \leq 1$), the \emph{density}\footnote{The density is the ratio of the 
number of edges to the number of nodes.}
of the input graph increases 
by a significant amount. 
Thus, while Randomized Response is an easy and convenient method for obtaining LEDP, it comes at a cost of large error 
and computation time. \eat{A major goal of our paper is to introduce additional new theoretical and implementation 
techniques that together with Randomized Response helps us obtain provably private, accurate, and computationally efficient LEDP algorithms.}

\eat{\smallskip}\noindent\textbf{Beyond Randomized Response.} 
Dhulipala et al.~\cite{DLRSSY22} were the first to study LEDP algorithms that used
privacy mechanisms beyond Randomized Response by 
using the geometric mechanism on small-round parallel algorithms
for $k$-core decomposition. 
However, their results are purely theoretical
and they do not discuss how to translate their theoretical guarantees to practice. 
In particular, they give additive error bounds $\geq \frac{\log_2^3(n)}{\eps}$ (where $n$ is the number of vertices);
on a graph with $n = 10^5$ nodes and $\eps = 0.5$, this translates to an additive error of $9164$. Most real-world graphs of this size 
have max core numbers of $10^2$ magnitude, so the additive error itself leads to a $\geq 91$-multiplicative
approximation factor---much too large for any practical use.

\eat{In this paper, we introduce a novel algorithm
to obtain LEDP $k$-core decompositions, by introducing a novel degree-thresholding technique to~\cite{DLRSSY22}, with better error guarantees 
in terms of the \emph{maximum degree} instead of the total number of nodes in the graph. 
\cref{table:kcore} (in~\cref{sec:experiments}) 
shows that the max degree is often much smaller than the number of nodes/edges.
Using our $k$-core decomposition algorithm, we develop a novel LEDP edge orientation algorithm where 
the out-degree is approximately bounded by the maximum core number (also known as the degeneracy) of the graph. 
We then formulate a novel variant of the RR algorithm of Eden \etal~\cite{ELRS23} and combine it with our 
LEDP out-degree orientation algorithm. This novel
triangle counting algorithm has error proportional to the maximum core number, improving
the best previous error bounds~\cite{IMC21communication,IMC21locally} for broad classes of graphs. \cref{table:kcore}
shows that the max core number is even smaller than the maximum degree. 
Since the maximum core number is an input-dependent private property of the graph, previously unexplored 
in local graph algorithms, accurately computing it while preserving privacy presents significant 
challenges. We present a novel, intricate analysis of our triangle counting algorithm's utility.}

\eat{A major goal of our paper is to introduce additional new theoretical and implementation 
techniques that together with Randomized Response helps us obtain provably private, accurate, and computationally efficient LEDP algorithms. We propose a novel LEDP $k$-core decomposition algorithm that incorporates a degree-thresholding technique to~\cite{DLRSSY22}, improving error guarantees by bounding them in terms of the \emph{maximum degree} rather than the total number of nodes. As demonstrated in~\cref{table:kcore} (in~\cref{sec:experiments}), the maximum degree is often significantly smaller than the number of nodes or edges, leading to improved accuracy in practical settings.

Building on our $k$-core decomposition algorithm, we develop a novel LEDP edge orientation algorithm, ensuring that the out-degree is approximately bounded by the graph's maximum core number (degeneracy). We further refine this approach by formulating a variant of the Randomized Response (RR) algorithm of Eden \etal~\cite{ELRS23} and integrating it with our LEDP out-degree orientation method. This results in a new triangle counting algorithm with error proportional to the maximum core number, improving upon the best-known error bounds~\cite{IMC21communication,IMC21locally} for a broad class of graphs. As shown in~\cref{table:kcore}, the maximum core number is typically even smaller than the maximum degree, further enhancing the algorithm’s accuracy. Since the maximum core number is an input-dependent private property of the graph, previously unexplored  in local graph algorithms, accurately computing it while preserving privacy presents significant challenges. We present a novel, intricate analysis of our triangle counting algorithm's utility.



\eat{\smallskip\noindent\textbf{Evaluating LEDP in a Distributed Setting}}
To bridge the gap between theory and practice, we conduct the first evaluation of LEDP algorithms in a \eat{realistic} distributed simulation environment \eat{with actual message passing}. While the privacy model remains consistent with prior works~\cite{IMC21locally, IMC21communication, Qin17, DLRSSY22}, our evaluation diverges from previous studies, which typically tested LEDP algorithms on a single processor. Instead, we simulate a distributed execution by partitioning the graph across multiple processors, each independently running LEDP algorithms on its assigned subgraph. These processors communicate privacy-preserving outputs via message passing over multiple rounds, closely modeling real-world distributed environments. We apply this evaluation framework to $k$-core decomposition and triangle counting, demonstrating the scalability and practicality of our algorithms in distributed settings.}

\eat{\subsection{Summary of Contributions}}

A major goal of this work is to develop new theoretical and implementation techniques that, together with Randomized Response, enable provably private, accurate, and computationally efficient LEDP algorithms. We propose a novel LEDP $k$-core decomposition algorithm that incorporates a degree-thresholding technique to~\cite{DLRSSY22}, improving error guarantees by bounding them in terms of the \emph{maximum degree} rather than the total number of nodes. As demonstrated in~\cref{table:kcore} (in~\cref{sec:experiments}), the maximum degree is often significantly smaller than the number of nodes or edges, leading to improved accuracy in practical settings.  

Building on our $k$-core decomposition algorithm, we introduce a novel LEDP edge orientation technique, ensuring that the out-degree is approximately bounded by the graph's maximum core number (degeneracy). We further refine this approach by formulating a variant of the Randomized Response (RR) algorithm of Eden et al.~\cite{ELRS23} and integrating it with our LEDP out-degree orientation method. This results in a new triangle counting algorithm with error proportional to the maximum core number, improving upon the best-known error bounds~\cite{IMC21communication,IMC21locally} for a broad class of graphs. As shown in~\cref{table:kcore}, the maximum core number is typically even smaller than the maximum degree, further enhancing accuracy. Since the maximum core number is an input-dependent private property of the graph, previously unexplored in local graph algorithms, accurately computing it while preserving privacy presents significant challenges. \eat{We present a detailed analysis of our triangle counting algorithm's utility 
in~\cref{subsec:tcount-theory}.} {We present our triangle counting algorithm's utility in~\cref{subsec:tcount-theory}, and refer the reader to the full paper~\cite{githubFullPaper} for a detailed analysis.}  

To bridge the gap between theory and practice, we conduct the first evaluation of LEDP algorithms in a distributed simulation environment. While maintaining the same privacy model as prior works~\cite{IMC21locally, IMC21communication, Qin17, DLRSSY22}, our evaluation differs by simulating a distributed execution rather than relying on a single-processor implementation. We partition the graph across multiple processors, each independently running LEDP algorithms on its assigned subgraph while communicating privacy-preserving outputs via message passing over multiple rounds. This closely models real-world distributed environments. We apply this evaluation to both $k$-core decomposition and triangle counting, demonstrating the scalability and practicality of our algorithms in large-scale distributed settings.  

Our approach integrates techniques from parallel and distributed algorithms, combining algorithmic and engineering insights to enhance practicality. Our algorithms scale to significantly larger graphs than those considered in prior private algorithms, handling up to billions of edges. We implement all our algorithms in Golang~\cite{DK15}. \eat{and, unlike previous works that simulated LEDP algorithms on a single core while ignoring communication overhead, we fully account for communication costs in our distributed simulation.} In summary, we make the following contributions:  



\eat{In this paper, we present several contributions addressing the challenge of \emph{provably practical and accurate local edge differentially private (LEDP) graph algorithms}. Our approach integrates techniques from parallel and distributed algorithms, combining algorithmic and engineering insights to enhance practicality.  Our algorithms scale to graphs significantly larger than those considered in prior private algorithms, handling up to billions of edges. We implement all our algorithms in Golang~\cite{DK15} and, unlike previous works that simulated LEDP algorithms on a single core while ignoring communication overhead, we fully account for communication costs in our distributed simulation. 

In summary, we make the following contributions:}

\eat{We implement and test all of our algorithms in our framework
using Golang~\cite{DK15}. 
Our implementations simulate \emph{all} communication costs associated with our LEDP algorithms. 
Previous work simulated LEDP algorithms via a central implementation on one core
without taking into account all communication costs and 
overhead. 
Using our framework, we give state-of-the-art LEDP algorithms for the $k$-core decomposition
and triangle counting problems. Our algorithms scale to graphs with billions of edges and give up to \emph{2 
orders of magnitude} improvements
in accuracy over previous work for triangle counting, and comparable accuracy on average to the
best non-private $k$-core decomposition algorithm. 
Moreover, our LEDP distributed simulation (LEDP-DS) framework implementations
do not exceed the runtimes of central implementations by more than a $1.1$x factor, 
demonstrating that despite incurring additional communication costs, our implementations still maintain efficiency.
Finally, we are the first to scale LEDP algorithms \eat{up} \pranay{to massive graphs} with \emph{over a billion} edges; 
previous implementations tested
\pranay{on} graphs \eat{up to}\pranay{with} millions of edges~\cite{IMC21locally,IMC21communication}.
Our source code is available at~\cite{githubCode}. To summarize, we make the following contributions:


\squishlist
    \item We design a new LEDP $k$-core decomposition algorithm that does not use Randomized Response
    and has provable privacy and error bounds. Our
    algorithm is based on the previously best-known level data structure 
    LEDP $k$-core decomposition algorithm of Dhulipala et al.~\cite{DLRSSY22}, 
    \eat{but we give better theoretical error bounds.} improving on their theoretical bounds.
    Our main contribution lies in thresholding the maximum number of levels a node can pass through based on its (private) degree, rather than the total number of nodes in the graph as was done before. It is well-established that a node's core number is upper bounded by its degree; therefore, in graph classes where the maximum degree is smaller than the total number of nodes, our algorithm offers improved theoretical error guarantees.
    \item We give the \emph{first} implementation of any private $k$-core decomposition algorithm in both the central 
    and LEDP models. We implement our 
    $k$-core decomposition algorithm in our framework and compare our LEDP implementation with our central implementation.
    We demonstrate low errors in our experiments, matching the theoretical approximation bounds of the best \emph{non-private}
    $k$-core decomposition algorithms, and show that our LEDP implementation does not incur too much additional 
    communication overhead compared to our central implementation.
    \item We design a \emph{novel} triangle counting algorithm that combines our LEDP $k$-core decomposition algorithm with Randomized Response.
    Our $k$-core decomposition algorithm also gives a LEDP low out-degree ordering that minimizes the out-degree
    of any node. Using this low out-degree ordering, we can obtain better theoretical error bounds than the best-known 
    triangle counting algorithms of Imola et al.~\cite{IMC21communication,IMC21locally} and Eden et al.~\cite{ELRS23} for bounded 
    degeneracy graphs, often seen in the real world. 
    We implement our algorithm and show that it \eat{improves on the error bounds of the best previous implementation~\cite{IMC21communication}
    by up to $89$x} achieves error reductions of up to \textbf{six orders of magnitude} compared to previous implementations~\cite{IMC21communication} and improves on the running time for dense graphs.
    \item We design and implement a new LEDP distributed simulation (LEDP-DS)
    framework for simulating distributed LEDP graph 
    algorithms. Our framework combines the \emph{massively parallel computation
    model}, used to theoretically study distributed computation across many machines, with the LEDP model. 
    Our implementations use this framework, which can be used broadly to design and simulate other LEDP algorithms on one
    machine.
    \item Our LEDP implementations are the first LEDP graph algorithms that scale to graphs with billions of edges.
\squishend

}

\squishlist
    \item We design a \emph{novel} LEDP $k$-core decomposition algorithm that does not use Randomized Response and provides provable privacy and error guarantees. By leveraging the input-dependent maximum degree property of the graph, we achieve improved theoretical bounds over the LEDP $k$-core decomposition algorithm of Dhulipala et al.~\cite{DLRSSY22}. A key innovation lies in thresholding the maximum number of levels a node can move up based on its noisy (private) degree. Since a node's core number is upper bounded by its degree, our algorithm offers stronger theoretical guarantees for most real-world graphs, where the maximum degree is significantly smaller than the number of nodes.  

    \item We present the \emph{first} implementation of a private $k$-core decomposition algorithm and demonstrate through empirical evaluation that it achieves an average error of $\mathbf{3}$\textbf{x} the exact values, markedly improving upon prior approaches. Furthermore, our LEDP implementation attains error rates that closely align with the theoretical approximation bounds of the best \emph{non-private} algorithms, underscoring its practical efficiency and accuracy.  


    \item We design a \emph{novel} LEDP triangle counting algorithm that modifies our $k$-core decomposition to construct a low out-degree ordering, minimizing each node’s out-degree. Leveraging this ordering, our algorithm achieves improved theoretical error bounds over the best-known methods of Imola et al.~\cite{IMC21communication,IMC21locally} and Eden et al.~\cite{ELRS23} for bounded degeneracy graphs, common in real-world networks. Our implementation reduces relative error by up to \textbf{89x} and improves the multiplicative approximation by up to \textbf{six orders of magnitude} over the best previous implementation~\cite{IMC21communication}, while maintaining competitive runtime performance.

    \item We conduct the first evaluation of LEDP graph algorithms in a distributed simulation setting with actual message passing. Unlike prior studies that relied on a single processor, we simulate a distributed environment by partitioning the graph across multiple processors. This approach provides a more realistic assessment of both computational and communication overhead in large-scale distributed scenarios. We demonstrate the scalability and practicality of this evaluation by applying it to our $k$-core decomposition and triangle counting algorithms.

    \item We present the first LEDP graph algorithm implementations that scale to \defn{billion}-edge graphs, whereas prior implementations were tested on graphs with millions of edges~\cite{IMC21locally,IMC21communication}. Our evaluation framework serves as a valuable tool for designing and testing other LEDP algorithms. Our source code is available at~\cite{githubCode}.   
\squishend
}

\eat{\subsection{Related Work}
Local differential privacy (LDP) for graph data has been explored in recent works~\cite{Qin17,IMC21locally,IMC21communication,sun2019analyzing,ye2020towards,ye2020lf,DLRSSY22,ELRS23,hillebrand2023communication}, focusing on tasks such as synthetic graph generation and subgraph counting. Some works~\cite{sun2019analyzing} consider an ``extended'' local view, where nodes observe their neighbors' edges, for accurate triangle counting. However, this model disregards realistic privacy concerns, as users in social networks, for example, don't want to reveal their friend lists. 


The local edge differential privacy (LEDP) model was first formally introduced by the seminal works of Qin et al.~\cite{Qin17}
and Imola, Murakami, and Chaudhury~\cite{IMC21locally} 
and expanded theoretically upon by a number of subsequent works~\cite{DLRSSY22,ELRS23,IMC21communication}. 
In particular, Imola, Murakami, and Chaudhury~\cite{IMC21locally,IMC21communication} gave the first practical implementations
of triangle and subgraph counting algorithms in this model. 
A very recent work of Hillebrand \etal~\cite{hillebrand2023communication} uses 
hash functions to decrease the error of LEDP triangle counting algorithms. However, their algorithm cannot scale to large graphs.

Thus far, all previous LEDP algorithms for triangle counting rely on variants of Randomized Response. 
The work of Imola, Murakami, and Chaudhury~\cite{IMC21locally} studied triangle counting in the 
one (non-interactive) and multiple (interactive) round models where they perform Randomized Response on the graph in the first 
round and use post-processing to obtain the triangle count on the graph obtained from Randomized Response.
They bound the standard deviation 
of the additive error obtained from Randomized Response by $O\left(\frac{n^2}{\eps} + \frac{n^{3/2}}{\eps^2}\right)$.
In a subsequent work, they reduce the communication cost of the protocol~\cite{IMC21communication} 
by using a combination of sampling and clipping techniques, and refine their standard deviation analysis from their previous paper by using the number of $4$-cycles, $C_4$. 
Their new theoretical standard deviation is $O\left(\frac{\sqrt{C_4}}{\eps} + \frac{n^{3/2}}{\eps^2}\right)$ 
for the interactive setting and $O(n^2)$ 
for the non-interactive setting.

Eden et al.~\cite{ELRS23} further improve the error for triangle counting using Randomized Response
by performing a different post-processing analysis; however, their work is purely theoretical and does not have experimental
evaluations. Their new post-processing analysis gives a standard deviation
of $O\left(\frac{\sqrt{C_4}}{\eps^2} + \frac{n^{3/2}}{\eps^3}\right)$ on the additive 
error for the \emph{non-interactive} (one-round) 
setting. They also give $\Omega(n^2)$ lower bounds on the additive error for the non-interactive setting and 
$\Omega\left(\frac{n^{3/2}}{\eps}\right)$ lower bounds in the interactive setting. Thus, the upper and lower bounds are separated
by a polynomial in the interactive (multi-round) setting. 
Existing works on LDP graph algorithms have not combined Randomized Response with other privacy mechanisms to improve error bounds. Moreover, prior work often neglects the structural input-dependent properties of real-world networks in their theoretical analysis. 


Thus far, all known algorithms for LEDP $k$-core decomposition are theoretical~\cite{DLRSSY22}. 
The $k$-core decomposition algorithm of Dhulipala et al.~\cite{DLRSSY22} uses what they call a \emph{level data structure}
to estimate the core numbers. Nodes can move up the levels of the structure, with higher levels representing larger core numbers. A node moves 
up a level, based on the induced degree among the neighbors at the same or higher levels while adding 
noise drawn from the symmetric geometric distribution to preserve privacy.
\eat{Nodes can move up the levels of the structure,  higher levels contain nodes with larger
core numbers and lower levels contain nodes with smaller core numbers. A node moves 
up a level, based on the induced degree among the neighbors at the same or higher levels while adding 
noise drawn from the symmetric geometric distribution to preserve privacy. }
However, the added noise scales with the number of nodes, ignoring input-specific adaptations, leading to significant errors for large graphs. Concurrent work by Dhulipala et al.~\cite{dhulipala2024nearoptimaldifferentiallyprivatekcore} introduces an improved theoretical approach using a generalized sparse vector technique to avoid cumulative privacy budget costs. While this advancement offers better theoretical guarantees, implementing it practically in a distributed setting poses challenges. The algorithm relies on a peeling process for $k$-core decomposition, which is inherently sequential and difficult to parallelize efficiently across multiple processors. Furthermore, the practical performance of these algorithms remained unexplored.}


\subsection{Related Work}
Local differential privacy (LDP) for graph data has been extensively studied~\cite{Qin17,IMC21locally,IMC21communication,sun2019analyzing,ye2020towards,ye2020lf,DLRSSY22,ELRS23,hillebrand2023communication}, focusing on tasks such as synthetic graph generation and subgraph counting. \revision{Some works~\cite{sun2019analyzing, liu2024edge} explore an \emph{extended local view}, in which each node knows its full 2-hop neighborhood (i.e., its neighbors’ edges) to improve triangle counting accuracy. In that model, triangle counting is trivial since each node sees its entire set of incident triangles—unlike our model, where nodes see only immediate (one-hop) neighbors; hence, we require more complex algorithms since nodes cannot see their incident triangles in LEDP. Moreover, the extended view is often unrealistic, since users (e.g., in social networks) may not wish to reveal their private (potentially sensitive) friend lists to their friends.}  

The LEDP model was \eat{first formally} introduced by Qin et al.~\cite{Qin17} and Imola, Murakami, and Chaudhury~\cite{IMC21locally}, with subsequent theoretical expansions~\cite{DLRSSY22,ELRS23,IMC21communication}. Imola et al.~\cite{IMC21locally,IMC21communication} provided the first practical LEDP implementations for triangle and subgraph counting. Recently, Hillebrand et al.~\cite{hillebrand2023communication} improved LEDP triangle counting using hash functions, though their method does not scale to large graphs. All prior LEDP triangle counting algorithms rely on Randomized Response. Imola et al.~\cite{IMC21locally} introduced an LEDP triangle counting algorithm in both non-interactive (single-round) and interactive (multi-round) settings, bounding the standard deviation of the additive error by $O\left(\frac{n^2}{\eps} + \frac{n^{3/2}}{\eps^2}\right)$. In a subsequent work, they reduce the  protocol's communication cost~\cite{IMC21communication} 
by using a combination of sampling and clipping techniques, and refined their standard deviation analysis by using the number of $4$-cycles, $C_4$. 
Their new theoretical standard deviation is $O\left(\frac{\sqrt{C_4}}{\eps} + \frac{n^{3/2}}{\eps^2}\right)$ 
for the interactive setting and $O(n^2)$ 
for the non-interactive setting. Eden et al.~\cite{ELRS23} further enhanced triangle counting accuracy with an improved post-processing analysis, achieving a standard deviation of $O\left(\frac{\sqrt{C_4}}{\eps^2} + \frac{n^{3/2}}{\eps^3}\right)$ for the non-interactive setting and 
establishing lower bounds of $\Omega(n^2)$ and $\Omega\left(\frac{n^{3/2}}{\eps}\right)$ for the non-interactive and interactive settings, respectively. Despite these advancements, prior work has neither combined Randomized Response with other privacy mechanisms to improve error bounds nor accounted for input-dependent properties of graphs 
in theoretical analyses. \eat{\revision{Liu et al.~\cite{liu2024edge} study the triangle counting problem under the Edge Relationship Local Differential Privacy (Edge-RLDP) model, which allows each node (user) to report a $2$-hop Extended Local View (ELV); that is, each node knows its precise $2$-hop neighborhood (the edges between its neighbors). Triangle counting is trivial in the Edge-RLDP model since each node knows the exact number of triangles it is incident to. In contrast, our work considers the stronger Local Edge Differential Privacy (LEDP) model, where nodes can only see their immediate (one-hop) neighbors and do \emph{not} see their $2$-hop neighborhood; hence, nodes don't know the number of triangles they are incident to in LEDP, leading to more complex algorithms.}}

For LEDP $k$-core decomposition, all known algorithms remain theoretical~\cite{DLRSSY22}. The algorithm by Dhulipala et al.~\cite{DLRSSY22} uses a \emph{level data structure}, where nodes ascend levels based on their noisy induced degrees, with noise drawn from a symmetric geometric distribution to ensure privacy. However, this noise scales with the number of nodes rather than adapting to input structure, leading to significant errors in large graphs. Recent concurrent and independent work by Dhulipala et al.~\cite{dhulipala2024nearoptimaldifferentiallyprivatekcore} introduces a generalized sparse vector technique to avoid cumulative privacy budget costs, achieving improved theoretical guarantees. However, implementing this approach in a distributed setting is challenging, as it relies on a peeling algorithm that is difficult to distribute. Additionally, the practical performance of these algorithms remains unexplored.

\section{Preliminaries}

Differential privacy on graphs is defined for \emph{edge-neighboring} inputs.
Edge-neighboring inputs are two graphs which differ in exactly one edge.
Here, we consider \emph{undirected} graphs.

\begin{definition}[Edge-Neighboring~\cite{NRS07}]\label{def:edge-adjacent-graphs}
Graphs $G_1 = (V_1, E_1)$ and $G_2 = (V_2, E_2)$ are
edge-neighboring if they differ in one edge, namely,
if $V_1 = V_2$ and the size of the symmetric difference of $E_1$ and $E_2$ is 1. \footnote{The \emph{symmetric difference}
of two sets is the set of elements that are in either set, but not in their intersection.}
\end{definition}

\defn{With high probability} (\defn{\whp}) is used in this paper to mean with probability at least $1 - \frac{1}{n^c}$ for any constant
$c \geq 1$.

The local edge differential privacy (LEDP) model assumes that each node in the input
graph keeps their adjacency list private. 
The model is defined in terms of $\eps$-DP
algorithms, called \emph{$\eps$-local randomizers ($\eps$-LR)}, that are run individually by every node. The
$\eps$-LRs are guaranteed to be $\eps$-DP when the neighboring inputs are adjacency lists that 
differ in one element. Following~\cite{IMC21communication}, we assume that the curator and 
all nodes act as honest-but-curious adversaries.

\begin{definition}[$\eps$-Edge Differential Privacy~\cite{DMNS06,NRS07}]\label{def:dp}
    Algorithm $\alg(G)$, that takes as input a graph $G$ and outputs 
    some value in
    $\range(\alg)$,\footnote{$\range(\cdot)$ denotes the 
    set of all possible outputs of a function.} is \textbf{$\eps$-edge differentially
    private} ($\eps$-edge DP) if for all $\rout \subseteq \range(\alg)$
    and all edge-neighboring graphs $G$ and $G'$, 
    \begin{align*}
        \frac{1}{e^{\eps}} \leq \frac{\prob[\alg(G') \in \rout]}{\prob[\alg(G) \in \rout]} \leq e^{\eps}.
    \end{align*}
\end{definition}

\begin{definition}[Local Randomizer (LR)~\cite{DLRSSY22}]\label{def:local-randomizer}
    An \defn{$\eps$-local randomizer} $R: \adj \rightarrow \rangeout$ for node $v$ is an $\eps$-edge DP 
    algorithm that takes as input the set of its neighbors $N(v)$, represented by
    an adjacency list $\adj = (b_1, \dots, b_{|N(v)|})$. In other words, $$\frac{1}{e^{\eps}} \leq \frac{\prob\left[R(\adj') \in Y\right]}{\prob\left[R(\adj) \in Y\right]} \leq e^{\eps} $$ for all 
    $\adj$ and $\adj'$  where the symmetric difference
    is $1$ and all sets of outputs $Y \subseteq \rangeout$. The probability is taken over the
    random choices of $R$ (but not over the choice of the input). 
\end{definition}

The previous definitions of LEDP~\cite{Qin17,IMC21locally,DLRSSY22} are satisfied by the following~\cref{def:ledp}. 
\cite{DLRSSY22} gives a slightly more general and complex
definition of LEDP in terms of \emph{transcripts} but all of the 
algorithms in our paper satisfy our definition below, which is also guaranteed to satisfy their more general
transcript-based definition.

\begin{definition}[Local Edge Differential Privacy (LEDP)~\cite{DLRSSY22}]\label{def:ledp}
   Given an input graph $G = (V, E)$, for any edge $\{u, v\}$, let algorithm $\alg$
   assign $\left((R^u_{1}(\adj_u, p_1), \eps^u_1), \dots, (R^u_{r}(\adj_u, p_r), \eps^u_r)\right)$ to be
   the set of $\eps^u_i$-local randomizers called by vertex $u$ during each interactive round and
   $\left((R^v_{1}(\adj_v, p_1), \eps^v_1), \dots, (R^v_{s}(\adj_v, p_s), \eps^v_s)\right)$ be the set of $\eps^v_i$-LRs 
   called by $v$.  The private adjacency lists of $u$ and $v$ are given by $\adj_u$ and $\adj_v$, respectively, and 
   $p_i$ are the new public information released in each round.
   Algorithm $\alg$ is \defn{$\eps$-local edge differentially private (LEDP)} if 
   for every edge, $\{u, v\}$:
   \begin{align*}
        \eps^u_1 + \cdots + \eps^u_r + \eps^v_{1} + \cdots + \eps^v_s \leq \eps.   
   \end{align*} 
\end{definition}

For intuition, each LR
takes as input the private adjacency list of the node $v$ and public information 
released in previous rounds; then, it releases new public information for 
$v$ which will inform the computation of other nodes in the next round. Hence, the algorithm is 
\emph{interactive}. Each time $v$ releases,
it loses some amount of privacy indicated by $\eps^v_i$ for the $i$-th LR. 
Since edge-neighboring graphs differ in exactly one edge, to ensure the privacy of the system,
it is sufficient to ensure that the 
privacy loss of every edge
sums up to $\eps$. Thus, $\eps$-LEDP algorithms
also satisfy $\eps$-DP (proven in~\cite{DLRSSY22}).

\eat{We defer our descriptions of standard privacy tools to~\cref{appendix:privacytools}.} \eat{\qq{Need to remove this reference
and refer to the full version of the paper or supplementary materials.}}
\eat{We defer descriptions of standard privacy tools to the full version of our paper~\cite{githubCode}.}  

\subsection{Privacy Tools}
\label{appendix:privacytools}
We make use of the following privacy tools and primitives.
We define all definitions below in terms of edge-neighboring adjacency lists since 
our tools will be applied to $\eps$-local randomizers.

\begin{definition}[Global Sensitivity~\cite{DMNS06}]\label{def:global-sensitivity}
    For a function $f: \domain \rightarrow \reals^d$, 
    where $\domain$ is the domain of $f$ and 
    $d$ is the dimension of the output, the 
    \defn{$\ell_1$-{sensitivity}} of 
    $f$ is $\df = \max_{\adj, \adj'} \norm{f(\adj) 
    - f(\adj')}_1$ for all pairs of $\{\adj, \adj'\}$ of neighboring adjacency lists
    (differing in one neighbor).
\end{definition}

Our algorithms and implementations in this paper use the \emph{symmetric geometric distribution} 
defined in previous papers~\cite{BV18,CSS11,DMNS06,DNPR10,AHS21,SCRCS11}. 
The symmetric geometric distribution is also often referred to as the ``discrete Laplace distribution.'' 
Using this distribution is quite crucial in practice in order to avoid the numerical errors associated with 
real-valued outputs from continuous distributions.

\begin{definition}[Symmetric Geometric Distribution~\cite{BV18,SCRCS11}]\label{def:geom}
    The \defn{symmetric geometric distribution}, denoted $\geom(b)$, 
    with input parameter $b \in (0, 1)$, takes
    integer values $i$ where the probability mass function at $i$ is given by
    $\frac{e^b - 1}{e^b + 1} \cdot e^{-|i| \cdot b}$.
\end{definition}



We denote a random variable drawn from this distribution by $X \sim \geom(b)$.
\defn{With high probability} (\defn{\whp}) is used in this paper to mean with probability at least $1 - \frac{1}{n^c}$ for any constant
$c \geq 1$.
As with all DP algorithms, privacy
is \emph{always} guaranteed and the approximation factors are guaranteed \whp. 
We can upper bound the symmetric geometric noise \whp using the following lemma.

\begin{lemma}[~\cite{BV18,CSS11,DMNS06,DNPR10,AHS21,SCRCS11,DLRSSY22}]\label{lem:noise-whp-bound}
With probability at least $1 - \frac{1}{n^c}$ for any constant $c \geq 1$, we can upper 
bound $\dedge \sim \geom\left(x\right)$ by $|\dedge| \leq \frac{c \ln n}{x}$. 
\end{lemma}

The \defn{geometric mechanism} is defined as follows. 

\begin{definition}[Geometric Mechanism~\cite{BV18,CSS11,DMNS06,DNPR10}]\label{def:sgd-mech}
    Given any function $f: \mathcal{D} \rightarrow \integers^d$, where 
    $\mathcal{D}$ is the domain of $f$ and $\df$ is the $\ell_1$-sensitivity 
    of $f$, the geometric mechanism is defined as
    $\mech(\adj, f(\cdot), \eps) = f(\adj) + (Y_1, \dots, Y_d)$,
    where $Y_i\sim \geom(\eps/\df)$ are independent and identically distributed (i.i.d.) random variables
    drawn from $\geom(\eps/\df)$ and $\adj$ is
    a private input adjacency list.
\end{definition}

\begin{definition}[Laplace Distribution]\label{def:lapalce}
    The probability density function of the Laplace distribution on $X \in \mathbb{R}$ is $Lap(b) = 2b \cdot \exp\left(-\left({\lvert X\rvert} \cdot b\right)\right)$
    
\end{definition}

\begin{lemma}[Laplace Mechanism~\cite{DMNS06}]\label{lem:laplace}
    Given a function $f : \mathcal{G} \rightarrow R$ with sensitivity $\Delta_f$, $f\left(\mathcal{G}\right) + Lap\left(\frac{\eps}{\Delta_f}\right)$ is $\eps$-differentially private.    
\end{lemma}

\begin{lemma}[Privacy of the Geometric
    Mechanism~\cite{BV18,CSS11,DMNS06,DNPR10}]\label{lem:sgd-private}
    The geometric mechanism is $\eps$-DP. 
\end{lemma}



In addition to the Geometric Mechanism, our paper also uses Randomized Response (RR). Randomized Response (RR) when applied to 
graphs flips the bit indicating the existence of an edge in the graph. We define RR in terms of the way it is used on graphs.

\begin{definition}[Randomized Response]\label{def:rr}
    Randomized response on input graph $G = (V, E)$, represented as an upper triangular
    adjacency matrix $M$ which only contains entries $M[i, j]$ where $i < j$, 
    flips every bit $M[i, j]$ (where $i < j$)
    in the matrix with probability $\frac{1}{e^{\eps} + 1}$.
\end{definition}

It is well-known that randomized response is $\eps$-DP.

\begin{lemma}[\cite{DMNS06}]\label{lem:rr-dp}
    Randomized response is $\eps$-DP.
\end{lemma}

The \emph{composition theorem} guarantees privacy for 
the \emph{composition} of multiple
algorithms with privacy guarantees of their own. 
In particular, this theorem covers the use case where multiple LEDP
algorithms are used on the \emph{same} dataset. 

\begin{theorem}[Composition Theorem~\cite{DMNS06,DL09,DRV10}]\label{thm:composition}
    A sequence of DP algorithms, $(\alg_1, \dots, \alg_k)$, with privacy parameters $(\eps_1, \dots, \eps_k)$ form at worst an $\left(\eps_1 + \cdots + \eps_k\right)$-DP algorithm under \emph{adaptive composition} (where the adversary can adaptively select algorithms after
    seeing the output of previous algorithms).
\end{theorem}
Finally, the post-processing theorem states that the result of post-processing on the output
of an $\eps$-LEDP algorithm is $\eps$-LEDP.

\begin{theorem}[Post-Processing~\cite{DMNS06,BS16}]\label{thm:post-processing}
Let $\mech$ be an $\eps$-LEDP algorithm and $h$ be an arbitrary (randomized)
mapping from $\range(\mech)$ to an arbitrary
set. The algorithm $h \circ \mech$ is $\eps$-LEDP.\footnote{$\circ$ is notation 
for applying $h$ on the outputs
of $\mech$.}
\end{theorem}

We use implementations by the Google privacy team~\cite{githubGitHubGoogledifferentialprivacy}, which also guarantee cryptographic security.  

\subsection{Problem Definitions}
Below, we define the $k$-core decomposition, low out-degree
ordering, and triangle counting problems that we study.

\begin{figure}[!t]
    \centering
    \includegraphics[width=0.45\linewidth]{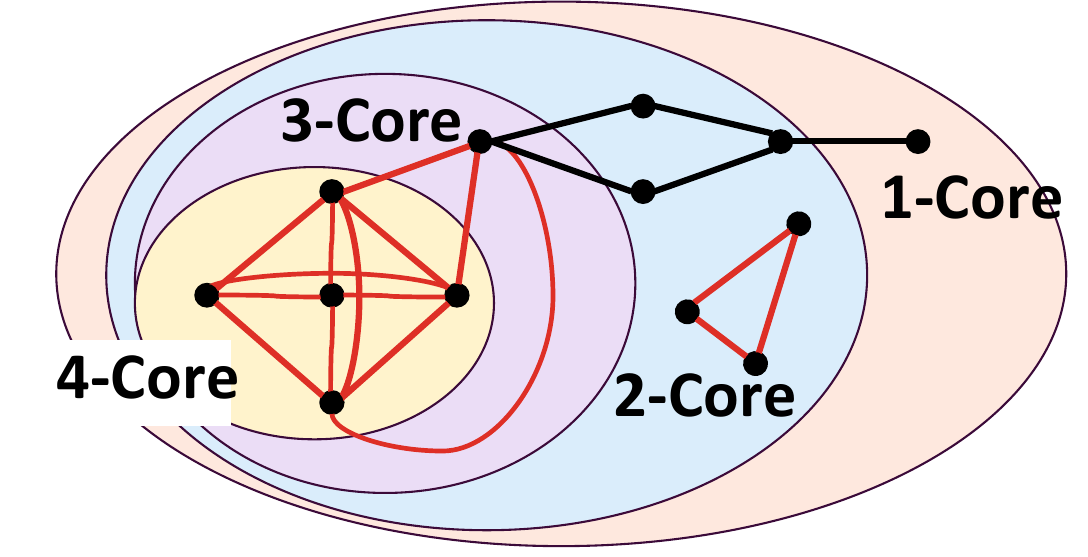}
    \caption{\footnotesize{Example $k$-core decomposition and triangles in a $4$-degenerate graph. Nodes are assigned core numbers based on the highest value core they belong
    to; e.g., a node in the $1$-core but not in the $2$-core is given the core number of $1$. Larger valued cores are contained within all smaller valued cores; e.g., the 
    $3$-core is contained in the $1$ and $2$-core. Red edges show the triangles, i.e., $3$-cycles in the graph. The 
    degeneracy of this graph is $4$.}}
    \label{fig:kcore}
\end{figure}

In this paper, we consider undirected graphs $G = (V, E)$ with $n = |V|$ nodes and $m = |E|$ edges. 
We use $[n]$ to denote $\{1, \dots, n\}$.
For ease of indexing, we set the IDs of $V$ to be $V=[n].$
The set of neighbors of a node $i \in [n]$ is denoted by $N(i)$, and the degree of node $i$
is denoted $\deg(i)$. 

\begin{definition}[$k$-Core Decomposition]\label{def:kcore}
    Given an input graph, $G = (V, E)$, a \defn{$k$-core}
    is a maximal induced subgraph in $G$ where every node has 
    degree at least $k$.
    A \defn{$k$-core decomposition} assigns a \defn{core number} to each node $v \in V$
    equal to $\kappa$ if $v$ is
   in the $\kappa$-core but not the $(\kappa+1)$-core. Let $\core(v)$ be the core number of $v$.
\end{definition}



See~\cref{fig:kcore} for an example. 
No \emph{exact} $k$-core decomposition algorithm satisfies the definition of DP (or LEDP). 
Hence, our algorithms take an input graph $G$ and output an 
\emph{approximate} core number for each node in the graph 
(\cref{def:approx-core number})
and an approximate \defn{low out-degree ordering} (\cref{def:low-outdegree}). 

\begin{definition}[$(\multfactor, \addfactor)$-Approximate Core Number~\cite{DLRSSY22}]\label{def:approx-core number}
    Let $\kest(v)$ be an approximation of the core number of $v$, and let $\multfactor \geq 1, \addfactor \geq 0$.
    The core estimate $\kest(v)$ is a \defn{$(\multfactor,
    \addfactor)$-approximate core number} of $v$ if
    $\core(v) - \addfactor \leq \kest(v) \leq \multfactor
    \cdot \core(v) + \addfactor$.
\end{definition}

We define the related concept of
an \defn{approximate low out-degree ordering} based on the definition of \defn{degeneracy}.

\begin{definition}[Degeneracy]\label{def:degeneracy}
An undirected graph $G = (V, E)$ is $\degen$-degenerate if every induced subgraph of $G$ has a node with 
degree at most $\degen$. The \emph{degeneracy} of $G$ is the smallest value of $\degen$ for which $G$ is $\degen$-degenerate.
\end{definition}

It is well known that degeneracy $\degen = \max_{v \in V}\{\core(v)\}$.

\begin{definition}[$(\multfactor, \addfactor)$-Approximate Low Out-Degree Ordering]\label{def:low-outdegree}
    Let $\order = [v_1, v_2, \dots, v_n]$ be a total ordering of nodes in a
    graph $G = (V, E)$. 
    The ordering $\order$ is an \defn{$(\multfactor, \addfactor)$-approximate
    low out-degree ordering} if
    orienting edges from earlier nodes to later nodes in $D$ 
    produces out-degree at most
    $\multfactor \cdot \degen + \addfactor$.
\end{definition}

\eat{Finally, we use our approximate low out-degree ordering to obtain a multi-round algorithm for 
triangle counting. Namely, using a publicly released approximate low out-degree ordering, each node knows which of its edges 
are oriented outwards. Using this orientation, each node can subsequently privately count the 
number of triangles formed by its outgoing edges.} 

\begin{definition}[Triangle Count]\label{def:triangle-counting}
    Given an undirected input graph $G = (V, E)$, the triangle count returns the number of $3$-cycles 
    in $G$. 
\end{definition}


\subsection{Distributed Simulation Model}
\label{sec:ledp-dist}
\revision{
Local Edge Differential Privacy (LEDP) is inherently decentralized: each user (or node) independently perturbs their private local data (adjacency list) before any communication. This model aligns naturally with distributed systems, where data is often siloed across machines or clients. To evaluate LEDP algorithms in such settings, we adopt a distributed simulation framework that closely mirrors real-world deployments. Specifically, we use a coordinator-worker model in which each worker is assigned a partition of nodes along with their full adjacency lists. Workers execute LEDP algorithms locally and communicate only privacy-preserving outputs to a central coordinator. The coordinator aggregates these responses and broadcasts public updates to all workers, proceeding iteratively over synchronous communication rounds. While assigning one machine per node is infeasible at scale, this simulation preserves the privacy and communication structure of LEDP and allows for practical evaluation of network overhead on large graphs.
}

\eat{LEDP-DS allows for the simulation on one machine of distributed LEDP systems for debugging purposes
before building the system across many machines (often expensive).} 
\eat{Thus, LEDP-DS moves us closer to a world with the best LEDP
guarantees combined with decentralized storage of 
massive graph structured data.} 
\section{Practical $k$-Core Decomposition Algorithm}
\label{sec:kcore}
\setcounter{algocf}{0} 

\revision{
We present \kcorealgo{}, a novel $k$-core decomposition algorithm that addresses key limitations of prior work~\cite{DLRSSY22} through principled algorithmic design. While their framework offers strong theoretical foundations under the $\epsilon$-LEDP model, its dependence on the total number of nodes leads to large additive error and excessive communication rounds. Our algorithm replaces this dependency with input-sensitive parameters—specifically, the graph’s maximum degree—through \emph{degree thresholding} and \emph{bias terms}. These techniques yield improved asymptotic bounds and significantly lower empirical error, as confirmed by our experiments. 
}


\eat{This section is organized as follows. In~\cref{subsec:kcore-desc}, we present the pseudocode and describe the algorithm in relation to it, focusing on the coordinator and worker roles, as well as its adaptation for distributed environments.
In~\cref{subsec:kcore-theory}, we present the theoretical guarantees of the algorithm, including privacy, approximation, communication rounds, memory, and communication costs. Due to space constraints, we relegate
some lemmas, theorems, and proofs to~\cref{appendix:appendixkcore}.} 

\SetKwProg{Fn}{Function}{}{end}\SetKwFunction{FRecurs}{FnRecursive}%
\SetKwFunction{FnUpdateLevels}{UpdatenodeLevels}
\SetKwFunction{FnInsert}{Incremental}
\SetKwFunction{FnStatic}{\kcorealgo{}}
\SetKwFunction{FnLEDPDensestSubgraph}{LEDPDensestSubgraph}
\SetKwFunction{FnDelete}{Decremental}
\SetKwFunction{FnCoreNumber}{EstimateCoreNumbers}
\SetKwFunction{FnLowOutdegree}{EstimateOrdering}
\SetKwFunction{FnCoord}{\kcorealgo{}}
\SetKwFunction{FnData}{ProcessDataCoord}
\SetKwFunction{FnWorker}{LevelMovingWorker}
\SetKwFunction{FnGraph}{DegreeThresholdWorker}

\eat{\begin{figure*}
    \centering
    \includegraphics[scale=0.75]{figures/kcore_algo.pdf}
    \caption{\footnotesize{Node's (blue) movement across levels in the coordinator's level data structure, updated each round based on the noisy neighbor counts released by worker processes. The node doesn't move up a level in round \emph{t+2}, due to thresholding, which upper-bounds the level for a node.}}
    \label{fig:kcore_algo}
\end{figure*}

\begin{figure}
    \centering
    \includegraphics[scale=0.75]{figures/kcore_algo_single_column.pdf}
    \caption{\small{Node's (blue) movement across levels in the coordinator's level data structure, updated each round based on the noisy neighbor counts released by worker processes. The node doesn't move up a level in round \emph{t+2}, due to thresholding, which upper-bounds the level for a node.}}
    \label{fig:kcore_algo}
\end{figure}}

\begin{figure*}
    \centering
    \includegraphics[scale=0.7]{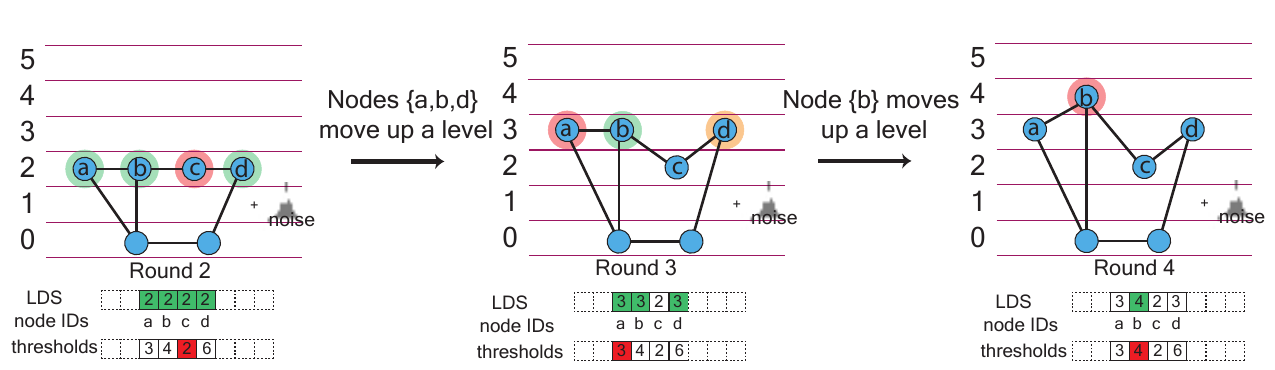}
    \caption{\revision{\footnotesize{Node movements in \kcorealgo{}'s Level Data Structure (LDS). Green: active nodes eligible to move; red: thresholded nodes; orange: active nodes that fail the noisy neighbor check. The LDS and threshold structures are shown alongside the graph. Noise is added during the level-moving step to ensure privacy, and snapshots illustrate node progression and halted movement due to thresholding.}}}
    \label{fig:kcore_algo}
\end{figure*}

\subsection{Algorithm Description}
\label{subsec:kcore-desc}


\eat{The $k$-core decomposition}\revision{Our} algorithm operates synchronously over \(O\left(\log(n) \log\left(D_{\max}\right) \right)\) rounds, where \(D_{\max}\) is the maximum degree of the graph. The algorithm outputs \(\left(2+\eta, O\left(\frac{\log(D_{\max})\log^2(n)}{\epsilon}\right)\right)\)-approximate core numbers with high probability, as well as a low out-degree ordering with the same approximation guarantee. Throughout this section, the term $\log(n)$ denotes $\log_{1 + \psi}(n)$, unless explicitly stated otherwise (where $\psi$ is a constant). The algorithm uses a level data structure (LDS)~\cite{DLRSSY22}, where nodes are assigned levels that are updated iteratively. Levels are partitioned into groups of equal size, with each group \(g_i\) containing \(\frac{\lceil \log(n) \rceil}{4}\) consecutive levels. We limit the number of rounds a node participates in, based on its noisy degree, which we refer to as degree thresholding. This significantly reduces the number of rounds, from $O(\log^2(n))$~\cite{DLRSSY22} to \(O\left(\log(n) \log\left(D_{\max}\right) \right)\). \eat{In each round, the algorithm processes nodes to determine whether they should move up a level in the LDS, based on a noisy count of its neighbors at the same level.} \revision{In each round, the algorithm uses a noisy count of a node's neighbors at the same level to decide if it should move up a level.} After processing all nodes in a round, the updated LDS is published for use in subsequent rounds. Once all rounds are complete, the algorithm estimates the core numbers of the nodes based on their final levels, using~\cref{alg:estimate}. Additionally, a low out-degree ordering is determined by sorting nodes from smaller to larger levels, breaking ties using node IDs.\footnote{This doesn't leak privacy, as IDs are assigned to nodes and not edges and reveal no information about the sensitive edge data.}. The algorithm is implemented in a distributed setting, where computation is divided between a coordinator and multiple workers. The pseudocode is structured to reflect this division. \revision{We now describe their respective roles.}\eat{The following parapgraphs describe the roles of the coordinator and the workers.}

\eat{\noindent\textbf{Coordinator} As detailed in~\cref{alg:private-kcore-coordinator}, it receives the graph size $n$, the number of workers, $M$, constant parameter $\psi > 0$, privacy parameter $\eps\in(0,1]$, constant privacy split fraction $\privatefraction \in (0, 1)$, and the bias term $b$ as input. It computes privacy parameters $\eps_1$ and $\eps_2$ for degree thresholding and noisy neighbor count based on $\eps$ and $\privatefraction$ as shown in~\cref{line:eps}. The coordinator maintains the \emph{level data structure} ($LDS$) and a communication channel, \emph{channel}, for receiving data from the workers. All nodes start at level 0 as indicated in~\cref{line:initial}; the nodes then move up levels according to bits released by the corresponding workers (discussed below). $LDS[i]$ contains the current level of node $i$. The coordinator signals the workers to load their graphs~(\cref{line:load}) in parallel, and uses the threshold values sent by the workers to compute the number of rounds~(\cref{line:rounds}). In each round $\lcur$, the coordinator calculates the group index of that round~(\cref{line:index}), and launches $M$ asynchronous worker processes~(\cref{line:kcore-launch}). Each worker sets a bit for each node in its subgraph to $1$, if the node should move up a level, or $0$ (if not) and sends this data back to the coordinator. The coordinator waits for all workers to complete (\cref{line:wait}), processes incoming data, and moves nodes based on worker-computed bits. After each round, it publishes the updated $LDS$ (\cref{line:publish}) for use in the next round. Once all rounds are complete, it estimates the core numbers for the nodes using the LDS.}

\revision{
\noindent\textbf{Coordinator.}
As described in~\cref{alg:private-kcore-coordinator}, the coordinator takes as input the graph size $n$, number of workers $M$, constant $\psi > 0$, privacy parameter $\eps \in (0,1]$, privacy split fraction $\privatefraction \in (0,1)$, and bias term $b$. It first computes the privacy budgets $\eps_1$ and $\eps_2$ for degree thresholding and noisy neighbor counts, respectively, based on $\eps$ and $\privatefraction$~(\cref{line:eps}). The coordinator maintains the \emph{level data structure} (LDS), where $LDS[i]$ stores the current level of node $i$, and a communication channel, \emph{channel}, for receiving messages from the workers. All nodes are initialized to level 0~(\cref{line:initial}) and are incrementally moved up in later rounds based on signals received from the workers. It begins by signaling the workers to load their assigned subgraphs in parallel~(\cref{line:load}) and then collects the degree threshold values to determine the total number of rounds~(\cref{line:rounds}). In each round $\lcur$, it computes the corresponding group index~(\cref{line:index}) and launches $M$ asynchronous worker processes~(\cref{line:kcore-launch}). Each worker returns a bit vector indicating whether each node in its subgraph should move up a level. After all processes complete~(\cref{line:wait}), the coordinator processes the responses and updates the LDS accordingly. It then publishes the new LDS~(\cref{line:publish}) before the next round begins. After the final round, the coordinator computes the estimated core numbers using the final LDS.

}


\eat{\noindent\textbf{Worker (Degree Thresholding)} As shown in~\cref{alg:private-kcore-worker-graph}, workers initialize their respective subgraphs, and compute the threshold for a node by adding symmetric geometric noise to the original degree \(d_v\) of the node, producing \(\widetilde{d}_v = d_v + X\), where \(X \sim \text{Geom}(\frac{\epsilon}{2})\)~(\cref{line:noise_degree}). An additional bias term is subtracted from \(\widetilde{d}_v\) to account for large negative noise; this bias term is calculated using the variance of the noise distribution. The threshold value is then calculated as \(\lceil \log_2(\widetilde{d}_v) \rceil \cdot L\), where \(L\) is the number of levels per group in the LDS~(\cref{line:deg-threshold}). Nodes store their thresholds to determine whether they should participate in a given round. Once all thresholds are computed, the worker identifies the maximum threshold value for its subgraph and sends this information back to the coordinator~(\cref{line:send_threshold}). 

\noindent\textbf{Worker (Level Moving)} In each round, a worker process determines for each node in its respective subgraph if it should move up a level in the LDS. As shown in~\cref{alg:private-kcore-worker}, for a node \(v\), if the threshold matches the current round \(\lcur\), the node is skipped~(\cref{line:level-skip}). Otherwise, if the node’s level is \(\lcur\), the algorithm calculates the number of neighbors\eat{of \(v\)} at the same level~(\cref{line:ngh-count}). This count is perturbed to ensure privacy, producing a noisy count \(\hnup_v = \nup_v + X + B\)~(\cref{line:noise-ngh}), where \(\nup_v\) is the original count, \(X\) is noise drawn from a symmetric geometric distribution with parameter \(s = \frac{\epsilon}{2 \cdot \left(v.\text{\emph{threshold}}\right)}\), and \(B\) is a bias term introduced to optimize error bounds in practice, calculated using the variance of the noise distribution. A node moves up a level if and only if \(\hnup_v > \upexp^{\gn(\lcur)}\), where \(\gn(\lcur)\) is the group index of the current level~(\cref{line:level-move}). After processing all nodes, the worker sends this information back to the coordinator~(\cref{line:send_nextLevels}).
}

\revision{
\noindent\textbf{Worker (Degree Thresholding).} As shown in~\cref{alg:private-kcore-worker-graph}, each worker begins by loading its assigned subgraph and initializing local structures. For each node $v$, it computes a \emph{noisy degree} $\widetilde{d}_v = d_v + X$, where $d_v$ is the true degree and $X \sim \text{Geom}(\frac{\epsilon}{2})$ is symmetric geometric noise~(\cref{line:noise_degree}). To mitigate large positive noise and reduce overestimation, a bias term—proportional to the noise's standard deviation—is subtracted from $\widetilde{d}_v$. The worker then computes a threshold value for each node as $\left\lceil \log_2(\widetilde{d}_v) \right\rceil \cdot L$, where $L$ is the number of levels per group in the LDS~(\cref{line:deg-threshold}). These thresholds determine the maximum number of rounds in which each node participates. The worker keeps track of the maximum threshold across its subgraph and returns it to the coordinator~(\cref{line:send_threshold}).

\noindent\textbf{Worker (Level Moving).} In each round, workers assess whether nodes in their subgraph should move up a level. As shown in~\cref{alg:private-kcore-worker}, if a node $v$ has already reached its threshold round $\lcur$, it is skipped~(\cref{line:level-skip}). Otherwise, if $v$ is currently at level $\lcur$, we count the number of its neighbors that are also at the same level~(\cref{line:ngh-count}). To ensure privacy, this count $\nup_v$ is perturbed to produce a noisy estimate $\hnup_v = \nup_v + X + B$~(\cref{line:noise-ngh}), where $X$ is symmetric geometric noise with parameter $s = \frac{\eps}{2 \cdot v.\textit{threshold}}$ and $B$ is an added bias term to counteract large negative noise. The node moves up a level if $\hnup_v > \upexp^{\gn(\lcur)}$, where $\gn(\lcur)$ is the group index corresponding to the current level~(\cref{line:level-move}). After processing all nodes, the worker sends the updated level-change bits to the coordinator~(\cref{line:send_nextLevels}).
}

\myparagraph{Bias Terms}
\eat{We introduce two different bias terms, 
one in the degree thresholding procedure,
and one in the level moving procedure. \revision{These bias terms are not manually tuned, but analytically derived based on the standard deviation of the symmetric geometric distribution.} \eat{For the bias terms, we use the approximate standard deviation of the symmetric geometric distribution. }} 
\revision{We introduce two analytically derived bias terms, based on the standard deviation of the symmetric geometric distribution—one for degree thresholding and one for level movement.} The first bias term is subtracted from the 
computed threshold to account for situations where a large positive noise
is chosen. If a large positive noise is chosen, we lose privacy proportion to the new threshold in the level moving 
step. Hence, our bias term biases the result to smaller thresholds, resulting in less privacy loss.
\eat{in the level moving step.}

The second bias term is added to the computed induced noisy degree to account for situations where a large negative 
noise prevents nodes from moving up the first few levels of the structure. Our added bias allows nodes with non-zero degrees to move up the first levels of the structure. Since the degree bounds
increase exponentially, this additional bias term accounts for smaller errors as nodes move up levels. \eat{Our bias terms are derived from theoretical analysis optimizing the trade-off between utility, privacy, and computational cost, rather than being manually tuned hyperparameters.} \revision{ Notably, we observe this behavior in our experiments when comparing to the baseline implementation of~\cite{DLRSSY22}, which omits the bias term: many nodes remain stuck at level 0, resulting in significantly higher error. This highlights the practical importance of our bias correction.
}

\revision{
\begin{example}
\cref{fig:kcore_algo} illustrates node movement in the LDS during \kcorealgo{}. In each round, nodes compute a noisy count of same-level neighbors and move up if it exceeds a threshold based on group index—unless blocked by their degree threshold. In the LDS, green marks eligible nodes; red in the threshold array marks those blocked; and in the graph view, green means movement and red/orange means restriction. For instance, node $c$ is blocked in Round 1, $a$ in Round 2, and $b$ in Round 3. In Round 2, node $d$ is not thresholded but remains at the same level due to failing the noisy neighbor check (\cref{alg:private-kcore-worker}~\cref{line:move-up-condition}).
\end{example}
}


\begin{algorithm}[!ht]
\small
    \textbf{Input:} graph size $n$; number of workers $M$; approx constant $\psi \in \left(0,1\right)$; 
    privacy parameter $\eps \in (0,1]$; split fraction $\privatefraction \in (0,1)$; bias term $b$. \\
    \textbf{Output:} Approximate core numbers and low out-degree ordering of each node in $G$.\\
    \Fn{\FnCoord{$n, \psi, \eps, \privatefraction, b$}} {
           Set $\lambda = \frac{(5-2\eta)\eta}{(\eta+5)^2}; L = \frac{\left\lceil{\log n}\right\rceil}{4}$\\ 
           Set $\eps_1 = \privatefraction\cdot \eps$ and $\eps_2 = \left(1 - \privatefraction\right) \cdot \eps$ \label{line:eps}\\
           Set $\mathcal{C} \leftarrow$ new $\text{Coordinator}\left(LDS, channel\right)$ \\
           Coordinator initializes $\mathcal{C}.LDS$ with $\mathcal{C}.LDS[i] \leftarrow 0 \ \forall i \in [n]$. \label{line:initial}\\
           Set \emph{maxDegreeThresholds} $\leftarrow \left[\right]$ \\
           \ParFor{$w = 1$ to $M$}{
                \emph{maxDegreeThresholds}$\left[w\right] = \FnGraph(w, \eps_1, L, b)$ \label{line:load}\\
           }
           Set \emph{numOfRounds}$=\min\left(4\log(n)\log(\widetilde{d}_{max})-1, \max \left(\text{\emph{maxDegreeThresholds}}\right)\right)$ \label{line:rounds}\\
           \For{$\lcur = 0$ to numOfRounds}{ 
                Set $\gn(\lcur) \leftarrow \left \lfloor \frac{r}{L} \right\rfloor$ \label{line:index}\\
                \ParFor{$w = 1$ to $M$}{
                    $\FnWorker\left(w, \lcur, \eps_2, \psi, \gn(\lcur), \mathcal{C}.LDS\right)$ \label{line:kcore-launch}\\
                }
                $\mathcal{C}.\text{wait}()$ \Comment{coordinator waits for workers to finish} \label{line:wait}\\
                \emph{nextLevels} $\leftarrow \mathcal{C}.channel$ \\
                \For{$i = 1$ to $n$}{
                    \If{nextLevels$\left[i\right] = 1$}{
                        $\mathcal{C}.LDS.\text{levelIncrease}\left(i\right)$ \\
                    }
                }
                Coordinator publishes updated $\mathcal{C}.LDS$ \label{line:publish}\\
           }
           Coordinator calls $cores \leftarrow \mathcal{C}.\FnCoreNumber(\mathcal{C}.LDS, L, \lambda, \psi)$ \\
           Coordinator produces $D$, a total order of all nodes, using levels from $\mathcal{C}.LDS$ (from smaller to larger) breaking ties by node ID\\
            \return $(cores, D)$
    }    
    \small\caption{\small\label{alg:private-kcore-coordinator} $k$-Core Decomposition and Ordering (Coordinator)}
\end{algorithm}


\begin{algorithm}[!t]
\small
    \textbf{Input:} worker ID $w$; privacy parameter $\eps \in (0, 1]$; levels per group $L$; bias term $b$.\\
    \Fn{\FnGraph{$w, \eps, L, b$}} {
        Set \emph{maxThreshold}$\leftarrow 0$ \\
        \For{node $~v :=$ localGraph}{
            Sample $X \sim \geom\left(\frac{\eps}{2}\right)$\\
            $\widetilde{d}_v \leftarrow d_v + X$ \label{line:noise_degree} \Comment{noised degree} \\
            $\widetilde{d}_v \leftarrow \widetilde{d}_v  + 1 - \min\left(b \cdot \frac{2\cdot e^{\eps}}{e^{2\eps} - 1}, \widetilde{d}_v\right)$  \label{line:bias1}\\
            $v$.\emph{threshold} $\leftarrow \left \lceil{\log_2 (\widetilde{d}_v)} \right \rceil \cdot L$  \Comment{thresholding} \label{line:deg-threshold}\\
            $v$.\emph{permZero} $\leftarrow 1$ \\
            \emph{maxThreshold}$=\max \left(\text{\emph{maxThreshold}}, v.\text{\emph{threshold}}\right)$ \\
        }
        $w.\text{send}\left(\text{\emph{maxThreshold}}\right)$ \label{line:send_threshold}\\
    }
    \small\caption{\small\label{alg:private-kcore-worker-graph} Degree Thresholding (Worker)}
\end{algorithm}

\begin{algorithm}[!ht]
\small
    \textbf{Input:} worker id $w$; round number $\lcur$; privacy parameter $\eps \in (0, 1]$; constant $\psi$; group index $\gn(\lcur)$; pointer to the coordinator $LDS$. \\
    \Fn{\FnWorker{$w, \lcur, \eps, \psi, \gn(\lcur), LDS$}}{
        Set \emph{nextLevels} $\leftarrow [0,\ldots,0]$ \\
        \For{node $~v :=$ localGraph}{
            \If{v.threshold $=r$}{ \label{line:level-skip}
                $v$.\emph{permZero} = $0$ \\
            }

            \emph{vLevel} $\leftarrow LDS.\text{getLevel}\left(v\right)$ \\
            \If{vLevel $= r$ and v.permZero $\neq 0$}{
                Let $\nup_v$ be the number of neighbors $j \in \adj_v$ where $LDS.\text{getLevel}\left(j\right) = \lcur$. \label{line:ngh-count}\\
                Set scale $s \leftarrow \frac{\eps}{2 \cdot \left(v.\text{threshold}\right)}$ \label{line:scaling-level-up}\\
                Sample $X \sim \geom(s)$.\\
                Set extra bias $B \leftarrow \frac{6e^s}{\left(e^{2s} - 1\right)^3}$ \label{line:bias2}\\
                Compute $\hnup_v \leftarrow \nup_v + X + B$. \label{line:noise-ngh}\\
                \If{$\hnup_i > \upexp^{\gn(\lcur)}$}{ \label{line:move-up-condition}
                    \emph{nextLevels}$\left[v\right] = 1$ \label{line:level-move}\\
                }\Else{\label{line:stay-same}
                    $v$.\emph{permZero} $=0$ \\
                }
            }
            
        }
        $w.\text{send}\left(w, \text{\emph{nextLevels}}\right)$ \label{line:send_nextLevels}\\
        $w.\text{done}\left(\right)$ \label{line:done}
    }
    \small\caption{\small\label{alg:private-kcore-worker} Level Moving (Worker)}
\end{algorithm}

\begin{algorithm}[!ht]
\SetKwProg{parfor}{parfor}{:}{}
\scriptsize
    \Fn{\FnCoreNumber{$LDS, L, \lambda, \eta$}}{
        \For{$i = 1$ to $n$}{
            $\kest(i)\leftarrow\upexpold^{\max\left(
            \left\lfloor\frac{LDS[i] + 1}{L}\right\rfloor
            -1, 0\right)}$.
        }
        \return $\{(i, \kest(i)) : i \in [n]\}$
	}
    \small\caption{\scriptsize\label{alg:estimate} Estimate Core Number (Coordinator)~\cite{LSYDS22}}
\end{algorithm}

\subsection{Theoretical Analysis}
\label{subsec:kcore-theory}

\paragraph{Memory Analysis \& Communication Cost}  
Let $M$ be the number of workers and $n$ the graph size. Each worker processes $S$ nodes, where $S = \lfloor n/M \rfloor$ for $M-1$ workers, and the last worker handles $n - (M-1) \lfloor n/M \rfloor$. The coordinator maintains the level data structure (LDS) and a communication channel, both requiring $O(n)$ space, resulting in a total memory usage of $O(n)$. Each worker processes $O(S)$ nodes, requiring $O(Sn)$ space for the graph and an additional $O(S)$ space for auxiliary structures, leading to a total of $O(Sn)$. In terms of communication, workers send one bit per node per round, incurring a per-worker cost of $O(S)$ and an overall round cost of $O(n)$. The coordinator receives and distributes the updated LDS, adding another $O(n)$ cost. Thus, the total communication overhead for the algorithm is \(O\left(n \log(n) \log\left(D_{\max}\right) \right)\).  

\paragraph{Privacy Guarantees}
Our privacy guarantees depend on the following procedures.
First, we perform degree-based thresholding, which upper bounds the number of levels a node can move up.
Second, we subtract and add bias terms to the results of our mechanisms. And finally, 
we scale our noise added in~\cref{line:scaling-level-up} of~\cref{alg:private-kcore-worker} by the noisy threshold. \eat{\qq{check that all pointers to 
lines in algorithms are pointing to Algorithm 3.1}}
We show that our algorithm can be implemented using 
local randomizers (\cref{def:local-randomizer}).
Then, we show that the local randomizers have appropriate privacy parameters to satisfy $\eps$-LEDP (\cref{def:ledp}). 

\begin{restatable}[Degree Threshold \revision{LR}]{lemma}{degreethresholdinglr}\label{lem:lr-degree}
    Our degree thresholding procedure run with privacy parameter $\eps'$ is a $(\eps'/2)$-local randomizer.
\end{restatable}
\begin{proof}
    Our degree-thresholding procedure upper bounds the number of levels that we iterate through using the 
    (private) degree of each node. Specifically, it adds symmetric geometric noise to the 
    degree $\widetilde{d}_u = d_u + \geom(\eps'/2)$ and then computes $\ceil{\log_{1+\eta}(\widetilde{d}_u)} \cdot L$, where
    $L$ is the number of levels per group. The sensitivity of the degree of any node is $1$ and by the privacy
    of the geometric mechanism (\cref{lem:sgd-private}\eat{\cite{BV18,CSS11,DMNS06,DNPR10}}), the output $\widetilde{d}_u$ is $(\eps'/2)$-DP. Then, 
    producing the final level upper bound uses post-processing (\cref{thm:post-processing}\eat{\cite{DMNS06,BS16}}) where privacy is 
    preserved. Hence, our output \eat{of the degree-thresholding procedure} is $(\eps'/2)$-DP and the algorithm can 
    be implemented as a $(\eps'/2)$-local randomizer.
\end{proof}


Using~\cref{lem:lr-degree}, we prove~\cref{thm:k-core}.

\begin{theorem}\label{thm:k-core}
    \cref{alg:private-kcore-coordinator} is $\eps$-LEDP.
\end{theorem}
\begin{proof}
    Our algorithm calls the local randomizers in~\cref{lem:lr-degree} with $\eps_1 = \eps \cdot \privatefraction$, 
    where $\privatefraction \in (0, 1)$ is
    a fraction which splits some portion of the privacy budget, 
    and then iterates through the 
    levels one-by-one while adding noise to the induced degree of each node consisting 
    of all neighbors of the node on the same or higher level. 
    We showed in~\cref{lem:lr-degree} that the degree thresholding procedure can be implemented as $(\eps_1/2)$-local randomizers.
    
    The key to our better error bounds is our upper bound on the number of levels
    we iterate through, bounded by our threshold. Since the thresholds are public outputs from the local randomizers,
    we can condition on these outputs. 
    Let the threshold picked for node $v$ be denoted 
    as $t_v$. Then, we add symmetric geometric noise to the induced degree of the node (among the neighbors at or above $v$'s current level)
    drawn from $\geom\left(\eps_2/(2 \cdot t_v)\right)$ where $\eps_2 = \eps \cdot \left(1-\privatefraction \right )$.
    Conditioning on the public levels of each node, the sensitivity of the induced degree of any node is $1$.
    By the privacy of the geometric mechanism, we obtain a 
    $\left(\eps_2/(2\cdot t_v)\right)$-local 
    randomizer for $v$. By composition (\cref{thm:composition}\eat{\cite{DMNS06,DL09,DRV10}}) over at most $t_v$ levels, the set of all local randomizers called on $v$, 
    is $(\eps_2/2)$-differentially private. For any edge, the sum of the privacy parameters of the set of all 
    local randomizers called on the endpoints of the edge is $2 \cdot \eps_1/2 + 2 \cdot \eps_2/2 = \eps_1 + \eps_2 = \privatefraction \cdot \eps + 
    (1-\privatefraction) \cdot \eps$.
    By~\cref{def:ledp}, this is $\eps$-LEDP.
    
    Finally, our bias terms, added or subtracted after applying the geometric mechanism, preserve privacy due to the post-processing invariance of differential privacy~(\cref{thm:post-processing}\eat{\cite{DMNS06,BS16}}). 
\end{proof}

\paragraph{Approximation Guarantees} 
Our algorithm given in the previous section
contains several changes that results in better theoretical bounds and  optimizes the practical performance on real-world datasets. 
To prove our approximation factors, we first
show~\cref{inv:degree-upper} and~\cref{inv:degree-lower} hold for our modified algorithm.
$D_{\max}$ is the graph's max degree.
\begin{restatable}[Degree Upper Bound~\cite{DLRSSY22}]{invariant}{invariantone}
\label{inv:degree-upper}
If node $i \in V_{\lcur}$ (where $V_\lcur$ contains nodes in level $\lcur$) and $\lcur < 4\log^2 n - 1$, then $i$
has at most $\upexp^{\lfloor r/(\numgrouplevels)\rfloor} + \frac{c\log(D_{\max})\log^2(n)}{\eps}$
neighbors in levels $\geq r$, with high probability, for constant $c > 0$.
    
\end{restatable}
\begin{restatable}[Degree Lower Bound~\cite{DLRSSY22}]{invariant}{invarianttwo}
  \label{inv:degree-lower}
    If node $i \in V_{\lcur}$  (where $V_\lcur$ contains nodes in level $\lcur$) and $\lcur > 0$, then
    $i$ has at least $(1 + \lf)^{\lfloor (r-1)/(\numgrouplevels)\rfloor} - \frac{c\log(D_{\max})\log^2(n)}{\eps}$ neighbors in levels $\geq r - 1$,
    with high probability, for constant $c > 0$.  
\end{restatable}

There are two parts to our analysis: first, we prove that our degree thresholding procedure does not keep the node at too low of a level, with high probability; 
second, we show that our new procedure for moving up levels adds at most the noise used in~\cite{DLRSSY22} and not more. 
Together, these two arguments maintain the invariants. 


\begin{restatable}{lemma}{degreethresholding}
\label{lem:degree-thresholding-approx}
Degree-thresholding satisfies~\cref{inv:degree-upper}. 
\end{restatable}
\begin{proof}
\label{proof:degreethresholding}
    By~\cref{lem:noise-whp-bound}, the noise we obtain for thresholding is upper bounded by $\frac{c'\ln n}{(\eps_1/2)} = \frac{2c'\ln n}{\privatefraction \cdot \eps}$ with probability at least $1 - \frac{1}{n^{c'}}$. 
    Thus, the noisy degree $\tilde{d_v}$ we obtain in~\cref{line:noise_degree} of~\cref{alg:private-kcore-worker-graph}
    follows $\tilde{d_v} \geq d_v - \frac{2c'\ln n}{\privatefraction \cdot \eps}$ with probability at least $1 - \frac{1}{n^{c'}}$.
    Let $r$ be the level we output as the threshold level.
    Then, we can compute the upper degree bound of this level to be at least
    $(1+\lf)^{\floor{r/(2\log n)}} = (1+\lf)^{\log_{1+\lf}(\tilde{d_v})} 
    \geq (1+\lf)^{\log_{1+\lf}\left(d_v - \left(\frac{2c'\ln n}{\privatefraction\eps}\right)\right)} = 
    d_v - \frac{2c'\ln n}{\privatefraction\eps}$. The maximum induced degree of $v$ on any level is at most $d_v$.
    Let $d_{v, r}$ be the induced degree of $v$ on the thresholded level $r$, it must hold that $d_{v, r} \leq \left(d_v - \frac{2c'\ln n}{\privatefraction\eps}\right) +
    \frac{4c'\ln n}{\privatefraction\eps} + \frac{2c_3\ln n}{\privatefraction \eps} \leq (1+\lf)^{\floor{r/(2\log n)}} + \frac{(4c'+2c_3)\ln n}{\privatefraction\eps}$. 
    Since $\privatefraction$ and $c'$ are both constants
    and~\cref{inv:degree-upper} allows for picking a large enough constant $c > 0$, we can pick $c \geq 2(c'+c_3)/\privatefraction$ and~\cref{inv:degree-upper} is
    satisfied where $c'$ and $c_3$ are fixed constants $\geq 1$.
\end{proof}

We do not have to prove that our thresholded level satisfies~\cref{inv:degree-lower} since we use the threshold level as an 
\emph{upper bound} of the maximum level that a node can be on. Hence, a node will not reach that level unless
the procedure for moving the node up the levels satisfies~\cref{inv:degree-lower}. We now prove that 
our level movement procedure satisfies both invariants.

\eat{\qq{below proof can be moved to the appendix}}

\begin{restatable}{lemma}{levelmoving}
\label{lem:level-moving-inv}
Our level moving algorithm satisfies~\cref{inv:degree-upper} and~\cref{inv:degree-lower}.
\end{restatable}
\begin{proof}
\label{proof:levelmoving}
    Our level moving algorithm is similar to~\cite{DLRSSY22} except that we pick noise based on the threshold. 
    Thus, by~\cref{lem:noise-whp-bound}, our algorithm picks noise that is at most $\frac{2c_1 \cdot t_v\ln n}{\eps}$ with probability 
    at least $1 - \frac{1}{n^{c_1}}$ where $t_v$ is the released threshold for $v$ and $c_1 \geq 1$ is a fixed constant. 
    Also, by~\cref{lem:noise-whp-bound}, 
    it holds that $t_v \leq \log\left(D_{\max} + \frac{2c_2\ln n}{\privatefraction\eps} - \frac{2c_3\ln n}{\privatefraction\eps}\right)\cdot L 
    \leq 2\log\left(D_{\max}\right)\log n$ 
    with probability at least $1 - \frac{1}{n^{c_2}}$. 
    A node moves up a level from level $r$ if its induced degree plus the noise exceeds the threshold $(1 + \lf)^{\floor{r/(2\log n)}}$.
    Thus, using our computed noise, the degree of a node must be at least $(1 + \lf)^{\floor{r/(2\log n)}} 
    - 2\log\left(D_{\max}\right)\left(\frac{2c_1\log^2 n}{\eps}\right)$ 
    with probability at least $1 - \frac{1}{n^{c_1}} - \frac{1}{n^{c_2}}$ when it moves up from
    level $r$. By choosing large enough constants $c_1, c_2 > 0$ and $c \geq 4c_1 c_2$, we satisfy~\cref{inv:degree-lower}.
    Similarly, the node does not move up from level $r$ when its induced degree plus noise is at most $(1 + \lf)^{\floor{r/(2\log n)}}$.
    By a symmetric argument to the above, we show that~\cref{inv:degree-upper} is satisfied.
\end{proof}

Finally, using~\cref{lem:degree-thresholding-approx} and~\cref{lem:level-moving-inv}, we can prove the final approximation factor
for our algorithm.

\eat{\qq{below proof can be moved to the appendix}}

\begin{restatable}{theorem}{kcoreapprox}
\label{thm:kcore-approx}
Our algorithm returns $(2+\eta, O(\log(D_{\max})\log^2(n)/\eps))$-approximate $k$-core numbers, with high probability, in \(O\left(\log(n) \log\left(D_{\max}\right) \right)\) rounds of communication.
\end{restatable}
\begin{proof}
\label{proof:kcoreapprox}
    By~\cref{lem:degree-thresholding-approx} and~\cref{lem:level-moving-inv}, our algorithm satisfies~\cref{inv:degree-upper}
    and~\cref{inv:degree-lower}, hence, our algorithm returns our desired approximation
    using Theorem 4.1 of~\cite{DLRSSY22}. Our algorithm iterates to value at most $O(\log n \log \left(D_{\max}\right))$ for $D_{\max}$ (the maximum threshold), resulting 
    in a total of \(O\left(\log n \log\left(D_{\max}\right) \right)\) rounds. 
\end{proof}

\section{Triangle Counting using Low Out-Degree Ordering}
\label{sec:tcount}
\setcounter{algocf}{0} 

We present our novel triangle counting algorithm, \tcountalgo{}, which leverages the low out-degree ordering obtained from the \kcorealgo{} algorithm along with randomized response (RR). \revision{Prior works rely on either all or a sample of neighboring edges after applying RR, and often suffer from error bounds that scale poorly with graph size. To address these limitations, our algorithm exploits a new input-dependent graph property, the degeneracy (maximum core number), to upper bound the number of oriented 4-cycles~(\cref{fig:oriented-c4}), yielding  significantly tighter error bounds in theory and practice.} \eat{While previous approaches rely on all neighboring edges after applying RR or on sampling subsets of edges, these methods often suffer from error bounds that scale poorly with the graph size. \eat{To address these limitations, our algorithm introduces a new technique by utilizing the upper bound on the number of oriented 4-cycles when counting triangles for a node.} To address these limitations, our algorithm leverages a key observation: the number of oriented 4-cycles~(\cref{fig:oriented-c4}) for a node is upper bounded by \(n^2 d^2\), where \(d\) is the maximum core number (degeneracy). By utilizing this bound, \eat{to count the number of triangles for a node,} the algorithm exploits a new input-dependent graph property, the degeneracy, to significantly improve error bounds in theory and practice.\eat{This approach enhances both the theoretical robustness and practical applicability of the algorithm for large-scale graphs.} \eat{In~\cref{sec:experiments}, we empirically evaluate the algorithm’s performance, comparing it to the baseline methods of Imola et al.\cite{IMC21locally} and Eden et al.\cite{ELRS23}. Our experimental results demonstrate significant accuracy improvements, achieving up to $\mathbf{91}$\textbf{x} improvement in terms of the relative error and \textbf{six orders of magnitude} better multiplicative approximation bounds compared to previous implementations.}}

\eat{This section is organized as follows. In~\cref{subsec:tcount-desc}, we present the pseudocode and describe the algorithm in relation to it, focusing on the coordinator and worker roles, as well as its adaptation for distributed environments.
In~\cref{subsec:tcount-theory} we analyze the expectation, variance, and runtime of our triangle counting algorithm. 
In particular, we prove our privacy, approximation, number of rounds of communication, memory, and communication cost guarantees. Due to space constraints, we relegate some lemmas,
theorems, and proofs to~\cref{appendix:tcount}.}

\SetKwProg{Fn}{Function}{}{end}\SetKwFunction{FRecurs}{FnRecursive}%
\SetKwFunction{FnUpdateLevels}{UpdatenodeLevels}
\SetKwFunction{FnInsert}{Incremental}
\SetKwFunction{FnStatic}{LEDPCoreDecomp}
\SetKwFunction{FnLEDPDensestSubgraph}{LEDPDensestSubgraph}
\SetKwFunction{FnDelete}{Decremental}
\SetKwFunction{FnCoreNumber}{EstimateCoreNumbers}
\SetKwFunction{FnLowOutdegree}{EstimateOrdering}
\SetKwFunction{FnTCoord}{\tcountalgo{}}
\SetKwFunction{FnTData}{Agregate}
\SetKwFunction{FnTWorkerRR}{RRWorker}
\SetKwFunction{FnTWorkerMaxOutdegree}{MaxOutDegreeWorker}
\SetKwFunction{FnTWorkerT}{CountTrianglesWorker}
\SetKwFunction{FnGraph}{LoadGraphWorker}
\SetKwFunction{FnCoord}{$k$-CoreD}

\begin{algorithm}[!ht]
\small
    \textbf{Input:} graph size $n$; number of workers $M$; constant $\psi \in \left(0, 1\right)$; privacy parameter $\eps \in (0, 1]$; 
     split fraction $\privatefraction \in (0,1)$; bias term $b$; \\
    \textbf{Output:} Noisy Triangle Count\\
    \Fn{\FnTCoord{$n, M, \psi, \eps, \privatefraction, b$}}{
        Set $Z \leftarrow \FnCoord(n, \psi, \frac{\eps}{4}, \privatefraction, b)$ (\cref{alg:private-kcore-coordinator})  \\
        Set $\mathcal{C}\leftarrow$ new $\text{Coordinator}\left(cRR, cTCount, cMaxOut, X\right)$ \eat{\qq{where are cTCount, cMaxOut, cRR, X initialized?}} \label{line:initialize-triangle-coordinator}\\
        \Comment{Round 1: Randomized Response} \label{line:round1-start} \\
        \ParFor{$w = 1$ to $M$}{ 
            $\FnTWorkerRR\left(w, n, \frac{\eps}{4}\right)$ \\
        }
        $\mathcal{C}.\text{wait}()$ \Comment{coordinator waits for workers to finish} \\
        $\mathcal{C}.\text{publishNoisyEdges}(cRR, X)$  \label{line:publishRR}\\
        \Comment{Round 2: Max Out-degree} \label{line:round2-start}\\
        \ParFor{$w = 1$ to $M$}{ 
            $\FnTWorkerMaxOutdegree\left(w, \frac{\eps}{4}, Z\right)$ \\
        }
        $\mathcal{C}.\text{wait}()$ \\
        $\widetilde{d}_{\max} \leftarrow \max\left(\{C.cMaxOut[i] \mid i \in [M]\}\right) + \frac{12\log(n)}{\eps}$ \label{line:max-D}\\
        \Comment{Round 3: Count Triangles} \label{line:round3-start}\\
        \ParFor{$w = 1$ to $M$}{ 
            $\FnTWorkerT\left(w, \frac{\eps}{4}, Z, X, \widetilde{d}_{\max} \right)$ \\
        }
        $\mathcal{C}.\text{wait}()$ \\
        $\widetilde{\Delta} \leftarrow \sum_{i=1}^{n}C.cTCount\left[i\right]$ \label{line:tcount}\\
        \return $\widetilde{\Delta}$
    }
    \small\caption{\small\label{alg:private-tcount-coordinator} $\eps$-LEDP Triangle Counting (Coordinator)}
\end{algorithm}

\begin{algorithm}[!ht]
\small
    \textbf{Input:} worker id $w$; graph size $n$; privacy parameter $\eps \in (0, 1]$;\\
    \Fn{\FnTWorkerRR{$w, n, \eps$}}{
        Set \emph{neighborsRR} $\leftarrow [][]$\\
        \For{node $~v :=$ localGraph}{
            For all $ngh_{v} \leftarrow \left\{j : j \in \left[n\right] \wedge j > v \right\}$  \label{line:upper_triang}\\
            \emph{neighborsRR}$\left[v\right] = \text{RandomizedResponse}_\eps(ngh_{v})$ \eat{(\cref{def:rr}) \qq{remove the ref for the submission and only include it for the full version}} \label{line:rr}\\ 
        }
        $w.\text{sendRR}(w, \text{\emph{neighborsRR}})$ \label{line:send_rr}\\
        $w.\text{done}\left(\right)$ 
    }
    \small\caption{\small\label{alg:private-tcount-workerRR} Randomized Response (Worker)}
\end{algorithm}

\begin{algorithm}[!ht]
\small
    \textbf{Input:} worker id $w$; graph size $n$; privacy parameter $\eps \in (0, 1]$;\\
    \Fn{\FnTWorkerMaxOutdegree{$w, \eps, Z$}}{
        Set $out_{\max} \leftarrow 0$.\\
        \For{node, adjacency list $~v, ngh :=$ localGraph}{
            Set $ngh_{v} \leftarrow \left\{j : j \in ngh \wedge Z[j] > Z[v] \right\}$  \label{line:outgoing} \\
            $B \sim \geom(\eps)$\\
            $out_{\max} \leftarrow \max(out_{\max}, |ngh_{v}| + B)$ \label{line:noise-out-degree}\\            
        }
        $w.\text{sendMaxOutdegree}(out_{\max})$ \label{line:send_out}\\
        $w.\text{done}\left(\right)$ 
    }
    \small\caption{\small\label{alg:private-maxOutdegree}  Noisy Max Out-Degree (Worker)}
\end{algorithm}

\begin{algorithm}[!ht]
\small
    \textbf{Input:} worker id $w$; privacy parameter $\eps \in (0, 1]$; low out-degree ordering $Z$, published noisy edges $X$; public maximum noisy out-degree $\widetilde{d}_{\max}$;\\
    \Fn{\FnTWorkerT{$w, \eps, Z, X, \widetilde{d}_{\max}$}}{
        Set \emph{workerTCount} $\leftarrow 0.0$\\
        \For{node, adjacency list $~v, ngh :=$ localGraph}{
            Set $\widetilde{\Delta} \leftarrow 0.0$ \\
            $\text{OutEdges}_v= \left\{j : j \in ngh \wedge Z\left[j\right] > Z\left[v\right]\right\}$ \\
            \For{$i_1 \in \{1, \dots, \min(\widetilde{d}_{\max}, |\text{OutEdges}_v|)\}$}{ \label{line:out-degree-bound}
                \For{$i_2 \in \{i_1 + 1, \dots, \min(\widetilde{d}_{\max}, |\text{OutEdges}_v|)\}$}{ \label{line:out-inner-truncated-loop}
                    $j \leftarrow \text{OutEdges}_v[i_1]$\label{triangle-count-worker:j}\\
                    $k \leftarrow \text{OutEdges}_v[i_2]$\\
                     $\widetilde{\Delta} \leftarrow \widetilde{\Delta} + \frac{X_{\left\{j,k\right\}} \cdot \left(e^\eps + 1\right) - 1}{e^\eps - 1}$ \label{line:count_triangle}\\  
                }
            }
            Sample $R \sim \text{Lap}\left(\frac{\eps}{2 \cdot \widetilde{d}_{\max}}\right)$ \label{line:laplace}\\
            $\widetilde{\Delta} \leftarrow \widetilde{\Delta} + R$ \\
            \emph{workerTCount} $\leftarrow workerTCount + \widetilde{\Delta}$
        }
        $w.\text{sendTCount}\left(w, \text{\emph{workerTCount}}\right)$ \label{line:Send_triangle}\\
        $w.\text{done}\left(\right)$ 
    }
    \small\caption{\small\label{alg:private-tcount-workerT} Triangle Counting (Worker)}
\end{algorithm}

\begin{figure}
    \centering
    \includegraphics[scale=0.5,width=\linewidth]{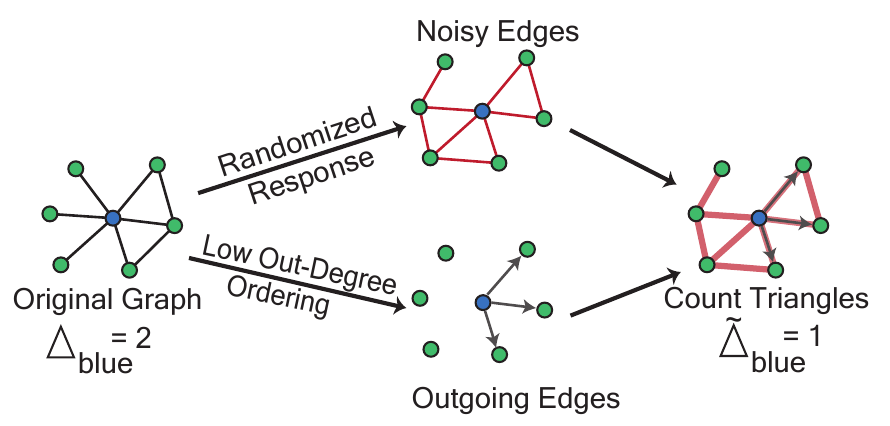}
    \caption{\revision{\footnotesize{\tcountalgo{} for blue node: true triangle count is $2$, but due to Randomized Response and low out-degree ordering, estimate is $1$.}}}
    \label{fig:tcount_algo}
\end{figure}

\subsection{Algorithm Description}
\eat{The triangle counting algorithm consists of three additional computation rounds after the low out-degree ordering ($Z$) is computed using \kcorealgo{}.} \revision{Our algorithm consists of three additional computation rounds after computing the low out-degree ordering ($Z$) using \kcorealgo{}.} \revision{In the first round, each node perturbs its adjacency list using Randomized Response (RR)~\cite{warner1965randomized}}, producing a privacy-preserving set of noisy edges for subsequent computations. In the second round, we calculate the maximum noisy out-degree, $\widetilde{d}_{\max}$, by determining each node's outgoing edges based on \(Z\). While the first two rounds can be combined, we separate them for clarity. In the final round, we compute the number of triangles incident to each node, using the noisy edges and the maximum noisy out-degree, $\widetilde{d}_{\max}$. The algorithm is implemented in a distributed setting, where computation is divided between a coordinator and multiple workers. \eat{The following sections describe this division, present the pseudocode and expand on each round of computation, and detail the roles of the coordinator and workers.}\revision{The pseudocode is structured to reflect this division. \eat{We now describe their respective roles.}}


\label{subsec:tcount-desc}
\noindent\textbf{Coordinator} As shown in~\cref{alg:private-tcount-coordinator}, the coordinator receives the graph size $n$, number of workers $M$, constant parameter $\psi > 0$, privacy parameter $\eps \in (0, 1]$, privacy split fraction $\privatefraction \in (0,1)$, and bias term $b$. It initializes three channels, $cRR, cMaxOut, cTCount$, to receive \eat{Randomized Response}\revision{RR} noisy edges, maximum noisy out-degrees, and local noisy triangle counts from workers (\cref{line:initialize-triangle-coordinator}). \eat{It also initializes \(X\) to store noisy edges for the entire graph.}  The coordinator manages the algorithm’s execution, collecting worker outputs and publishing updates each round. It first computes the low out-degree ordering, \(Z\), using \kcorealgo{}. In the first round, it launches $M$ asynchronous worker processes (\cref{line:round1-start}), each computing and sending noisy edges after \eat{Randomized Response}\revision{RR}. Before the second round, the coordinator aggregates and stores them in \(X\), then publishes \(X\) for global access~(\cref{line:publishRR}), enabling workers to utilize the noisy public edges in subsequent computations. In the second round (\cref{line:round2-start}), workers compute noisy maximum out-degrees for their subgraphs using~\cref{alg:private-maxOutdegree}. The coordinator then determines \(\widetilde{d}_{\max}\), the maximum noisy out-degree across all workers (\cref{line:max-D}). Finally, in the third round (\cref{line:round3-start}), each worker counts the triangles incident to its nodes using the low out-degree ordering, published noisy edges, and \(\widetilde{d}_{\max}\). Workers send noisy local triangle counts to the coordinator, which aggregates them to compute the overall noisy triangle count (\cref{line:tcount}).

\noindent\textbf{Worker (Randomized Response)} As specified in~\cref{alg:private-tcount-workerRR}, workers maintain a \emph{neighborsRR} data structure to store noisy edges. For each node $v$, noisy edges are computed via Randomized Response (RR) with parameter $\eps$, processing only the upper triangular part of the adjacency matrix~(\cref{line:upper_triang}), as the graph is undirected. Specifically, for a node $v$, all indices greater than $v$ are processed using $\text{RandomizedResponse}_\eps(ngh_v)$, which flips the existence of each edge $(v,ngh_v)$ with probability $\frac{1}{e^{\eps} + 1}$~(\cref{line:rr}). Once computed, workers send noisy edges to the coordinator~(\cref{line:send_rr}).  

\noindent\textbf{Worker (Noisy Max Out-Degree)} In~\cref{alg:private-maxOutdegree}, workers maintain a variable $out_{max}$ which stores the maximum noisy out-degree \revision{of their subgraph.} For each node $v$, the worker first computes the out-degree $d_v$ using the order provided in $Z$, where an edge $\left(v,j\right)$ is considered outgoing if $Z[j] > Z[v]$~(\cref{line:outgoing}). The out-degree $d_v$ is the number of outgoing edges from $v$. Then, the worker adds symmetric geometric noise with parameter $\eps$ to $d_v$\revision{, computes the max noisy out-degree, and sends $out_{\max}$ to the coordinator~(\cref{line:send_out}).} 

\noindent\textbf{Worker (Count Triangles)} \eat{The pseudocode for our algorithm is outlined in}\revision{As shown in}~\cref{alg:private-tcount-workerT}, \revision{each worker} computes the number of triangles incident to each node in its respective subgraph. For each node $v$, the outgoing edges are identified and sorted in ascending order by node IDs. The triangle count is determined by iterating over all unique pairs of outgoing neighbors $\{j, k\}$ of $v$, up to $\widetilde{d}_{\max}$~(\cref{line:out-degree-bound},\ref{line:out-inner-truncated-loop}). For each pair, the triangle contribution is calculated as: $\frac{X_{\{j,k\}} \cdot (e^\eps + 1) - 1}{e^\eps - 1}$, where $X_{\{j,k\}}$ represents the noisy presence (1) or absence (0) of an edge between $j$ and $k$~(\cref{line:count_triangle}). To ensure privacy, additional noise is added to the triangle counts using the \revision{Laplace distribution\footnote{We use Laplace noise here as it offers a smoother tradeoff for smaller parameters.}} with parameter $\frac{\eps}{2 \cdot \widetilde{d}_{\max}}$, where $\widetilde{d}_{\max}$ is the global maximum noisy out-degree. Upon completing the computation, each worker aggregates and returns the  noisy triangle count for its entire subgraph~(\cref{line:Send_triangle}).

\revision{
\begin{example}
    In \cref{fig:tcount_algo}, we apply \tcountalgo{} to estimate the number of triangles incident to the blue node. The algorithm first orients the edges using a low out-degree ordering, so each node only considers neighbors with higher order. Randomized Response is then applied to the original adjacency list, and the resulting noisy edges are used in combination with the oriented edges to count triangles. As shown in the figure, while the blue node is part of two true triangles in the original graph, only one triangle is preserved under the noisy edges.
\end{example}

}

\subsection{Theoretical Analysis}
\label{subsec:tcount-theory}

\paragraph{Memory Analysis \& Communication Cost}  
Let $M$ be the number of workers and $n$ the graph size. Each worker processes $S$ nodes, where $S = \lfloor n/M \rfloor$ for $M-1$ workers, and the last worker handles $n - (M-1) \lfloor n/M \rfloor$ nodes. The coordinator manages three communication channels and publishes the noisy edges for the entire graph. The \emph{cRR} structure, which aggregates noisy edges, requires $O(n^2)$ space, while \emph{cTCount} and \emph{cMaxOut}, which collect triangle counts and maximum noisy out-degree, require $O(M)$ space each. Storing published noisy edges further adds $O(n^2)$ space, resulting in a total coordinator memory requirement of $O(n^2 + M)$. Each worker processes $O(S)$ nodes, requiring $O(Sn)$ space for the graph. The \emph{neighborsRR} structure for storing noisy edges demands $O(Sn)$ space, while computing the maximum noisy out-degree requires $O(S)$. The final triangle count computation takes $O(S \cdot \widetilde{d}_{\max})$ space, where $\widetilde{d}_{\max}$ is the maximum noisy out-degree, leading to an overall worker memory requirement of $O(Sn)$. The algorithm runs three communication rounds beyond those for low out-degree ordering. Workers first send noisy edges, incurring $O(Sn)$ communication cost, followed by sending the maximum noisy out-degree and triangle counts, each requiring $O(M)$ communication. Thus, the total communication overhead for the algorithm is $O(n^2 + M)$.  

\paragraph{Privacy Guarantees} As before, our privacy guarantees are proven by implementing our 
triangle counting algorithm using local randomizers. 

\begin{lemma}\label{lem:triangle-privacy} 
    Our triangle counting algorithm is $\eps$-LEDP. 
\end{lemma}
\eat{\begin{proof}
Our triangle counting algorithm releases three sets of information, in addition to the call of~\cref{alg:private-kcore-coordinator}. 
    First, each \eat{individual} node performs randomized response on 
    the upper triangular part of the adjacency matrix to determine a privacy-preserving
    set of edges. Then, each node releases its privacy-preserving out-degree.
    Finally, each node releases a privacy-preserving triangle count using
    its outgoing edges obtained from the low out-degree ordering.
    
    First, by~\cref{thm:k-core}, the low out-degree ordering is
    $(\eps/4)$-LEDP.
    Then, by~\cref{lem:rr-dp}, each adjacency list output from our randomized response 
    is a $(\eps/4)$-local randomizer. 
    Each node computes its noisy out-degree~(\cref{line:noise-out-degree}), where the sensitivity is $1$, conditioning on Z, for neighboring adjacency lists.
    Hence, by the privacy of the geometric mechanism (\cref{lem:sgd-private}),
    each node uses a $(\eps/4)$-local randomizer to output its noisy degree.

    Finally, we condition on $\widetilde{d}_{\max}$ to bound the sensitivity of our 
    triangle counts~(\cref{line:out-degree-bound}). To bound the sensitivity, we truncate the outgoing adjacency list (computed using $Z$) 
    of each node by $\widetilde{d}_{\max}$. We now compute the sensitivity of the truncated list 
    on neighboring adjacency lists $\adj$ and $\adj'$. Suppose, without loss of generality, that 
    $\adj'$ contains one more neighbor than $\adj$. Let $w$ be the neighbor that is in $\adj'$ but not in 
    $\adj$. Let $\overline{\adj}$ and $\overline{\adj'}$ be the truncated adjacency lists for $\adj$ and $\adj'$, respectively.
    To upper bound the sensitivity, let's suppose in the worst case, $\overline{\adj}$ contains one node, $u$,
    not in $\overline{\adj'}$ and $\overline{\adj'}$ contains $w$ (which is not in $\overline{\adj}$). \eat{For the sake of upper bounding the sensitivity, suppose that in the worst-case, $\overline{\adj}$ contains one node, $u$,
    not in $\overline{\adj'}$ and $\overline{\adj'}$ contains $w$ (which is not in $\overline{\adj}$).} 
    Then, assume first that $j = u$: the
    first for loop (\cref{line:out-degree-bound}) counts at most $\widetilde{d}_{\max}$
    additional triangles in $\overline{\adj}$ for $u$ (symmetrically counts at most $\widetilde{d}_{\max}$ additional
    triangles in $\overline{\adj'}$ for $w$). Suppose, without loss of generality,
    in the worst-case, $u$ returns $\widetilde{d}_{\max}$ additional triangles
    and $w$ returns no triangles.
    Suppose we set $j$ to every other node (i.e.\ $j \neq u$)
    in $\overline{\adj}$; then, $j$ encounters $u$ and not $w$ in the second for loop (\cref{line:out-inner-truncated-loop}) in 
    $\overline{\adj}$. Thus, in total, we encounter an additional
    $\widetilde{d}_{\max}$ triangles in $\overline{\adj}$ for all $j \neq u$.
    In total, we encounter $2\widetilde{d}_{\max}$ additional triangles in $\overline{\adj}$
    than in $\overline{\adj'}$ and the sensitivity of the number of counted triangles is $2 \widetilde{d}_{\max}$. 
    Using the privacy of the Laplace mechanism (\cref{lem:laplace}), we implement outputting the 
    local triangle counts via an $(\eps/4)$-local randomizer. 
    

    Finally, we make the observation that the edge that differs between neighboring input graphs, $G$ and $G'$, contributes
    to the out-degree of at most one node.
    Thus, using composition (\cref{thm:composition}) over all four $(\eps/4)$-local randomizers,
    we obtain an $\eps$-LEDP algorithm for triangle counting.
\end{proof}}

\begin{proof}
Our triangle counting algorithm calls~\cref{alg:private-kcore-coordinator}, which by~\cref{thm:k-core} is $(\eps/4)$-LEDP. Additionally, we release three sets of information, each of which we show to be $(\eps/4)$-LEDP.

First, each node applies Randomized Response to the upper triangular adjacency matrix to generate a privacy-preserving set of edges. By~\cite{DMNS06}, this adjacency list output is a $(\eps/4)$-local randomizer.

Second, each node releases its privacy-preserving out-degree. By~\cref{line:noise-out-degree}, the sensitivity of the out-degree (conditioning on $Z$) is 1 for neighboring adjacency lists. By the privacy of the geometric mechanism (\cite{BV18,CSS11,DMNS06,DNPR10}), each node uses a $(\eps/4)$-local randomizer to output its noisy degree.

Third, each node releases a privacy-preserving triangle count using its outgoing edges from the low out-degree ordering. To bound the sensitivity, we truncate the outgoing adjacency list (computed using $Z$) of each node by $\widetilde{d}_{\max}$. Given neighboring adjacency lists $\adj$ and $\adj'$, assume $\adj'$ contains one additional neighbor $w$ (without loss of generality). Let $\overline{\adj}$ and $\overline{\adj'}$ be the truncated adjacency lists. In the worst case, $\overline{\adj}$ contains a node $u$ not in $\overline{\adj'}$, while $\overline{\adj'}$ contains $w$ (not in $\overline{\adj}$). Let $j$ be defined as in~\cref{triangle-count-worker:j}. If $j = u$, the first for-loop (\cref{line:out-degree-bound}) counts at most $\widetilde{d}_{\max}$ additional triangles for $u$ (symmetrically for $w$). Assuming $u$ returns $\widetilde{d}_{\max}$ triangles and $w$ returns none, then for all other nodes $j \neq u$, the second for-loop (\cref{line:out-inner-truncated-loop}) encounters at most $\widetilde{d}_{\max}$ additional triangles. Thus, the total difference in counted triangles between $\overline{\adj}$ and $\overline{\adj'}$ is $2\widetilde{d}_{\max}$, giving a sensitivity of $2\widetilde{d}_{\max}$. By the privacy of the Laplace mechanism (\cite{DMNS06}), outputting local triangle counts is an $(\eps/4)$-local randomizer.

Since the differing edge between neighboring graphs $G$ and $G'$ affects at most one node’s out-degree, applying composition (\cite{DMNS06,DL09,DRV10}) over all four $(\eps/4)$-local randomizers results in an $\eps$-LEDP triangle counting algorithm.
\end{proof}

\paragraph{Approximation Guarantees}

\eat{First, since each node uses three additional rounds in addition to computing the low out-degree ordering, 
our algorithm runs in $O(\log(n) \log\left(D_{\max}\right))$ rounds.}
One of the major novelties in our proofs is via a new intricate use of the Law
of Total Expectation and Law of Total Variance for the events where 
the out-degrees of each node is upper bounded by the noisy maximum out-degree $\widetilde{d}_{\max}$
(which is, in turn, upper bounded by the degeneracy $\widetilde{O}(d)$).
Such use cases were unnecessarily in~\cite{IMC21communication,ELRS23} because they did not use 
oriented edges. \revision{We first upper bound the out-degree by $\widetilde{O}\left(\frac{d}{\eps}\right)$.
\begin{lemma}\label{lem:dmax}
    Given a graph where edges are oriented according to~\cref{alg:private-tcount-workerT}, the maximum out-degree of any node is at most 
    $O\left(d + \frac{\log(D_{\max})\log^2(n)}{\eps}\right)$.
\end{lemma}

\begin{proof}
    By Invariant 1 and Theorem 3.4 (in the supplementary materials), the out-degree of any node $v$ is at most $(2 + \eta)k(v) + O\left(\frac{\log(D_{\max})\log^2(n)}{\eps}\right)$, with high probability, where $k(v)$ is the core number of $v$. The largest core number is equal to the degeneracy of the graph~\cite{seidman1983network}.
    Hence, the maximum out-degree of any node is upper bounded by $(2+\eta)d + O\left(\frac{\log(D_{\max})\log^2(n)}{\eps}\right) = 
    O\left(d + \frac{\log(D_{\max})\log^2(n)}{\eps}\right)$
\end{proof}

\begin{lemma}\label{orientedfourcyclesbound}
Given a graph where edges are oriented according to~\cref{alg:private-tcount-workerT},
the number of oriented 4-cycles, denoted by $\overrightarrow{C}_4$, where each cycle contains two non-adjacent nodes with outgoing edges to the remaining two nodes (see~\cref{fig:oriented-c4}), is at most $\widetilde{O}\left(\frac{n^2 d^2}{\eps^2}\right)$, with high probability, where $d$ is the degeneracy of the graph.
\end{lemma}
\begin{proof}
There are $O(n^2)$ unordered pairs of vertices $\{w, x\}$ that may serve as the black nodes in an oriented 4-cycle (see~\cref{fig:oriented-c4}). For a fixed pair $\{w, x\}$, define $S_{w,x} = \{ u \in V : w \to u \text{ and } x \to u \}$ to be the set of vertices that are the outgoing endpoints of the
outgoing edges from both $w$ and $x$. 

Under the low out-degree orientation computed by~\cref{alg:private-tcount-workerT} and by~\cref{lem:dmax}, each vertex has at most $O\left(d + \frac{\log(D_{\max})\log^2(n)}{\eps}\right)$ out-neighbors and thus can have at most $O\left(\left(d + \frac{\log(D_{\max})\log^2(n)}{\eps}\right)^2\right) = O\left(d^2 + 
\frac{\log^2(D_{\max})\log^4(n)}{\eps^2}\right)$ 
outgoing red pairs. Each pair $\{u, v\} \subseteq S_{w,x}$ forms an oriented 4-cycle with $\{w, x\}$, contributing at most this many
cycles per black pair of vertices. Summing over all $O(n^2)$ black vertex pairs yields a total of at most $O\left(n^2 \left(d^2 + 
\frac{\log^2(D_{\max})\log^4(n)}{\eps^2}\right)\right) 
= \widetilde{O}\left(\frac{n^2 d^2}{\eps^2}\right)$ oriented 4-cycles.
\end{proof}
}

\begin{figure}
    \centering
    \includegraphics[width=0.2\linewidth]{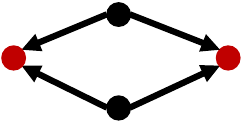}
    \caption{Oriented cycle of length $4$; two non-adjacent black nodes have edges oriented toward the remaining red nodes.}
    \label{fig:oriented-c4}
\end{figure}

\eat{\qq{below proof can go into the appendix}}

\begin{restatable}{lemma}{triangleexpectation}
\label{lem:triangle-expectation}
    In expectation, our algorithm returns a $2$-approximation of the true triangle count\eat{,i.e.,}: $\frac{(n^3-1) \cdot T}{n^3} \leq \expect[\widetilde{\Delta}] \leq T$.
\end{restatable}
\begin{proof}
\label{proof:triangleexpectation}
    We first prove that, in expectation,
    $\expect\left[\frac{X_{j, k} \cdot (e^{\eps/4} + 1) - 1}{e^{\eps/4} - 1}\right] = \mathbbm{1}_{j, k}$
    where $\mathbbm{1}_{j, k}$ is the indicator variable for whether edge $\{j, k\}$ exists in the original private graph. 
    Since randomized response flips the bit indicating the existence of an edge between each pair of vertices with probability $\frac{1}{e^{\eps/4} + 1}$,
    we can simplify the expression to be:

    \begin{align*}
        \expect\left[\frac{X_{j, k} \cdot (e^{\eps/4} + 1) - 1}{e^{\eps/4} - 1}\right] &= \frac{\expect[X_{j, k}] \cdot (e^{\eps/4} + 1)
        - 1}{e^{\eps/4} - 1}.
    \end{align*}

    When $\mathbbm{1}_{j, k} = 1$, the probability that we obtain a bit of $1$ is $\frac{e^{\eps/4}}{e^{\eps/4} + 1}$ and
    our expression simplifies to
    \begin{align*}
        \frac{\expect[X_{j, k}] \cdot (e^{\eps/4} + 1)
        - 1}{e^{\eps/4} - 1} &= \frac{\frac{e^{\eps/4}}{e^{\eps/4} + 1} \cdot (e^{\eps/4} + 1)
        - 1}{e^{\eps/4} - 1}
        = 1.
    \end{align*}

    When $\mathbbm{1}_{j, k} = 0$, the probability that we obtain a bit of $1$ is $\frac{1}{e^{\eps/4} + 1}$
    then our expression simplifies to
    \begin{align*}
        \frac{\expect[X_{j, k}] \cdot (e^{\eps/4} + 1)
        - 1}{e^{\eps/4} - 1} &= \frac{\frac{1}{e^{\eps/4} + 1} \cdot (e^{\eps/4} + 1)
        - 1}{e^{\eps/4} - 1} 
        = 0.
    \end{align*}

    The expected value of the random variable obtained for each edge $e$ using our randomized response procedure is 
    equal to $\mathbbm{1}_e$.

    Now, we condition on the event that $\widetilde{d}_{\max}$ upper bounds the out-degree of every vertex after computing 
    and conditioning on $D$ (the ordering). Let $U$ be this event.
    
    Then, for each node, we compute all pairs of its outgoing neighbors and use the random variable indicating existence of
    the edge spanned by the endpoints to count every triangle composed of a pair of outgoing edges. Let $T_{v, j, k}$ 
    be the random variable representing the existence of the 
    queried triangle for node $v$ and outgoing edges $(v, j)$ and $(v, k)$. Then, $\expect[T_{v, j, k} \mid U] = \expect[X_{j, k}]$
    since both outgoing edges $(v, j)$ and $(v, k)$ exists. By what we showed above, it then follows that $\expect[T_{v, j, k}] = \mathbbm{1}_{j,k}$.

    We must account for the symmetric geometric noise. Since the expectation of the symmetric geometric noise is $0$ and we
    add together all of the values for each of the drawn noises, the expected total noise is $0$ by linearity of expectations. 
    Each triangle has a unique node that queries it since every triangle has a unique orientation of edges where two edges
    are outgoing from a vertex in the triangle. Hence, the expected sum of all queried triangles is $T$, conditioned on $U$, by linearity of 
    expectations.

    Finally, we remove the conditioning on the event $U$. Recall that $\expect[W] = \sum_{y} \expect[W\mid Y = y] \cdot \prob(Y = y)$. 
    Since truncation can only decrease the number of queried triangles (and hence the expectation), we upper and lower bound 
    $0 \leq \expect[\widetilde{\Delta} \mid \neg U] \leq T$. Hence we only need to figure out the probability $P(U)$ to upper and lower bound
    $\expect[\widetilde{\Delta}]$. The probability of event $U$ is the probability that $\max\left(\{d_v + \geom(\eps/4)\mid v \in V\}\right) 
    + \frac{\log(n)}{\eps} \geq \max(\{d_v \mid v \in V\})$ where $d_v$ is the out-degree of node $v$ given order $D$. This probability is lower bounded by 
    the probability that $\max(\{d_v \mid v \in V\}) + \geom(\eps/4) + \frac{c\log(n)}{\eps} \geq \max(\{d_v \mid v \in V\})$ for a fixed 
    constant $c \geq 1$; this 
    means we want to lower bound $\prob\left(\left(Q + \frac{c\log(n)}{\eps}\right) \geq 0\right)$ where $Q \sim \geom(\eps/4)$. By concentration of random variables chosen from the symmetrical geometric 
    distribution, we know that $\prob\left(\left(Q + \frac{c\log(n)}{\eps}\right) \geq 0\right) \geq 1 - \frac{1}{n^3}$ for large enough constant $c \geq 1$. Specifically, setting $c = 3$ gives us this bound. Hence, we can upper and lower bound $\frac{(n^3-1) \cdot T}{n^3} \leq \expect[\widetilde{\Delta}] \leq T$.
\end{proof}

To calculate the variance of the triangle count obtained from our algorithm, we 
use a quantity denoted by $\overrightarrow{C}_4$ which is the number of oriented
$4$-cycles where there exists two non-adjacent nodes with outgoing 
edges to the other two nodes in the cycle. See~\cref{fig:oriented-c4} for an example.
\eat{To calculate the variance of the triangle count obtained from our algorithm, we 
use a quantity denoted by $\overrightarrow{C}_4$. The number of oriented cycles of length $4$
is denoted by $\overrightarrow{C}_4$ and the quantity is equal to the number of
$4$-cycles where there exists two non-adjacent nodes with outgoing 
edges to the other two nodes in the cycle. See~\cref{fig:oriented-c4} for an example.} The number of oriented cycles of 
length $4$ is upper bounded by $n^2 d^2$ where $d$ is the degeneracy (maximum core number) in the graph, resulting in significant gains in utility 
over previous results which use the number of total (not oriented) $4$-cycles which could be as 
large as $\Omega(n^4)$.

\begin{restatable}{lemma}{trianglecountingvariance}
    \label{lem:triangle-counting-variance}
    Our triangle counting algorithm returns a count with variance $O\left(\frac{n d^2 \log^6 n}{\eps^4} + \overrightarrow{C}_4\right)$ where 
    $d$ is the degeneracy (max core number) of the graph, $\overrightarrow{C}_4$ is the number of oriented cycles of length $4$,
    and $T$ is the number of (true) triangles in the private graph.
\end{restatable}
\begin{proof}
\label{proof:trianglecountingvariance}
    As before, we first calculate the variance conditioned on the event $U$ that $\widetilde{d}_{\max}$ upper
bounds the out-degrees of every node. Then, we use the law of total variance to remove this condition.

For notation simplicity, we omit the condition on $U$ from the right hand sides of the below equations.
First, the variance of $\mathrm{Var}[X_{\{i, j\}}] = \frac{e^{\eps}}{(e^{\eps} + 1)^2}$, 
since $X_{\{i, j\}}$ is a Bernoulli variable.
Then, let $\hat{T}$ be our returned triangle count. 
The variance 
{\scriptsize
\begin{align*}
    &\var{\hat{T} \mid U}\\ &= \var{\sum_{v \in [n]} \left(\sum_{j, k \in Out(v)} \left(\frac{X_{j, k} \cdot (e^{\eps/4} + 1) - 1}{e^{\eps/4} - 1}\right) + \lap\left(\frac{\eps}{2 \widetilde{d}_{\max}}\right)\right)}.
\end{align*}}

Recall $\widetilde{d}_{\max}$ is our noisy maximum out-degree of any vertex. 
Then, our variance simplifies to
{\scriptsize
\begin{align*}
    \left(\frac{e^{\eps/4} + 1}{e^{\eps/4} - 1}\right)^2 \cdot \var{\sum_{v \in [n]} \sum_{j, k \in Out(v)} X_{j, k}} + \var{\sum_{v \in [n]} 
    \lap\left(\frac{\eps}{2\widetilde{d}_{\max}}\right)}.
\end{align*}}

By the variance of the Laplace distribution, we can compute 
{\small
\begin{align*}
    \var{\sum_{v \in [n]} 
    \lap\left(\frac{\eps}{2\widetilde{d}_{\max}}\right)} \leq n \cdot \frac{8\widetilde{d}_{\max}^2}{\eps^2}.
\end{align*}}

Now, it remains to compute $\var{\sum_{v \in [n]} \sum_{j, k \in Out(v)} X_{j, k}}$. The covariance is $0$ if two queried pairs
do not query the same $X_{j, k}$. The covariance is non-zero only in the case of queries which overlap in $X_{i,j}$.
This occurs only in the case of an oriented $4$-cycle where $X_{i, j}$ is shared between two queries. Let $P_2$ be the 
set of all pairs of such queries that share $X_{j, k}$. In this case, the covariance is upper bounded by $\expect[X_{j, k}^2]$.
Hence, we can simplify our expression to be 
{\small
\begin{align}
    &\var{\sum_{v \in [n]} \sum_{j, k \in Out(v)} X_{j, k}}\label{eq:var-1}\\
    &\leq \sum_{v \in [n]}\sum_{j, k \in Out(v)} \var{X_{j, k}} + 
    2 \cdot \sum_{T_{v, j, k}, T_{w, j, k} \in P_2} \left(\expect[X_{j, k}^2]\right)\label{eq:var-2} \\
    &\leq n\cdot \widetilde{d}_{\max}^2 \cdot \frac{e^{\eps/4}}{(e^{\eps/4} + 1)^2} 
    + \overrightarrow{C}_4 \cdot \left(\frac{e^{\eps/4}}{e^{\eps/4} + 1}\right),\label{eq:var-3}
\end{align}}
where $\overrightarrow{C}_4$ indicates the number of directed cycle of length $4$ (\cref{fig:oriented-c4}). 

Hence, our total variance is upper bounded by
{\small
\begin{align*}
\frac{8n\widetilde{d}_{\max}^2}{\eps^2} + n\cdot \widetilde{d}_{\max}^2 \cdot \frac{e^{\eps/4}}{(e^{\eps/4} + 1)^2} + \overrightarrow{C}_4 \cdot \frac{e^{\eps/4}}{e^{\eps/4} + 1}.
\end{align*}}

Finally, by the guarantees of our $k$-core decomposition algorithm, the maximum out-degree $d_{\max}$ is bounded by 
$d_{\max} \leq (2+\eta) \cdot d + O\left(\frac{\log(D_{\max})\log^2 n}{\eps}\right)$, with
high probability, where $d$ is the degeneracy of the input graph and $D_{\max} \leq n$ is the maximum degree in the graph.
Finally, we also know that the noise generated for $\widetilde{d}_{\max}$ is upper bounded by 
$O\left(\frac{\log n}{\eps}\right)$. Thus, $\widetilde{d}_{\max} \leq (2 + \eta) \cdot d + O\left(\frac{\log^3 n}{\eps}\right)$ Hence,
our variance is upper bounded by $O\left(\frac{n d^2 \log^6 n}{\eps^4} + \overrightarrow{C}_4\right)$.

We now remove our condition on $U$ and use the law of total variance to compute our unconditional variance. Recall the law of total 
variance states $\var{Y} = \expect[\var{Y \mid X}] + \var{\expect[Y \mid X]}$. 
We computed $\var{\hat{T}\mid U}$ above. We now compute $\var{\hat{T} \mid \neg U}$. The main difference 
between when the event $U$ occurs and \eat{when the event $U$} does not occur is that some adjacency lists of the outgoing 
neighbors will be truncated. Consequently, we sum over fewer $X_{j, k}$ variables in~\cref{eq:var-1,eq:var-2}.
Thus, the variance when $U$ does not occur is upper bounded by the variance when $U$ does occur. 

Now we calculate $\var{\expect[\hat{T} \mid U]} = \expect[\expect[\hat{T} \mid U]^2] - \expect[\expect[\hat{T} \mid U]]^2$.
By our calculation in the proof of~\cref{lem:triangle-expectation}, we can calculate 
$\expect[\hat{T} \mid U] = T$ and let $\frac{(n^3 - 1) \cdot T}{n^3} \leq Y = \expect[\hat{T} \mid \neg U] \leq T$. Then, 
{\small
\begin{align*}
    \expect[\expect[\hat{T} \mid U]^2] - \expect[\expect[\hat{T} \mid U]]^2 &= \left(\frac{T^2}{2} + \frac{Y^2}{2}\right) - 
    \left(\frac{T}{2} + \frac{Y}{2}\right)^2\\
    &= \frac{T^2 + Y^2}{2}- \frac{T^2 + 2TY + Y^2}{4}\\
    &= \frac{T^2 + Y^2 - 2TY}{4}\\
    &= \left(\frac{T-Y}{2}\right)^2\\
    &\leq \left(\frac{T-\frac{(n^3 - 1) \cdot T}{n^3}}{2}\right)^2\\
    &\leq \left(\frac{T}{2n^3}\right)^2\\
    &\leq \frac{1}{4}. 
\end{align*}}

Thus, our final variance is $O\left(\frac{n d^2 \log^6 n}{\eps^4} + \overrightarrow{C}_4\right)$.
\end{proof}


\begin{restatable}{theorem}{trianglecounting}
\label{thm:triangle-counting}
     With high constant probability, our triangle counting algorithm returns a $\left(1+\eta, O\left(\frac{\sqrt{n}d\log^3 n}{\eps^{2}} + \sqrt{\overrightarrow{C}_4}\right)\right)$-approximation of the true triangle count.
\end{restatable}
\begin{proof}
     We use Chebyshev's inequality with the standard deviation calculated from~\cref{lem:triangle-counting-variance}. The $(1+\eta)$-approximation
     comes from our $\left(1 - \frac{1}{n^3}\right)$-approximation of the expectation.
\end{proof}

See~\cref{table:kcore} for the $d$ values of real-world graphs; when $d = O(1)$ is constant, as is the case for real-world graphs,
then $\overrightarrow{C}_4 = O(n^2)$ and $T = O(n)$. We improve previous theoretical additive errors from $O\left(\frac{\sqrt{C_4}}{\eps} + 
\frac{n^{3/2}}{\eps^2}\right)$ \cite{IMC21communication} to $O\left(\sqrt{\overrightarrow{C}_4} + \frac{\sqrt{n}\log^3 n}{\eps^2}\right)$,
an improvement of at least a $\Omega(\sqrt{n})$ factor, translating to massive practical gains.

\section{Experimental Evaluation}\label{sec:experiments}
\eat{In this section, we evaluate the performance and accuracy of our algorithms for $k$-core decomposition (\kcorealgo{}) and triangle counting (\tcountalgo{}) in 
the distributed (LEDP-DS) computation framework 
to benchmark the runtime and accuracy. 
For $k$-core decomposition we compare our algorithm
with the LEDP \textbf{$k$-Core} decomposition algorithm of~\cite{DLRSSY22}, which we implement.
For triangle counting, we also compare our algorithm with \textbf{ARROneNS}$_\Delta\left(\textbf{Lap}\right)$~\cite{IMC21communication} and \textbf{GroupRR}~\cite{hillebrand2023communication}. Our results show that our algorithms effectively scale to graphs with more than a billion edges, while consistently maintaining approximations significantly better than our theoretical bounds and previous works.
For $k$-core decomposition, our algorithm reduces the number of rounds by nearly \defn{two orders of magnitude} over~\cite{DLRSSY22} across all graphs, while improving the accuracy. For triangle counting, our approach demonstrates an accuracy improvement of nearly \defn{six orders of magnitude} on the multiplicative approximation over prior methods and achieve speed-ups for large graphs.}

\eat{In this section, we evaluate the performance and accuracy of our $k$-core decomposition (\kcorealgo{}) and triangle counting (\tcountalgo{}) algorithms in a distributed simulation, benchmarking runtime and accuracy. For $k$-core decomposition, we compare our algorithm with the LEDP \textbf{$k$-Core} decomposition algorithm of~\cite{DLRSSY22}, which we implement. For triangle counting, we also compare against \textbf{ARROneNS}$_\Delta\left(\textbf{Lap}\right)$~\cite{IMC21communication} and \textbf{GroupRR}~\cite{hillebrand2023communication}. Our results show that our algorithms scale to graphs with over a billion edges while consistently achieving significantly better approximations than theoretical bounds and prior work.  

For $k$-core decomposition, our algorithm reduces the number of rounds by nearly \defn{two orders of magnitude} over~\cite{DLRSSY22} across all graphs while improving accuracy. For triangle counting, our approach improves multiplicative approximation accuracy by nearly \defn{six orders of magnitude} over previous methods and achieves speed-ups on large graphs.}

\revision{

In this section, we evaluate the performance and accuracy of our $k$-core decomposition (\kcorealgo{}) and triangle counting (\tcountalgo{}) algorithms under Local Edge Differential Privacy (LEDP) using a distributed simulation. We benchmark against prior LEDP algorithms, and to highlight the limitations of Randomized Response (RR), we additionally implement RR-based baselines for both problems. All RR baselines are evaluated purely in terms of accuracy, as they are centralized algorithms and not directly comparable in runtime to our distributed setting. Furthermore, RR introduces significant computational overhead due to increased graph density from edge perturbation. \textbf{Consequently, we omit RR results on many large graphs: $k$-core baselines fail due to out-of-memory (OOM) errors, while triangle counting exceeds timeout limits. These failures occur independently, underscoring the instability and inefficiency of naive RR methods on large-scale graphs.}

\noindent\textbf{$k$-Core Baselines.}
We compare \kcorealgo{} against the LEDP $k$-core decomposition algorithm of~\cite{DLRSSY22} (denoted $k$-Core), which we implement. Additionally, we construct an RR-based baseline (denoted $k$-CoreRR) that runs the standard peeling algorithm on the RR-perturbed graph and applies a scaling factor to correct the induced degrees, ensuring unbiased estimates. Our method achieves significantly better accuracy, reducing approximation error by up to \textbf{two orders of magnitude} over $k$-CoreRR and outperforming $k$-Core~\cite{DLRSSY22} in both accuracy and efficiency—reducing the number of rounds by nearly \textbf{two orders of magnitude} across all graphs. $k$-CoreRR consistently fails on larger datasets due to memory exhaustion caused by increased graph density.

\noindent\textbf{Triangle Counting Baselines.}
We compare \tcountalgo{} against two LEDP triangle counting algorithms: \textbf{ARROneNS}$_\Delta\left(\textbf{Lap}\right)$\cite{IMC21communication} and \textbf{GroupRR}\cite{hillebrand2023communication}. For the RR baseline (denoted as TCountRR), we follow the approach from~\cite{edenICALPS}. Our algorithm achieves up to \textbf{six orders of magnitude} improvement in multiplicative accuracy while providing substantial speedups in larger graphs.

\eat{\noindent\textbf{RR Baselines.} To demonstrate the limitations of Randomized Response (RR), we implement naive RR-based baselines for both problems. For triangle counting, we follow the algorithm from~\cite{edenICALPS}, applying it to the perturbed adjacency matrix. For $k$-core decomposition, we run the standard peeling algorithm on the noisy graph and apply a derived scaling factor to the induced degrees to correct for RR noise. While \cite{edenICALPS} provides the scaling for triangle counting, we extend this approach to $k$-core estimation.
Our comparisons focus solely on accuracy, as the RR baselines are centralized and not comparable in runtime to our distributed simulation. Moreover, RR perturbation drastically inflates graph density, resulting in significant resource overhead. \textbf{As a result, we omit RR results for many large graphs due to either out-of-memory (OOM) errors for $k$-core or exceeding the timeout for triangle counting. Notably, these occur independently - there are several graphs where triangle counting fails due to timeout but $k$-core results are still obtained. This highlights the computational instability of naive RR methods at scale.} Across all feasible datasets, our methods deliver superior accuracy—improving approximation by up to \textbf{two orders of magnitude for $k$-core} and \textbf{six for triangle counting}.}

\eat{\textbf{TODO:} Write why we do better and how TCountRR is comparable for denser graphs.} 
}
\begin{table}[t]
\caption{Graph size, maximum core number, and number of triangles.} \label{table:kcore}
\begin{center}
\small
\resizebox{\columnwidth}{!}{%
\begin{tabular}[!t]{l|r|r|r|r|r}
\toprule
{Graph Name} & Num. Vertices & Num. Edges & Max. Degree & Max. Core Num. ($d$) & Num. Triangles\\
\midrule
{\emph{ email-eu-core}}         & 986       &1,329,336 & 345 & 34 & 105,461 \\
{\emph{ wiki  }  }    		 & 7115  &100,761 & 1065 & 35 & 608,387 \\
{\emph{ enron        }  }    & 36,692        &183,830 & 1,383& 43 & 727,044 \\
{\emph{ brightkite      }  }    	 & 58,228      &214,078 & 1,134& 52 & 494,728\\
{\emph{ ego-twitter}}         & 81,306       &1,342,296 & 3,383& 96 &  13,082,506\\
{\emph{ gplus}}         & 107,614       &12,238,285 & 20,127& 752 & 1,073,677,742 \\
{\emph{ stanford}}         & 281,903       &1,992,635 & 38,625 & 71 & 11,329,473 \\
{\emph{ dblp  }  }           & 317,080          &1,049,866 & 343   & 113 & 2,224,385\\
{\emph{ brain        }  }        & 784,262          &267,844,669 & 21,743 & 1200 & --\\
{\emph{ orkut        }  }    & 3,072,441        &117,185,083 & 33,313 & 253 & --\\
{\emph{ livejournal  }  }    		 & 4,846,609  &42,851,237 & 20,333 & 372 & --\\
{\emph{ twitter      }  }    	 & 41,652,230       &1,202,513,046 & 2,997,487 & 2488 & --\\
{\emph{ friendster}}         & 65,608,366       &1,806,067,135 & 5214 & 304 & --\\
\end{tabular}
}
\end{center}
\end{table}

\begin{figure*}[!ht]
     \centering
     
     \begin{subfigure}[c]{0.33\textwidth}
     \centering
         \includegraphics[width =\textwidth]{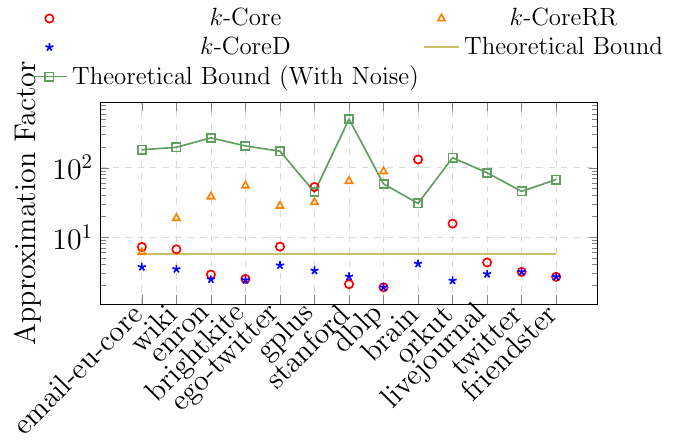}
         \caption{\revision{Average Approximation Factor}}
         \label{fig:kcore_avg_approx}
     \end{subfigure}
     \hfill
     \begin{subfigure}[c]{0.33\textwidth}
         \centering
         \includegraphics[width =\textwidth]{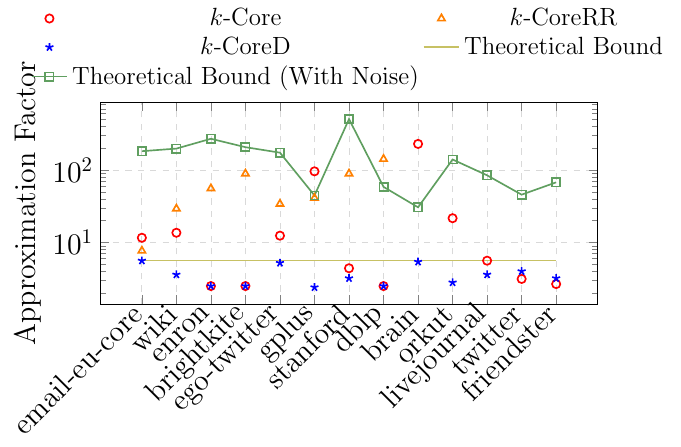}
         \caption{\revision{$80^{th}$ Percentile Approximation Factor}}
         \label{fig:kcore_80_approx}
     \end{subfigure}
     \hfill
     \begin{subfigure}[c]{0.33\textwidth}
         \centering
         \includegraphics[width =\textwidth]{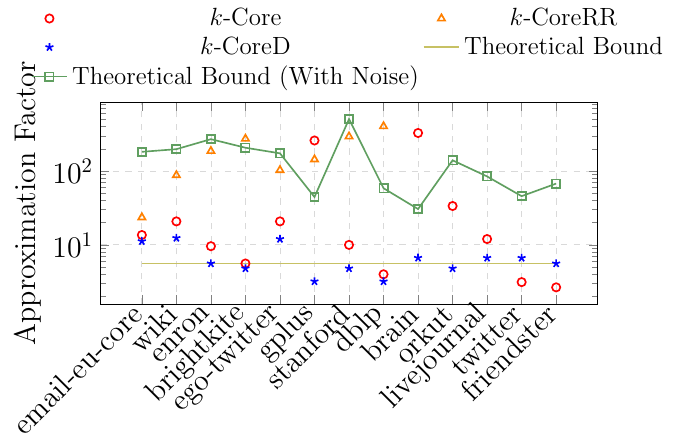}
         \caption{\revision{$95^{th}$ Percentile Approximation Factor}}
         \label{fig:kcore_95_approx}
     \end{subfigure}

        \caption{$k$-core Decomposition Results.}
        \label{fig:kcore_all}
\end{figure*}

\begin{figure}[t]
     \centering
     \includegraphics[width =0.65\linewidth]{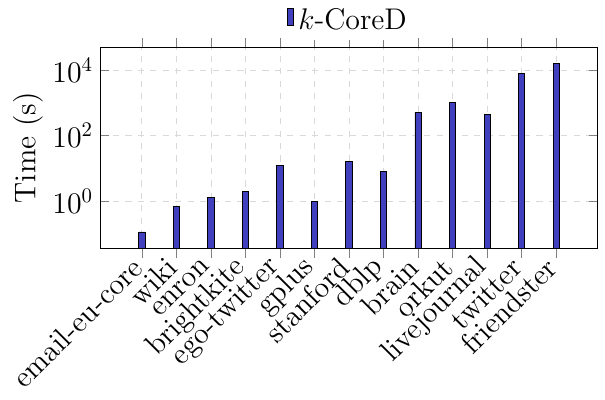}
     \caption{$k$-Core Decomposition Avg. Response Time}
     \label{fig:kcore_runtime}
\end{figure}

\begin{figure}[t]
\centering
     \includegraphics[width=0.8\linewidth]{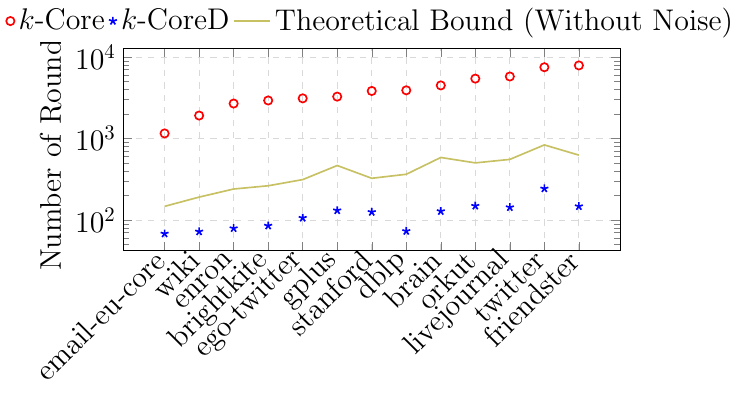}
        \caption{$k$-Core Decomposition Number of Rounds}
        \label{fig:kcore_rounds}
\end{figure}

     

\begin{figure*}[!ht]
     \centering
     
     \begin{subfigure}[c]{0.16\textwidth}
         \centering
         \includegraphics[width =\textwidth]{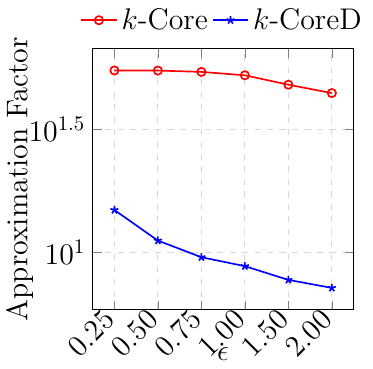}
         \caption{gplus}
         \label{fig:kcore_avg_approx_gplus_eps}
     \end{subfigure}
     \hfill
     \begin{subfigure}[c]{0.16\textwidth}
         \centering
         \includegraphics[width =\textwidth]{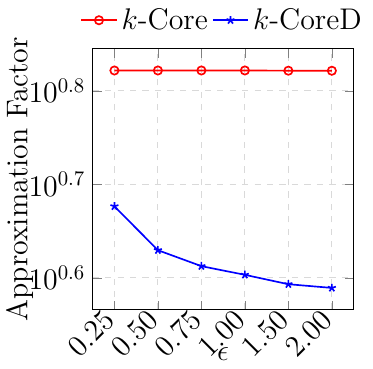}
         \caption{wiki}
         \label{fig:kcore_avg_approx_wiki_eps}
     \end{subfigure}
     \hfill
     \begin{subfigure}[c]{0.16\textwidth}
         \centering
         \includegraphics[width =\textwidth]{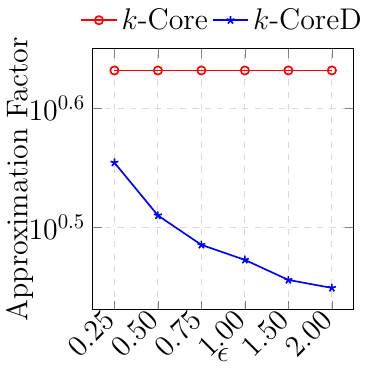}
         \caption{livejournal}
         \label{fig:kcore_avg_approx_liverjournal_eps}
     \end{subfigure}
        \hfill
        \begin{subfigure}[c]{0.16\textwidth}
         \centering
         \includegraphics[width =\textwidth]{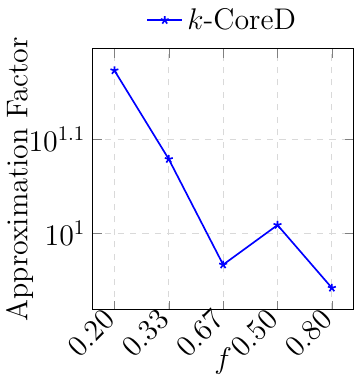}
         \caption{gplus}
         \label{fig:kcore_avg_approx_gplus_fraction}
     \end{subfigure}
     \hfill
     \begin{subfigure}[c]{0.16\textwidth}
         \centering
         \includegraphics[width =\textwidth]{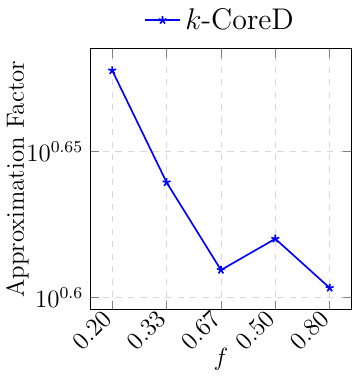}
         \caption{wiki}
         \label{fig:kcore_avg_approx_wiki_fraction}
     \end{subfigure}
     \hfill
     \begin{subfigure}[c]{0.16\textwidth}
         \centering
         \includegraphics[width =\textwidth]{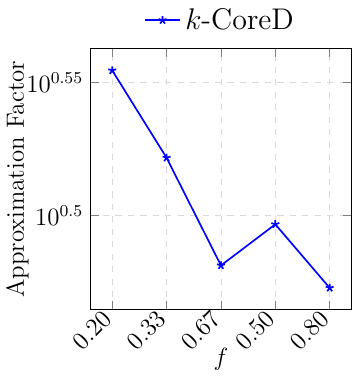}
         \caption{livejournal}
         \label{fig:kcore_avg_approx_liverjournal_fraction}
     \end{subfigure}

        \caption{Avg. approximation factor for $k$-Core decomposition vs. epsilon ($\eps$) and split fraction ($\privatefraction$).}
        \label{fig:kcore_ablation}
\end{figure*}

\eat{\noindent\textbf{Experimental Setup} To evaluate our algorithms in a distributed simulation environment, we partition the input graph across $M$ \textbf{worker processors} and a \textbf{single coordinator processor}. Each worker is responsible for a subset of the nodes and their complete adjacency lists and run LEDP algorithms locally while preserving privacy. The workers communicate their LEDP outputs to the coordinator, which aggregates the data and publishes new public information to all workers. This process is executed over multiple synchronous rounds of communication, simulating a real-world distributed environment. Unless otherwise specified, our experimental setup consists of $\mathbf{80}$ \textbf{worker processors} and \textbf{a single coordinator}. Each reported value is the mean across $\mathbf{5}$ \textbf{runs}, with a \textbf{4-hour} wall-clock limit per run.}
\revision{
\noindent\textbf{Experimental Setup.} To evaluate our algorithms in a distributed simulation, we partition the input graph across $M$ \textbf{worker processors} and a \textbf{single coordinator processor}. Each worker handles a subset of nodes and their full adjacency lists, running LEDP algorithms locally. Workers communicate their privacy-preserving outputs to the coordinator, which aggregates the data and broadcasts new public information. This proceeds over multiple synchronous rounds, simulating a real-world distributed setting. We use $\mathbf{80}$ \textbf{worker processors} and \textbf{a single coordinator}. Each value is averaged over $\mathbf{5}$ \textbf{runs}, with a \textbf{4-hour} wall-clock limit per run.
}


\eat{\noindent\textbf{Parameters} We use the following values for the experiments $\eps = 1.0; \text{bias term} = 8$, the approximation factor $(2+\eta)$ set to be
$5.625$ (the same approximation factor as used in non-private $k$-core decomposition experiments~\cite{LSYDS22}), and privacy split fraction $\privatefraction = 0.8$ where we use $0.8 \cdot \eps$ for thresholding and $0.2 \cdot \eps$ for adding noise to the neighbor count. 
Additionally, we conduct an ablation study to analyze the impact of key parameters, $\eps, \privatefraction$, on the utility of our algorithms. Our theoretical proofs show that our approximation falls within a $(2+\eta)$-multiplicative factor. \eat{, demonstrating how these parameters influence accuracy.}}

\revision{
\noindent\textbf{Parameters} We use $\eps = 1.0$, bias term $=8$, approximation factor $(2+\eta)=5.625$ (matching non-private $k$-core experiments~\cite{LSYDS22}), and privacy split fraction $\privatefraction=0.8$ (allocating $0.8\cdot\eps$ to thresholding and $0.2\cdot\eps$ to level moving step). We also conduct an ablation study on key parameters $\eps$ and $\privatefraction$, and our theoretical proofs show the approximation falls within a $(2+\eta)$-multiplicative factor.
}

\noindent\textbf{Compute Resources} We run experiments on a Google Cloud \texttt{c3-standard-176} instance (3.3 GHz Intel Sapphire Rapids CPUs, 88 physical cores, 704 GiB RAM) with hyper-threading disabled. The code, implemented in Golang~\cite{DK15}, is publicly available~\cite{githubCode}.

\eat{\noindent\textbf{Datasets} We test our algorithms on a diverse set of $13$ real-world
undirected graphs from SNAP~\cite{leskovec2014snap}, the DIMACS Shortest Paths
Challenge road networks~\cite{demetrescu2008implementation}, and the Network
Repository~\cite{nr}, namely
\defn{email-eu-core}, \defn{wiki}, \defn{enron}, \defn{brightkite}, \defn{ego-twitter}, \defn{gplus}, \defn{stanford}, \defn{dblp}, \defn{brain},
\defn{orkut}, \defn{livejournal}, and \defn{friendster}.  We also used \defn{twitter}, a
symmetrized version of the Twitter network~\cite{kwak2010twitter}. \emph{Brain} is a highly dense human brain network
from NeuroData (\hyperlink{https://neurodata.io/}{https://neurodata.io/}).
We remove duplicate edges, zero-degree vertices, and self-loops. \cref{table:kcore}
reflects the graph sizes \emph{after} this removal and gives the exact maximum core number and the number of triangles. 
Exact triangle counts for some graphs are omitted due to time or memory constraints.}

\revision{
\noindent\textbf{Datasets} We test our algorithms on a diverse set of 13 real-world undirected graphs from SNAP~\cite{leskovec2014snap}, the DIMACS Shortest Paths Challenge road networks~\cite{demetrescu2008implementation}, and the Network Repository~\cite{nr}: \defn{email-eu-core}, \defn{wiki}, \defn{enron}, \defn{brightkite}, \defn{ego-twitter}, \defn{gplus}, \defn{stanford}, \defn{dblp}, \defn{orkut}, \defn{livejournal}, and \defn{friendster}. We also use \defn{twitter}, a symmetrized version of the Twitter network~\cite{kwak2010twitter}, and \defn{brain}, a highly dense human brain network from NeuroData (\hyperlink{https://neurodata.io/}{https://neurodata.io/}). We remove duplicate edges, zero-degree vertices, and self-loops. \cref{table:kcore} reflects the graph statistics after this removal. Exact triangle counts for some graphs are omitted due to time or memory constraints.
}

\begin{figure*}[!ht]
     \centering
     
     \begin{subfigure}[c]{0.33\textwidth}
         \centering
         \includegraphics[width =\textwidth]{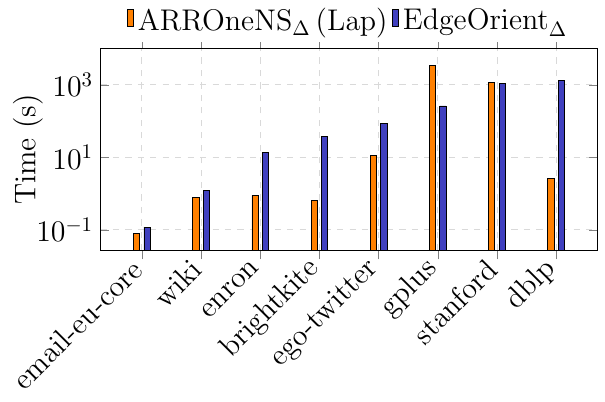}
         \caption{Average Response Time}
         \label{fig:tcount_runtime}
     \end{subfigure}
     \hfill
     \begin{subfigure}[c]{0.33\textwidth}
         \centering
         \includegraphics[width =\textwidth]{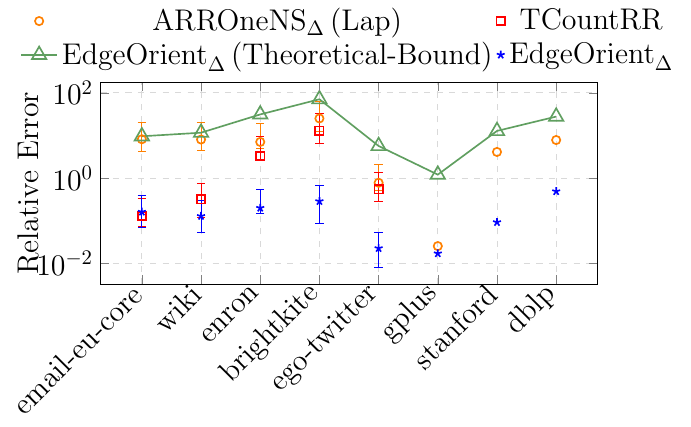}
         \caption{\revision{Relative Error}}
         \label{fig:tcount_rel_error}
     \end{subfigure}
     \hfill
     \begin{subfigure}[c]{0.33\textwidth}
         \centering
    \includegraphics[width=\textwidth]{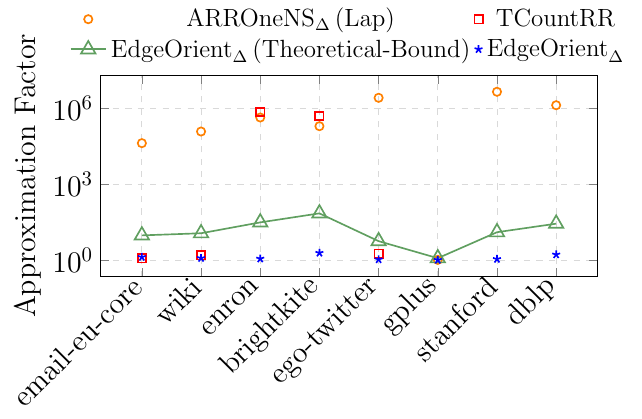}
    \caption{\revision{Average Approximation Factor}}
    \label{fig:tcount_approx}
     \end{subfigure}
     
        \caption{Triangle Counting Results.}
        \label{fig:tcount_all}
\end{figure*}


     


\begin{figure}[!ht]
     \centering
     
     \begin{subfigure}[c]{0.15\textwidth}
         \centering
         \includegraphics[width =\textwidth]{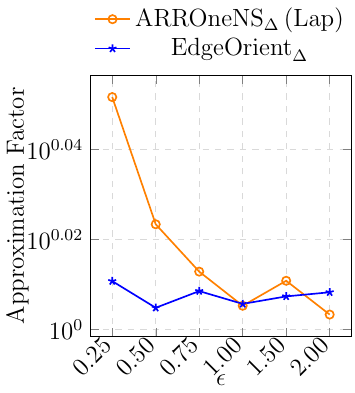}
         \caption{gplus}
         \label{fig:tcount_avg_approx_gplus_eps}
     \end{subfigure}
     \hfill
     \begin{subfigure}[c]{0.15\textwidth}
         \centering
         \includegraphics[width =\textwidth]{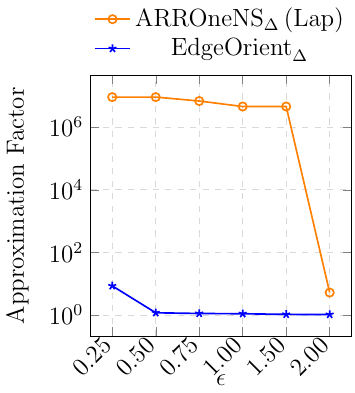}
         \caption{stanford}
         \label{fig:tcount_avg_approx_stanford_eps}
     \end{subfigure}
     \hfill
     \begin{subfigure}[c]{0.15\textwidth}
         \centering
         \includegraphics[width =\textwidth]{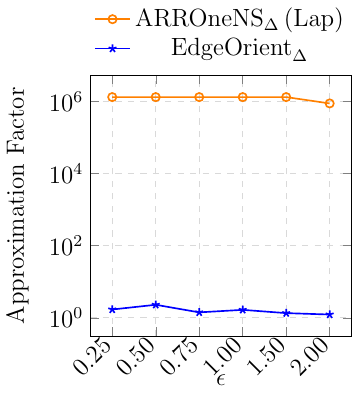}
         \caption{dblp}
         \label{fig:tcount_avg_approx_dblp_eps}
     \end{subfigure}
        \caption{Avg. approx factor for triangle counting vs. epsilon ($\eps$).}
        \label{fig:tcount_eps}
\end{figure}

\subsection{$k$-Core Decomposition}
\label{sec:kcore_exp}

\eat{\noindent\textbf{Response Time} \cref{fig:kcore_rounds} shows the number of communication rounds required by our $k$-core decomposition algorithm, \kcorealgo{}, compared to the baseline, $k$-Core,~\cite{DLRSSY22}. Our approach reduces the number of rounds by \defn{two orders of magnitude}, well within our theoretical bound (without noise) of \(O(\log(n) \cdot \log(D_{\max}))\), demonstrating its scalability and efficiency for large-scale graphs. 

While the comparison of the number of rounds clearly demonstrates the improvements achieved by \kcorealgo{}, a direct comparison of runtime between the two algorithms is not possible. In the baseline algorithm, the absence of bias terms prevents nodes from progressing through the levels in the LDS, effectively resulting in no  computations being performed during each round. Conversely, \kcorealgo{} ensures node movement in every round, thereby increasing the computational workload per round. For this reason, we focus on presenting and analyzing the runtime results of \kcorealgo{}.

\cref{fig:kcore_runtime} shows the response times of \kcorealgo{} across all datasets. Our algorithm efficiently processes large-scale graphs, including billion-edge datasets like \emph{twitter} and \emph{friendster}, within \textbf{four hours}. These results validate the scalability and practicality of \kcorealgo{}, demonstrating the impact of degree thresholding and bias terms in delivering both theoretical and practical improvements.}

\revision{
\noindent\textbf{Response Time.} \cref{fig:kcore_runtime} shows the response times of \kcorealgo{} across all datasets. Our algorithm efficiently processes large-scale graphs, including billion-edge datasets like \emph{twitter} and \emph{friendster}, within \textbf{four hours}. These results validate the scalability and practicality of \kcorealgo{}, demonstrating the impact of degree thresholding and bias terms in delivering both theoretical and practical improvements.

While measuring response time, a direct runtime comparison with the $k$-Core algorithm~\cite{DLRSSY22} is not meaningful. That algorithm does not include bias correction, which often results in large negative noise causing most nodes to remain stuck at level 0. Since the algorithm proceeds for a fixed number of communication rounds (matching the levels in the LDS), but no nodes move beyond the first level, negligible work is performed in subsequent rounds—making the runtime unrealistically low while returning (the same) poor approximation for many nodes. Instead, we compare the number of communication rounds. As shown in~\cref{fig:kcore_rounds}, \kcorealgo{} reduces the number of rounds by \defn{two orders of magnitude} compared to the baseline, aligning with our theoretical bound of $O(\log(n) \cdot \log(D_{\max}))$. This translates to significantly lower communication overhead and improved scalability in distributed settings.
}



\noindent\textbf{Accuracy} We calculate the approximation factor for each node as $a_v = \frac{\max\left(s_v, t_v\right)}{\min\left(s_v, t_v\right)}$, where $s_v$ is the approximate core number and $t_v$ is the true core number. We use this metric to be consistent with the best-known \emph{non-private}
$k$-core decomposition implementations~\cite{LSYDS22,dhulipala2017julienne}. 
These individual node approximation factors facilitate the computation of aggregate metrics: \eat{including}the average, maximum, $80^{th}$, and $95^{th}$ percentile approximation factors for each graph. The theoretical approximation bound in the absence of noise is calculated as $(2+\eta)$, which is $5.625$ for all graphs (labeled
Theoretical Bound). Additionally, we adjust this bound to account for noise by incorporating the additive error term $\frac{\log_{1+\eta/5}^3\left(D_{max}\right)}{\eps}$, where $n$ is the number of nodes, $D_{max}$ is the maximum degree (labeled Theoretical Bound (With Noise)). 
For this 
bound, we compute the effect of the additive noise on the multiplicative factor by adding 
$\frac{\log_{1+\eta/5}^3\left(D_{max}\right)}{\eps \cdot k_{\max}}$, where $k_{\max}$ is the maximum core number,
to $5.625$. Note that such a theoretical bound is a \emph{lower bound} on the effect of the additive 
error on the multiplicative factor; for smaller core numbers, e.g.\ $k_{\min} << k_{\max}$,
the additive error leads to a \emph{much greater} factor. 

\eat{\cref{fig:kcore_avg_approx,fig:kcore_80_approx,fig:kcore_95_approx} present the approximation factors achieved by our approach and the baseline algorithm~\cite{DLRSSY22}, alongside the theoretical bounds across various datasets. On average, our method maintains approximation factors below \(\mathbf{4x}\) across all datasets, with the $80^{\text{th}}$ percentile staying under \(\mathbf{5.5x}\), as illustrated in~\cref{fig:kcore_avg_approx} and~\cref{fig:kcore_80_approx}, demonstrating significantly lower variance compared to the baseline. These results remain well within the theoretical bounds without noise. Compared to the baseline, \kcorealgo{} consistently achieves better or comparable performance. Notably, for graphs such as \textit{brain} and \textit{gplus}, \kcorealgo{} reduces the approximation factors from $131.55$ to $4.11$ and from $52.71$ to $3.27$, improvements of over \textbf{31x} and \textbf{16x}, respectively. Similarly, for \textit{orkut} and \textit{wiki}, our algorithm improves the approximation factors by \textbf{6.6x} and \textbf{1.9x}, respectively.

However, for many graphs, the difference between the baseline and our approach is less pronounced. This is due to the baseline algorithm's inability to move nodes up levels in the LDS, which results in an approximate core number of $2.5$ for most nodes. Given that real-world graphs often exhibit small core numbers~(\cref{table:kcore}), the average approximation factor for the baseline becomes skewed, particularly for graphs with a significant proportion of low-core nodes. In contrast, for graphs with larger core numbers, such as \emph{gplus}, \emph{brain}, and \emph{orkut}, our method demonstrates significant improvements. The advantages of our algorithm become more evident when examining the $80^{\text{th}}$ and $95^{\text{th}}$ percentile approximation factors. As shown in~\cref{fig:kcore_80_approx}, our method consistently achieves a notable reduction in the $80^{\text{th}}$ percentile error compared to the baseline, with reductions of up to \textbf{42x}, \textbf{40x}, and \textbf{7.7x} for graphs with large core numbers, such as \textit{brain}, \textit{gplus}, and \textit{orkut}, respectively. Similarly, in the $95^{\text{th}}$ percentile error (\cref{fig:kcore_95_approx}), our method achieves reductions of nearly \textbf{49x}, \textbf{81x}, and \textbf{7x} for the same graphs. These results underscore the robustness and scalability of our algorithm in handling diverse graph structures with varying core number distributions.}

\revision{
\cref{fig:kcore_avg_approx,fig:kcore_80_approx,fig:kcore_95_approx} present the approximation factors achieved by our approach~(\kcorealgo{}) compared to the baseline $k$-core algorithm ($k$-Core) from~\cite{DLRSSY22}, and the randomized response baseline ($k$-CoreRR), along with theoretical bounds across various datasets. On average, our method maintains approximation factors below \(\mathbf{4x}\) across all datasets, with the $80^{\text{th}}$ percentile staying under \(\mathbf{5.5x}\), as illustrated in~\cref{fig:kcore_avg_approx} and~\cref{fig:kcore_80_approx}, demonstrating significantly lower variance compared to the baselines. These results remain well within the theoretical bounds without noise. Compared to $k$-Core~\cite{DLRSSY22}, \kcorealgo{} consistently achieves better or comparable performance. Notably, for graphs such as \textit{brain} and \textit{gplus}, \kcorealgo{} reduces the approximation factors from $131.55$ to $4.11$ and from $52.71$ to $3.27$, improvements of over \textbf{31x} and \textbf{16x}, respectively. Similarly, for \textit{orkut} and \textit{wiki}, our algorithm improves the approximation factors by \textbf{6.6x} and \textbf{1.9x}, respectively.
Compared to $k$-CoreRR, our algorithm achieves consistently better performance across all datasets—often by \textbf{two to three orders of magnitude}. Notably, on the \textit{dblp} graph, \kcorealgo{} improves the approximation factor by over \textbf{47x}. These results underscore the limitations of naive RR-based methods, which suffer from inflated graph density due to edge perturbation.

However, for many graphs, the difference between $k$-Core~\cite{DLRSSY22} and our approach is less pronounced. This is due to the $k$-Core algorithm's inability to move nodes up levels in the LDS, which results in an approximate core number of $2.5$ for most nodes. Given that real-world graphs often exhibit small core numbers~(\cref{table:kcore}), the average approximation factor for $k$-Core becomes skewed, particularly for graphs with a significant proportion of low-core nodes. In contrast, for graphs with larger core numbers, such as \emph{gplus}, \emph{brain}, and \emph{orkut}, our method demonstrates significant improvements. The advantages of our algorithm become more evident when examining the $80^{\text{th}}$ and $95^{\text{th}}$ percentile approximation factors. As shown in~\cref{fig:kcore_80_approx}, our method consistently achieves a notable reduction in the $80^{\text{th}}$ percentile error compared to $k$-Core, with reductions of up to \textbf{42x}, \textbf{40x}, and \textbf{7.7x} for graphs with large core numbers, such as \textit{brain}, \textit{gplus}, and \textit{orkut}, respectively. Similarly, in the $95^{\text{th}}$ percentile error (\cref{fig:kcore_95_approx}), our method achieves reductions of nearly \textbf{49x}, \textbf{81x}, and \textbf{7x} for the same graphs. These results underscore the robustness and scalability of our algorithm in handling diverse graph structures with varying core number distributions.
}

\noindent\textbf{Ablation Study} We analyze the effect of varying the privacy parameter, $\eps$, and the privacy split fraction, $\privatefraction$, on the utility of the $k$-core decomposition algorithm by plotting the average approximation factor across different datasets: \textit{gplus}, \textit{wiki}, and \textit{livejournal}, for our algorithm and the baseline. From the results shown in~\cref{fig:kcore_ablation}, we observe that the approximation factor improves as $\eps$ increases, which aligns with the theoretical expectations of differential privacy where higher $\eps$ allows for less noise and greater utility.

\eat{When varying the privacy split fraction, $\privatefraction$~(\cref{fig:kcore_ablation}), we note that an optimal value of $0.8$ consistently minimizes the approximation factor across all datasets. 
This is the case since degree thresholding affects the amount of noise that will be added per level for later computations.
This demonstrates that allocating $0.8 \cdot \eps$ for degree thresholding and $0.2 \cdot \eps$ for the level moving step strikes a balance between the two steps, ensuring better overall performance. Consequently, we use $\privatefraction=0.8$ for all other experiments.}

\revision{Further, we note that an optimal value of $0.8$ consistently minimizes the approximation factor across all datasets~(\cref{fig:kcore_ablation}). This is the case since degree thresholding affects the amount of noise that will be added per level for later computations. Thus, we use $\privatefraction=0.8$ for all other experiments, as it strikes a balance between the two steps, ensuring better overall performance.
}
\subsection{Triangle Counting}
\label{sec:tcount_exp}

\noindent\textbf{Response Time} \cref{fig:tcount_runtime} shows the average response times of our LEDP triangle counting algorithm, \tcountalgo{}, implemented in our distributed (LEDP-DS) framework. We compare our algorithm to \textbf{ARROneNS}$\mathbf{_\Delta\left(\textbf{Lap}\right)}$ from~\cite{IMC21communication}. While ARROneNS$_\Delta\left(\text{Lap}\right)$ is implemented in C++, our Golang implementation demonstrates comparable performance across most datasets, with notable speedups for large graphs. Specifically, for \emph{gplus}, our algorithm achieves a \textbf{speedup of $\mathbf{3.45}$x}, while for other graphs such as \emph{email-eu-core, stanford}, and \emph{wiki}, our performance is comparable despite the communication overhead. However, for the \emph{enron, brightkite}, and \emph{dblp} dataset, our algorithm is slower. This discrepancy is likely due to the smaller sizes of the graphs so a centralized algorithm will perform better than a distributed algorithm. 
These results emphasize the scalability and practicality of \tcountalgo{} for large-scale graph analysis, highlighting its ability to handle diverse graphs.

\noindent\textbf{Accuracy} Following the evaluation methodology \eat{outlined} in~\cite{IMC21communication}, we compute relative error for a graph using $\frac{\lvert \widetilde{\Delta} - \Delta \rvert}{\Delta}$, where $\widetilde{\Delta}$ represents the approximated triangle count and $\Delta$ the true triangle count. Additionally, we apply the theoretical bounds from~\cref{thm:triangle-counting} to our analysis. According to \cref{fig:tcount_rel_error}, our algorithm, \tcountalgo{}, consistently achieves relative errors ranging from $\mathbf{10^{-1}}$ to $\mathbf{10^{-2}}$ across all datasets, remaining well within the theoretical bounds.

Compared to ARROneNS$_\Delta\left(\text{Lap}\right)$, our algorithm gives better relative 
errors by \textbf{53x - 89x}, 
for all graphs except \emph{gplus}, 
where we achieve slightly better but comparable accuracy\eat{for \emph{gplus}}. In contrast to GroupRR~\cite{hillebrand2023communication}, which we could only run on the \emph{wiki} dataset due to the $4$-hour timeout limit, our algorithm not only matches accuracy but also has a response time \defn{two orders of magnitude faster}.

\revision{


On small or dense graphs like \textit{wiki} or \textit{email-eu-core}, the TCountRR—which incurs a dominant error term growing proportional to $\frac{n^{3/2}}{\eps^2}$—performs similarly to ours because $n$ is small~\cite{edenICALPS}. However, in larger graphs, $n^{3/2}$ explodes, whereas our error depends instead on $\sqrt{n}\,d$ (with $d\ll n$). TCountRR must also consider all $O(n^4)$ $4$-cycles, but our method only needs $\widetilde{O}(n^2 d^2)$ $4$-cycles. Consequently, on larger datasets such as \emph{enron} and \emph{brightkite}, we reduce the relative error (\cref{fig:tcount_rel_error}) by roughly \textbf{17x} and \textbf{45.5x}, respectively.
Compared to using 60 worker cores, we achieve up to a $1.41$x speedup.

}

To further analyze the limitations of prior approaches, we compute the multiplicative approximation factor of the triangle count, defined as $\frac{\max\left(\widetilde{\Delta}, \Delta\right)}{\max\left(1, \min\left(\widetilde{\Delta}, \Delta\right)\right)}$. Unlike relative error, this metric explicitly accounts for cases where algorithms, such as the one in~\cite{IMC21communication}, produce negative triangle counts, which severely undermines their utility in real-world scenarios. As shown in~\cref{fig:tcount_approx}, our algorithm achieves consistently small approximation factors across all graphs, remaining within $\mathbf{[1.01,1.93]}$, and reducing the factor by \textbf{six orders of magnitude} compared to~\cite{IMC21communication} \revision{and TCountRR}. \revision{This highlights the robustness of our approach, \tcountalgo{}, in maintaining low and stable approximation factors across diverse graph structures.}\eat{, demonstrating superior performance compared to previous methods that fail to deliver consistent utility across different datasets.}


\noindent\textbf{Ablation Study} To evaluate the impact of the privacy parameter \(\eps\) on utility, we analyze and plot the approximation factors for varying values of \(\eps\) across three representative graphs: \textit{gplus}, \textit{dblp}, and \textit{stanford}\revision{, comparing} against those reported in~\cite{IMC21communication}. Our results demonstrate that for \textit{stanford} and \textit{dblp}, our algorithm achieves a significant reduction in approximation factor by up to \textbf{six orders of magnitude}. For \textit{gplus}, while the approximation factors are comparable for higher values of \(\eps\), our algorithm achieves better utility for smaller values of \(\eps\). Specifically, at \(\eps=0.25\), our algorithm achieves an approximation factor of $1.025$, compared to $1.126$ for ARROneNS$_\Delta\left(\text{Lap}\right)$, an improvement of \textbf{9.8\%}. Similarly, at \(\eps=0.50\), our algorithm achieves a factor of $1.011$, compared to $1.055$, an improvement of \textbf{4.2\%}
This highlights the ability of our approach to offer better utility even under stricter privacy constraints, underscoring its advantage over~\cite{IMC21communication}. Additionally, we observe that the utility of our algorithm consistently improves as the privacy parameter \(\eps\) increases, \revision{aligning} with the theoretical expectations of differential privacy, where higher values of \(\eps\) results in less noise.
\section{Conclusion}
\label{sec:conclusion}

\eat{Large-scale network analysis (social, biological, etc.) often involves privacy concerns when analyzing sensitive data. This paper addresses this using local edge differential privacy (LEDP), where nodes protect their private information without sending \eat{the information}it to 
a trusted central authority. We introduce novel LEDP algorithms for key graph statistics, $k$-core decomposition and triangle counting, improving upon prior work, and present the first distributed implementation framework to simulate distributed LEDP algorithms on one machine and evaluate them on real-world data. Experiments show significant improvements, where the $k$-core decomposition approximations are bounded by the theoretical guarantees of non-private algorithms on average and our triangle count errors are reduced by nearly two orders of magnitude over previous work, with comparable runtimes\eat{comparable to existing centralized implementations}. In ongoing work, we are extending this framework to multi-machine testing with network communication, acknowledging additional latency and overhead, and to support a broader range of algorithms. \revision{We are also exploring variants of our algorithms that utilize multiple coordinators to eliminate the current single point of failure and alleviate bottlenecks associated with centralized coordination.} We encourage the community to utilize our open-source framework (available at~\cite{githubCode}\eat{\qq{change this citation to the github link})} for evaluating their LEDP graph algorithms, fostering further advancements in this exciting field.}

\revision{

Large-scale network analysis often raises privacy concerns for sensitive data. We employ local edge differential privacy (LEDP), letting nodes protect their edges without a trusted authority. We propose novel LEDP algorithms for $k$-core decomposition and triangle counting that surpass prior work in accuracy and theoretical guarantees and introduce the first distributed framework to simulate these algorithms on a single machine. Experiments show our $k$-core approximations meet non-private theoretical bounds on average, while triangle counting errors are nearly two orders of magnitude lower than previous LEDP methods, with similar runtimes. Ongoing work extends this framework to multi-machine testing with real communication to capture added latency and explores multi-coordinator variants to eliminate single points of failure and reduce bottlenecks. These enhancements will support a broader range of LEDP algorithms. Our open-source framework (at~\cite{githubCode}) invites the community to build on LEDP graph algorithms.

}








\bibliographystyle{ACM-Reference-Format}
\bibliography{ref}


\begin{thebibliography}{100}


\ifx \showCODEN    \undefined \def \showCODEN     #1{\unskip}     \fi
\ifx \showDOI      \undefined \def \showDOI       #1{#1}\fi
\ifx \showISBNx    \undefined \def \showISBNx     #1{\unskip}     \fi
\ifx \showISBNxiii \undefined \def \showISBNxiii  #1{\unskip}     \fi
\ifx \showISSN     \undefined \def \showISSN      #1{\unskip}     \fi
\ifx \showLCCN     \undefined \def \showLCCN      #1{\unskip}     \fi
\ifx \shownote     \undefined \def \shownote      #1{#1}          \fi
\ifx \showarticletitle \undefined \def \showarticletitle #1{#1}   \fi
\ifx \showURL      \undefined \def \showURL       {\relax}        \fi
\providecommand\bibfield[2]{#2}
\providecommand\bibinfo[2]{#2}
\providecommand\natexlab[1]{#1}
\providecommand\showeprint[2][]{arXiv:#2}

\bibitem[git(2024)]%
        {githubCode}
 \bibinfo{year}{2024}\natexlab{}.
\newblock \bibinfo{title}{DistributedLEDPGraphAlgos}.
\newblock \bibinfo{howpublished}{\url{https://github.com/mundrapranay/DistributedLEDPGraphAlgos}}.
\newblock


\bibitem[Abowd et~al\mbox{.}(2022)]%
        {abowd20222020}
\bibfield{author}{\bibinfo{person}{John~M Abowd}, \bibinfo{person}{Robert Ashmead}, \bibinfo{person}{Ryan Cumings-Menon}, \bibinfo{person}{Simson Garfinkel}, \bibinfo{person}{Micah Heineck}, \bibinfo{person}{Christine Heiss}, \bibinfo{person}{Robert Johns}, \bibinfo{person}{Daniel Kifer}, \bibinfo{person}{Philip Leclerc}, \bibinfo{person}{Ashwin Machanavajjhala}, {et~al\mbox{.}}} \bibinfo{year}{2022}\natexlab{}.
\newblock \showarticletitle{The 2020 census disclosure avoidance system topdown algorithm}.
\newblock \bibinfo{journal}{\emph{Harvard Data Science Review}}  \bibinfo{volume}{2} (\bibinfo{year}{2022}).
\newblock


\bibitem[Ai et~al\mbox{.}(2021)]%
        {doi:10.1142/S012918312150087X}
\bibfield{author}{\bibinfo{person}{Jun Ai}, \bibinfo{person}{Yayun Liu}, \bibinfo{person}{Zhan Su}, \bibinfo{person}{Fengyu Zhao}, {and} \bibinfo{person}{Dunlu Peng}.} \bibinfo{year}{2021}\natexlab{}.
\newblock \showarticletitle{K-core decomposition in recommender systems improves accuracy of rating prediction}.
\newblock \bibinfo{journal}{\emph{International Journal of Modern Physics C}} \bibinfo{volume}{32}, \bibinfo{number}{07} (\bibinfo{year}{2021}), \bibinfo{pages}{2150087}.
\newblock
\urldef\tempurl%
\url{https://doi.org/10.1142/S012918312150087X}
\showDOI{\tempurl}


\bibitem[Al~Hasan and Dave(2018)]%
        {al2018triangle}
\bibfield{author}{\bibinfo{person}{Mohammad Al~Hasan} {and} \bibinfo{person}{Vachik~S Dave}.} \bibinfo{year}{2018}\natexlab{}.
\newblock \showarticletitle{Triangle counting in large networks: a review}.
\newblock \bibinfo{journal}{\emph{Wiley Interdisciplinary Reviews: Data Mining and Knowledge Discovery}} \bibinfo{volume}{8}, \bibinfo{number}{2} (\bibinfo{year}{2018}), \bibinfo{pages}{e1226}.
\newblock


\bibitem[Alon et~al\mbox{.}(1997)]%
        {Alon1997Triangle}
\bibfield{author}{\bibinfo{person}{Noga Alon}, \bibinfo{person}{Raphael Yuster}, {and} \bibinfo{person}{Uri Zwick}.} \bibinfo{year}{1997}\natexlab{}.
\newblock \showarticletitle{Finding and counting given length cycles}.
\newblock \bibinfo{journal}{\emph{Algorithmica}} \bibinfo{volume}{17}, \bibinfo{number}{3} (\bibinfo{year}{1997}), \bibinfo{pages}{209--223}.
\newblock


\bibitem[Balcer and Vadhan(2018)]%
        {BV18}
\bibfield{author}{\bibinfo{person}{Victor Balcer} {and} \bibinfo{person}{Salil~P. Vadhan}.} \bibinfo{year}{2018}\natexlab{}.
\newblock \showarticletitle{Differential Privacy on Finite Computers}. In \bibinfo{booktitle}{\emph{9th Innovations in Theoretical Computer Science Conference (ITCS)}}. \bibinfo{pages}{43:1--43:21}.
\newblock


\bibitem[Bebee et~al\mbox{.}(2018)]%
        {bebee2018amazon}
\bibfield{author}{\bibinfo{person}{Bradley~R Bebee}, \bibinfo{person}{Daniel Choi}, \bibinfo{person}{Ankit Gupta}, \bibinfo{person}{Andi Gutmans}, \bibinfo{person}{Ankesh Khandelwal}, \bibinfo{person}{Yigit Kiran}, \bibinfo{person}{Sainath Mallidi}, \bibinfo{person}{Bruce McGaughy}, \bibinfo{person}{Mike Personick}, \bibinfo{person}{Karthik Rajan}, {et~al\mbox{.}}} \bibinfo{year}{2018}\natexlab{}.
\newblock \showarticletitle{Amazon Neptune: Graph Data Management in the Cloud.}. In \bibinfo{booktitle}{\emph{ISWC (P\&D/Industry/BlueSky)}}.
\newblock


\bibitem[Bishnu et~al\mbox{.}(2025a)]%
        {bishnu2025arboricityrandomedgequeries}
\bibfield{author}{\bibinfo{person}{Arijit Bishnu}, \bibinfo{person}{Debarshi Chanda}, {and} \bibinfo{person}{Gopinath Mishra}.} \bibinfo{year}{2025}\natexlab{a}.
\newblock \bibinfo{title}{Arboricity and Random Edge Queries Matter for Triangle Counting using Sublinear Queries}.
\newblock
\newblock
\showeprint[arxiv]{2502.15379}~[cs.DS]
\urldef\tempurl%
\url{https://arxiv.org/abs/2502.15379}
\showURL{%
\tempurl}


\bibitem[Bishnu et~al\mbox{.}(2025b)]%
        {DBLP:journals/corr/abs-2502-15379}
\bibfield{author}{\bibinfo{person}{Arijit Bishnu}, \bibinfo{person}{Debarshi Chanda}, {and} \bibinfo{person}{Gopinath Mishra}.} \bibinfo{year}{2025}\natexlab{b}.
\newblock \showarticletitle{Arboricity and Random Edge Queries Matter for Triangle Counting using Sublinear Queries}.
\newblock \bibinfo{journal}{\emph{CoRR}}  \bibinfo{volume}{abs/2502.15379} (\bibinfo{date}{February} \bibinfo{year}{2025}).
\newblock
\urldef\tempurl%
\url{https://doi.org/10.48550/arXiv.2502.15379}
\showURL{%
\tempurl}


\bibitem[Bonchi et~al\mbox{.}(2014)]%
        {Bonchi2014Core}
\bibfield{author}{\bibinfo{person}{Francesco Bonchi}, \bibinfo{person}{Aristides Gionis}, {and} \bibinfo{person}{Francesco Gullo}.} \bibinfo{year}{2014}\natexlab{}.
\newblock \showarticletitle{Core decomposition of uncertain graphs}. In \bibinfo{booktitle}{\emph{Proceedings of the 20th ACM SIGKDD International Conference on Knowledge Discovery and Data Mining}}. \bibinfo{pages}{1316--1325}.
\newblock


\bibitem[Brito et~al\mbox{.}(2023)]%
        {brito2023global}
\bibfield{author}{\bibinfo{person}{Felipe~T Brito}, \bibinfo{person}{Victor~AE Farias}, \bibinfo{person}{Cheryl Flynn}, \bibinfo{person}{Subhabrata Majumdar}, \bibinfo{person}{Javam~C Machado}, {and} \bibinfo{person}{Divesh Srivastava}.} \bibinfo{year}{2023}\natexlab{}.
\newblock \showarticletitle{Global and local differentially private release of count-weighted graphs}.
\newblock \bibinfo{journal}{\emph{Proceedings of the ACM on Management of Data}} \bibinfo{volume}{1}, \bibinfo{number}{2} (\bibinfo{year}{2023}), \bibinfo{pages}{1--25}.
\newblock


\bibitem[Bun and Steinke(2016)]%
        {BS16}
\bibfield{author}{\bibinfo{person}{Mark Bun} {and} \bibinfo{person}{Thomas Steinke}.} \bibinfo{year}{2016}\natexlab{}.
\newblock \showarticletitle{Concentrated Differential Privacy: Simplifications, Extensions, and Lower Bounds}. In \bibinfo{booktitle}{\emph{International Conference on Theory of Cryptography}}. \bibinfo{pages}{635--658}.
\newblock


\bibitem[Chan et~al\mbox{.}(2011)]%
        {CSS11}
\bibfield{author}{\bibinfo{person}{T.-H.~Hubert Chan}, \bibinfo{person}{Elaine Shi}, {and} \bibinfo{person}{Dawn Song}.} \bibinfo{year}{2011}\natexlab{}.
\newblock \showarticletitle{Private and Continual Release of Statistics}.
\newblock \bibinfo{journal}{\emph{ACM Trans. Inf. Syst. Secur.}} \bibinfo{volume}{14}, \bibinfo{number}{3}, Article \bibinfo{articleno}{26} (\bibinfo{date}{Nov.} \bibinfo{year}{2011}), \bibinfo{numpages}{24}~pages.
\newblock
\showISSN{1094-9224}
\urldef\tempurl%
\url{https://doi.org/10.1145/2043621.2043626}
\showDOI{\tempurl}


\bibitem[Chang and Ferrara(2022)]%
        {chang2022comparative}
\bibfield{author}{\bibinfo{person}{Ho-Chun~Herbert Chang} {and} \bibinfo{person}{Emilio Ferrara}.} \bibinfo{year}{2022}\natexlab{}.
\newblock \showarticletitle{Comparative analysis of social bots and humans during the COVID-19 pandemic}.
\newblock \bibinfo{journal}{\emph{Journal of Computational Social Science}} \bibinfo{volume}{5}, \bibinfo{number}{2} (\bibinfo{year}{2022}), \bibinfo{pages}{1409--1425}.
\newblock


\bibitem[Ciaperoni et~al\mbox{.}(2020)]%
        {ciaperoni2020relevance}
\bibfield{author}{\bibinfo{person}{Martino Ciaperoni}, \bibinfo{person}{Edoardo Galimberti}, \bibinfo{person}{Francesco Bonchi}, \bibinfo{person}{Ciro Cattuto}, \bibinfo{person}{Francesco Gullo}, {and} \bibinfo{person}{Alain Barrat}.} \bibinfo{year}{2020}\natexlab{}.
\newblock \showarticletitle{Relevance of temporal cores for epidemic spread in temporal networks}.
\newblock \bibinfo{journal}{\emph{Scientific reports}} \bibinfo{volume}{10}, \bibinfo{number}{1} (\bibinfo{year}{2020}), \bibinfo{pages}{12529}.
\newblock


\bibitem[Demetrescu et~al\mbox{.}(2008)]%
        {demetrescu2008implementation}
\bibfield{author}{\bibinfo{person}{Camil Demetrescu}, \bibinfo{person}{Andrew~V Goldberg}, \bibinfo{person}{David~S Johnson}, {et~al\mbox{.}}} \bibinfo{year}{2008}\natexlab{}.
\newblock \showarticletitle{Implementation challenge for shortest paths}.
\newblock In \bibinfo{booktitle}{\emph{Encyclopedia of Algorithms}}. \bibinfo{publisher}{Springer US}, \bibinfo{pages}{395--398}.
\newblock


\bibitem[Dhulipala et~al\mbox{.}(2017)]%
        {dhulipala2017julienne}
\bibfield{author}{\bibinfo{person}{Laxman Dhulipala}, \bibinfo{person}{Guy~E. Blelloch}, {and} \bibinfo{person}{Julian Shun}.} \bibinfo{year}{2017}\natexlab{}.
\newblock \showarticletitle{Julienne: A Framework for Parallel Graph Algorithms Using Work-efficient Bucketing}. In \bibinfo{booktitle}{\emph{ACM Symposium on Parallelism in Algorithms and Architectures (SPAA)}}. \bibinfo{pages}{293--304}.
\newblock


\bibitem[Dhulipala et~al\mbox{.}(2024)]%
        {dhulipala2024nearoptimaldifferentiallyprivatekcore}
\bibfield{author}{\bibinfo{person}{Laxman Dhulipala}, \bibinfo{person}{George~Z. Li}, {and} \bibinfo{person}{Quanquan~C. Liu}.} \bibinfo{year}{2024}\natexlab{}.
\newblock \bibinfo{title}{Near-Optimal Differentially Private k-Core Decomposition}.
\newblock
\newblock
\showeprint[arxiv]{2312.07706}~[cs.DS]
\urldef\tempurl%
\url{https://arxiv.org/abs/2312.07706}
\showURL{%
\tempurl}


\bibitem[Dhulipala et~al\mbox{.}(2022)]%
        {DLRSSY22}
\bibfield{author}{\bibinfo{person}{Laxman Dhulipala}, \bibinfo{person}{Quanquan~C. Liu}, \bibinfo{person}{Sofya Raskhodnikova}, \bibinfo{person}{Jessica Shi}, \bibinfo{person}{Julian Shun}, {and} \bibinfo{person}{Shangdi Yu}.} \bibinfo{year}{2022}\natexlab{}.
\newblock \showarticletitle{Differential Privacy from Locally Adjustable Graph Algorithms: k-Core Decomposition, Low Out-Degree Ordering, and Densest Subgraphs}. In \bibinfo{booktitle}{\emph{63rd {IEEE} Annual Symposium on Foundations of Computer Science, {FOCS} 2022, Denver, CO, USA, October 31 - November 3, 2022}}. \bibinfo{publisher}{{IEEE}}, \bibinfo{pages}{754--765}.
\newblock


\bibitem[Dinitz et~al\mbox{.}(2024)]%
        {dinitz2024tight}
\bibfield{author}{\bibinfo{person}{Michael Dinitz}, \bibinfo{person}{Satyen Kale}, \bibinfo{person}{Silvio Lattanzi}, {and} \bibinfo{person}{Sergei Vassilvitskii}.} \bibinfo{year}{2024}\natexlab{}.
\newblock \bibinfo{title}{Almost Tight Bounds for Differentially Private Densest Subgraph}.
\newblock
\newblock
\showeprint[arxiv]{2308.10316}~[cs.DS]


\bibitem[Donovan and Kernighan(2015)]%
        {DK15}
\bibfield{author}{\bibinfo{person}{Alan~A.A. Donovan} {and} \bibinfo{person}{Brian~W. Kernighan}.} \bibinfo{year}{2015}\natexlab{}.
\newblock \bibinfo{booktitle}{\emph{The Go Programming Language} (\bibinfo{edition}{1st} ed.)}.
\newblock \bibinfo{publisher}{Addison-Wesley Professional}.
\newblock
\showISBNx{0134190440}


\bibitem[Dwork and Lei(2009)]%
        {DL09}
\bibfield{author}{\bibinfo{person}{Cynthia Dwork} {and} \bibinfo{person}{Jing Lei}.} \bibinfo{year}{2009}\natexlab{}.
\newblock \showarticletitle{Differential Privacy and Robust Statistics}. In \bibinfo{booktitle}{\emph{Proceedings of the Forty-First Annual ACM Symposium on Theory of Computing}}. \bibinfo{pages}{371–380}.
\newblock


\bibitem[Dwork et~al\mbox{.}(2006)]%
        {DMNS06}
\bibfield{author}{\bibinfo{person}{Cynthia Dwork}, \bibinfo{person}{Frank McSherry}, \bibinfo{person}{Kobbi Nissim}, {and} \bibinfo{person}{Adam Smith}.} \bibinfo{year}{2006}\natexlab{}.
\newblock \showarticletitle{Calibrating Noise to Sensitivity in Private Data Analysis}. In \bibinfo{booktitle}{\emph{Proceedings of the Third Conference on Theory of Cryptography}}. \bibinfo{pages}{265–284}.
\newblock


\bibitem[Dwork et~al\mbox{.}(2010a)]%
        {DNPR10}
\bibfield{author}{\bibinfo{person}{Cynthia Dwork}, \bibinfo{person}{Moni Naor}, \bibinfo{person}{Toniann Pitassi}, {and} \bibinfo{person}{Guy~N. Rothblum}.} \bibinfo{year}{2010}\natexlab{a}.
\newblock \showarticletitle{Differential Privacy under Continual Observation}. In \bibinfo{booktitle}{\emph{Proceedings of the Forty-Second ACM Symposium on Theory of Computing}}. \bibinfo{pages}{715--724}.
\newblock


\bibitem[Dwork et~al\mbox{.}(2010b)]%
        {DRV10}
\bibfield{author}{\bibinfo{person}{Cynthia Dwork}, \bibinfo{person}{Guy~N. Rothblum}, {and} \bibinfo{person}{Salil Vadhan}.} \bibinfo{year}{2010}\natexlab{b}.
\newblock \showarticletitle{Boosting and Differential Privacy}. In \bibinfo{booktitle}{\emph{Proceedings of the IEEE 51st Annual Symposium on Foundations of Computer Science}}. \bibinfo{pages}{51–60}.
\newblock


\bibitem[Eden et~al\mbox{.}(2023a)]%
        {edenICALPS}
\bibfield{author}{\bibinfo{person}{Talya Eden}, \bibinfo{person}{Quanquan~C. Liu}, \bibinfo{person}{Sofya Raskhodnikova}, {and} \bibinfo{person}{Adam Smith}.} \bibinfo{year}{2023}\natexlab{a}.
\newblock \showarticletitle{{Triangle Counting with Local Edge Differential Privacy}}. In \bibinfo{booktitle}{\emph{50th International Colloquium on Automata, Languages, and Programming (ICALP 2023)}} \emph{(\bibinfo{series}{Leibniz International Proceedings in Informatics (LIPIcs)}, Vol.~\bibinfo{volume}{261})}, \bibfield{editor}{\bibinfo{person}{Kousha Etessami}, \bibinfo{person}{Uriel Feige}, {and} \bibinfo{person}{Gabriele Puppis}} (Eds.). \bibinfo{publisher}{Schloss Dagstuhl -- Leibniz-Zentrum f{\"u}r Informatik}, \bibinfo{address}{Dagstuhl, Germany}, \bibinfo{pages}{52:1--52:21}.
\newblock
\showISBNx{978-3-95977-278-5}
\showISSN{1868-8969}
\urldef\tempurl%
\url{https://doi.org/10.4230/LIPIcs.ICALP.2023.52}
\showDOI{\tempurl}


\bibitem[Eden et~al\mbox{.}(2023b)]%
        {ELRS23}
\bibfield{author}{\bibinfo{person}{Talya Eden}, \bibinfo{person}{Quanquan~C. Liu}, \bibinfo{person}{Sofya Raskhodnikova}, {and} \bibinfo{person}{Adam~D. Smith}.} \bibinfo{year}{2023}\natexlab{b}.
\newblock \showarticletitle{Triangle Counting with Local Edge Differential Privacy}. In \bibinfo{booktitle}{\emph{50th International Colloquium on Automata, Languages, and Programming, {ICALP} 2023, July 10-14, 2023, Paderborn, Germany}} \emph{(\bibinfo{series}{LIPIcs}, Vol.~\bibinfo{volume}{261})}, \bibfield{editor}{\bibinfo{person}{Kousha Etessami}, \bibinfo{person}{Uriel Feige}, {and} \bibinfo{person}{Gabriele Puppis}} (Eds.). \bibinfo{publisher}{Schloss Dagstuhl - Leibniz-Zentrum f{\"{u}}r Informatik}, \bibinfo{pages}{52:1--52:21}.
\newblock
\urldef\tempurl%
\url{https://doi.org/10.4230/LIPICS.ICALP.2023.52}
\showDOI{\tempurl}


\bibitem[Evfimievski et~al\mbox{.}(2003)]%
        {evfimievski2003limiting}
\bibfield{author}{\bibinfo{person}{Alexandre Evfimievski}, \bibinfo{person}{Johannes Gehrke}, {and} \bibinfo{person}{Ramakrishnan Srikant}.} \bibinfo{year}{2003}\natexlab{}.
\newblock \showarticletitle{Limiting privacy breaches in privacy preserving data mining}. In \bibinfo{booktitle}{\emph{Proceedings of the twenty-second ACM SIGMOD-SIGACT-SIGART symposium on Principles of database systems}}. \bibinfo{pages}{211--222}.
\newblock


\bibitem[Farhadi et~al\mbox{.}(2021)]%
        {AHS21}
\bibfield{author}{\bibinfo{person}{Alireza Farhadi}, \bibinfo{person}{MohammadTaghi Hajiaghayi}, {and} \bibinfo{person}{Elaine Shi}.} \bibinfo{year}{2021}\natexlab{}.
\newblock \showarticletitle{Differentially Private Densest Subgraph}.
\newblock \bibinfo{journal}{\emph{CoRR}}  \bibinfo{volume}{abs/2106.00508} (\bibinfo{year}{2021}), \bibinfo{pages}{11581--11597}.
\newblock
\showeprint[arxiv]{2106.00508}
\urldef\tempurl%
\url{https://arxiv.org/abs/2106.00508}
\showURL{%
\tempurl}


\bibitem[Farias et~al\mbox{.}(2020)]%
        {farias2020local}
\bibfield{author}{\bibinfo{person}{Victor~AE Farias}, \bibinfo{person}{Felipe~T Brito}, \bibinfo{person}{Cheryl Flynn}, \bibinfo{person}{Javam~C Machado}, \bibinfo{person}{Subhabrata Majumdar}, {and} \bibinfo{person}{Divesh Srivastava}.} \bibinfo{year}{2020}\natexlab{}.
\newblock \showarticletitle{Local dampening: Differential privacy for non-numeric queries via local sensitivity}.
\newblock \bibinfo{journal}{\emph{arXiv preprint arXiv:2012.04117}} (\bibinfo{year}{2020}).
\newblock


\bibitem[Fu et~al\mbox{.}(2023)]%
        {fu2023gc}
\bibfield{author}{\bibinfo{person}{Nan Fu}, \bibinfo{person}{Weiwei Ni}, \bibinfo{person}{Sen Zhang}, \bibinfo{person}{Lihe Hou}, {and} \bibinfo{person}{Dongyue Zhang}.} \bibinfo{year}{2023}\natexlab{}.
\newblock \showarticletitle{GC-NLDP: A graph clustering algorithm with local differential privacy}.
\newblock \bibinfo{journal}{\emph{Computers \& Security}}  \bibinfo{volume}{124} (\bibinfo{year}{2023}), \bibinfo{pages}{102967}.
\newblock


\bibitem[Fujimoto et~al\mbox{.}(2022)]%
        {fujimoto2022integrated}
\bibfield{author}{\bibinfo{person}{Kayo Fujimoto}, \bibinfo{person}{Dimitrios Paraskevis}, \bibinfo{person}{Jacky~C Kuo}, \bibinfo{person}{Camden~J Hallmark}, \bibinfo{person}{Jing Zhao}, \bibinfo{person}{Andre Hochi}, \bibinfo{person}{Lisa~M Kuhns}, \bibinfo{person}{Lu-Yu Hwang}, \bibinfo{person}{Angelos Hatzakis}, {and} \bibinfo{person}{John~A Schneider}.} \bibinfo{year}{2022}\natexlab{}.
\newblock \showarticletitle{Integrated molecular and affiliation network analysis: Core-periphery social clustering is associated with HIV transmission patterns}.
\newblock \bibinfo{journal}{\emph{Social networks}}  \bibinfo{volume}{68} (\bibinfo{year}{2022}), \bibinfo{pages}{107--117}.
\newblock


\bibitem[Giatsidis et~al\mbox{.}(2016)]%
        {giatsidis2016k}
\bibfield{author}{\bibinfo{person}{Christos Giatsidis}, \bibinfo{person}{Fragkiskos~D Malliaros}, \bibinfo{person}{Nikolaos Tziortziotis}, \bibinfo{person}{Charanpal Dhanjal}, \bibinfo{person}{Emmanouil Kiagias}, \bibinfo{person}{Dimitrios~M Thilikos}, {and} \bibinfo{person}{Michalis Vazirgiannis}.} \bibinfo{year}{2016}\natexlab{}.
\newblock \showarticletitle{A k-core decomposition framework for graph clustering}.
\newblock \bibinfo{journal}{\emph{arXiv preprint arXiv:1607.02096}} (\bibinfo{year}{2016}).
\newblock


\bibitem[Guo et~al\mbox{.}(2023)]%
        {guo2023community}
\bibfield{author}{\bibinfo{person}{Taolin Guo}, \bibinfo{person}{Shunshun Peng}, \bibinfo{person}{Yong Li}, \bibinfo{person}{Mingliang Zhou}, {and} \bibinfo{person}{Trieu-Kien Truong}.} \bibinfo{year}{2023}\natexlab{}.
\newblock \showarticletitle{Community-based social recommendation under local differential privacy protection}.
\newblock \bibinfo{journal}{\emph{Information Sciences}}  \bibinfo{volume}{639} (\bibinfo{year}{2023}), \bibinfo{pages}{119002}.
\newblock


\bibitem[Guo et~al\mbox{.}(2022)]%
        {guo2022integrative}
\bibfield{author}{\bibinfo{person}{Yang Guo}, \bibinfo{person}{Fatemeh Esfahani}, \bibinfo{person}{Xiaojian Shao}, \bibinfo{person}{Venkatesh Srinivasan}, \bibinfo{person}{Alex Thomo}, \bibinfo{person}{Li Xing}, {and} \bibinfo{person}{Xuekui Zhang}.} \bibinfo{year}{2022}\natexlab{}.
\newblock \showarticletitle{Integrative COVID-19 biological network inference with probabilistic core decomposition}.
\newblock \bibinfo{journal}{\emph{Briefings in Bioinformatics}} \bibinfo{volume}{23}, \bibinfo{number}{1} (\bibinfo{year}{2022}), \bibinfo{pages}{bbab455}.
\newblock


\bibitem[Hehir et~al\mbox{.}(2022)]%
        {hehir2022consistent}
\bibfield{author}{\bibinfo{person}{Jonathan Hehir}, \bibinfo{person}{Aleksandra Slavkovi{\'c}}, {and} \bibinfo{person}{Xiaoyue Niu}.} \bibinfo{year}{2022}\natexlab{}.
\newblock \showarticletitle{Consistent spectral clustering of network block models under local differential privacy}.
\newblock \bibinfo{journal}{\emph{The Journal of privacy and confidentiality}} \bibinfo{volume}{12}, \bibinfo{number}{2} (\bibinfo{year}{2022}).
\newblock


\bibitem[Henzinger et~al\mbox{.}(2024)]%
        {HSZ24}
\bibfield{author}{\bibinfo{person}{Monika Henzinger}, \bibinfo{person}{A.~R. Sricharan}, {and} \bibinfo{person}{Leqi Zhu}.} \bibinfo{year}{2024}\natexlab{}.
\newblock \showarticletitle{Tighter Bounds for Local Differentially Private Core Decomposition and Densest Subgraph}.
\newblock \bibinfo{journal}{\emph{CoRR}}  \bibinfo{volume}{abs/2402.18020} (\bibinfo{year}{2024}).
\newblock
\urldef\tempurl%
\url{https://doi.org/10.48550/ARXIV.2402.18020}
\showDOI{\tempurl}
\showeprint[arXiv]{2402.18020}


\bibitem[Hidano and Murakami(2022)]%
        {hidano2022degree}
\bibfield{author}{\bibinfo{person}{Seira Hidano} {and} \bibinfo{person}{Takao Murakami}.} \bibinfo{year}{2022}\natexlab{}.
\newblock \showarticletitle{Degree-preserving randomized response for graph neural networks under local differential privacy}.
\newblock \bibinfo{journal}{\emph{arXiv preprint arXiv:2202.10209}} (\bibinfo{year}{2022}).
\newblock


\bibitem[Hillebrand et~al\mbox{.}(2023a)]%
        {hillebrand2023communication}
\bibfield{author}{\bibinfo{person}{Quentin Hillebrand}, \bibinfo{person}{Vorapong Suppakitpaisarn}, {and} \bibinfo{person}{Tetsuo Shibuya}.} \bibinfo{year}{2023}\natexlab{a}.
\newblock \showarticletitle{Communication Cost Reduction for Subgraph Counting under Local Differential Privacy via Hash Functions}.
\newblock \bibinfo{journal}{\emph{arXiv preprint arXiv:2312.07055}} (\bibinfo{year}{2023}).
\newblock


\bibitem[Hillebrand et~al\mbox{.}(2023b)]%
        {hillebrand2023unbiased}
\bibfield{author}{\bibinfo{person}{Quentin Hillebrand}, \bibinfo{person}{Vorapong Suppakitpaisarn}, {and} \bibinfo{person}{Tetsuo Shibuya}.} \bibinfo{year}{2023}\natexlab{b}.
\newblock \showarticletitle{Unbiased locally private estimator for polynomials of laplacian variables}. In \bibinfo{booktitle}{\emph{Proceedings of the 29th ACM SIGKDD Conference on Knowledge Discovery and Data Mining}}. \bibinfo{pages}{741--751}.
\newblock


\bibitem[Holme and Litvak(2017)]%
        {holme2017cost}
\bibfield{author}{\bibinfo{person}{Petter Holme} {and} \bibinfo{person}{Nelly Litvak}.} \bibinfo{year}{2017}\natexlab{}.
\newblock \showarticletitle{Cost-efficient vaccination protocols for network epidemiology}.
\newblock \bibinfo{journal}{\emph{PLoS computational biology}} \bibinfo{volume}{13}, \bibinfo{number}{9} (\bibinfo{year}{2017}), \bibinfo{pages}{e1005696}.
\newblock


\bibitem[Imola et~al\mbox{.}(2023)]%
        {imola2023differentially}
\bibfield{author}{\bibinfo{person}{Jacob Imola}, \bibinfo{person}{Alessandro Epasto}, \bibinfo{person}{Mohammad Mahdian}, \bibinfo{person}{Vincent Cohen-Addad}, {and} \bibinfo{person}{Vahab Mirrokni}.} \bibinfo{year}{2023}\natexlab{}.
\newblock \showarticletitle{Differentially private hierarchical clustering with provable approximation guarantees}. In \bibinfo{booktitle}{\emph{International Conference on Machine Learning}}. PMLR, \bibinfo{pages}{14353--14375}.
\newblock


\bibitem[Imola et~al\mbox{.}(2021)]%
        {IMC21locally}
\bibfield{author}{\bibinfo{person}{Jacob Imola}, \bibinfo{person}{Takao Murakami}, {and} \bibinfo{person}{Kamalika Chaudhuri}.} \bibinfo{year}{2021}\natexlab{}.
\newblock \showarticletitle{Locally Differentially Private Analysis of Graph Statistics}. In \bibinfo{booktitle}{\emph{30th {USENIX} Security Symposium}}. \bibinfo{pages}{983--1000}.
\newblock


\bibitem[Imola et~al\mbox{.}(2022)]%
        {IMC21communication}
\bibfield{author}{\bibinfo{person}{Jacob Imola}, \bibinfo{person}{Takao Murakami}, {and} \bibinfo{person}{Kamalika Chaudhuri}.} \bibinfo{year}{2022}\natexlab{}.
\newblock \showarticletitle{{Communication-Efficient} Triangle Counting under Local Differential Privacy}. In \bibinfo{booktitle}{\emph{31st USENIX Security Symposium}}. \bibinfo{pages}{537--554}.
\newblock


\bibitem[Jiang et~al\mbox{.}(2023)]%
        {jiang2023personalized}
\bibfield{author}{\bibinfo{person}{Linyu Jiang}, \bibinfo{person}{Yukun Yan}, \bibinfo{person}{Zhihong Tian}, \bibinfo{person}{Zuobin Xiong}, {and} \bibinfo{person}{Qilong Han}.} \bibinfo{year}{2023}\natexlab{}.
\newblock \showarticletitle{Personalized sampling graph collection with local differential privacy for link prediction}.
\newblock \bibinfo{journal}{\emph{World Wide Web}} \bibinfo{volume}{26}, \bibinfo{number}{5} (\bibinfo{year}{2023}), \bibinfo{pages}{2669--2689}.
\newblock


\bibitem[Kasiviswanathan et~al\mbox{.}(2011)]%
        {kasiviswanathan2011can}
\bibfield{author}{\bibinfo{person}{Shiva~Prasad Kasiviswanathan}, \bibinfo{person}{Homin~K Lee}, \bibinfo{person}{Kobbi Nissim}, \bibinfo{person}{Sofya Raskhodnikova}, {and} \bibinfo{person}{Adam Smith}.} \bibinfo{year}{2011}\natexlab{}.
\newblock \showarticletitle{What can we learn privately?}
\newblock \bibinfo{journal}{\emph{SIAM J. Comput.}} \bibinfo{volume}{40}, \bibinfo{number}{3} (\bibinfo{year}{2011}), \bibinfo{pages}{793--826}.
\newblock


\bibitem[Kim et~al\mbox{.}(2021)]%
        {KGS21}
\bibfield{author}{\bibinfo{person}{Muah Kim}, \bibinfo{person}{Onur Günlü}, {and} \bibinfo{person}{Rafael~F. Schaefer}.} \bibinfo{year}{2021}\natexlab{}.
\newblock \showarticletitle{Federated Learning with Local Differential Privacy: Trade-Offs Between Privacy, Utility, and Communication}. In \bibinfo{booktitle}{\emph{ICASSP 2021 - 2021 IEEE International Conference on Acoustics, Speech and Signal Processing (ICASSP)}}. \bibinfo{pages}{2650--2654}.
\newblock
\urldef\tempurl%
\url{https://doi.org/10.1109/ICASSP39728.2021.9413764}
\showDOI{\tempurl}


\bibitem[Kitsak et~al\mbox{.}(2010)]%
        {kitsak2010}
\bibfield{author}{\bibinfo{person}{Maksim Kitsak}, \bibinfo{person}{Lazaros~K. Gallos}, \bibinfo{person}{Shlomo Havlin}, \bibinfo{person}{Fredrik Liljeros}, \bibinfo{person}{Lev Muchnik}, \bibinfo{person}{H.~Eugene Stanley}, {and} \bibinfo{person}{Hernán~A. Makse}.} \bibinfo{year}{2010}\natexlab{}.
\newblock \showarticletitle{Identification of influential spreaders in complex networks}.
\newblock \bibinfo{journal}{\emph{Nature Physics}} \bibinfo{volume}{6}, \bibinfo{number}{11} (\bibinfo{date}{Nov.} \bibinfo{year}{2010}), \bibinfo{pages}{888--893}.
\newblock
\showISSN{1745-2481}
\urldef\tempurl%
\url{https://doi.org/10.1038/nphys1746}
\showDOI{\tempurl}


\bibitem[Kulkarni(2019)]%
        {10.1145/3299869.3300102}
\bibfield{author}{\bibinfo{person}{Tejas Kulkarni}.} \bibinfo{year}{2019}\natexlab{}.
\newblock \showarticletitle{Answering Range Queries Under Local Differential Privacy}. In \bibinfo{booktitle}{\emph{Proceedings of the 2019 International Conference on Management of Data}} (Amsterdam, Netherlands) \emph{(\bibinfo{series}{SIGMOD '19})}. \bibinfo{publisher}{Association for Computing Machinery}, \bibinfo{address}{New York, NY, USA}, \bibinfo{pages}{1832–1834}.
\newblock
\showISBNx{9781450356435}
\urldef\tempurl%
\url{https://doi.org/10.1145/3299869.3300102}
\showDOI{\tempurl}


\bibitem[Kwak et~al\mbox{.}(2010)]%
        {kwak2010twitter}
\bibfield{author}{\bibinfo{person}{Haewoon Kwak}, \bibinfo{person}{Changhyun Lee}, \bibinfo{person}{Hosung Park}, {and} \bibinfo{person}{Sue Moon}.} \bibinfo{year}{2010}\natexlab{}.
\newblock \showarticletitle{What is {Twitter}, a Social Network or a News Media?}. In \bibinfo{booktitle}{\emph{www}}. \bibinfo{pages}{591--600}.
\newblock


\bibitem[Leskovec and Krevl(2014)]%
        {leskovec2014snap}
\bibfield{author}{\bibinfo{person}{Jure Leskovec} {and} \bibinfo{person}{Andrej Krevl}.} \bibinfo{year}{2014}\natexlab{}.
\newblock \showarticletitle{{SNAP Datasets}: {Stanford} Large Network Dataset Collection}.
\newblock  (\bibinfo{year}{2014}).
\newblock


\bibitem[Li et~al\mbox{.}(2023)]%
        {li2023fine}
\bibfield{author}{\bibinfo{person}{Xiaoguang Li}, \bibinfo{person}{Ninghui Li}, \bibinfo{person}{Wenhai Sun}, \bibinfo{person}{Neil~Zhenqiang Gong}, {and} \bibinfo{person}{Hui Li}.} \bibinfo{year}{2023}\natexlab{}.
\newblock \showarticletitle{Fine-grained poisoning attack to local differential privacy protocols for mean and variance estimation}. In \bibinfo{booktitle}{\emph{32nd USENIX Security Symposium (USENIX Security 23)}}. \bibinfo{pages}{1739--1756}.
\newblock


\bibitem[Lin et~al\mbox{.}(2022)]%
        {lin2022towards}
\bibfield{author}{\bibinfo{person}{Wanyu Lin}, \bibinfo{person}{Baochun Li}, {and} \bibinfo{person}{Cong Wang}.} \bibinfo{year}{2022}\natexlab{}.
\newblock \showarticletitle{Towards private learning on decentralized graphs with local differential privacy}.
\newblock \bibinfo{journal}{\emph{IEEE Transactions on Information Forensics and Security}}  \bibinfo{volume}{17} (\bibinfo{year}{2022}), \bibinfo{pages}{2936--2946}.
\newblock


\bibitem[Liu et~al\mbox{.}(2023b)]%
        {liu2023some}
\bibfield{author}{\bibinfo{person}{Fang Liu}, \bibinfo{person}{Dong Wang}, {and} \bibinfo{person}{Tian Yan}.} \bibinfo{year}{2023}\natexlab{b}.
\newblock \showarticletitle{Some examples of privacy-preserving sharing of COVID-19 pandemic data with statistical utility evaluation}.
\newblock \bibinfo{journal}{\emph{BMC Medical Research Methodology}} \bibinfo{volume}{23}, \bibinfo{number}{1} (\bibinfo{year}{2023}), \bibinfo{pages}{120}.
\newblock


\bibitem[Liu et~al\mbox{.}(2022b)]%
        {LSYDS22}
\bibfield{author}{\bibinfo{person}{Quanquan~C. Liu}, \bibinfo{person}{Jessica Shi}, \bibinfo{person}{Shangdi Yu}, \bibinfo{person}{Laxman Dhulipala}, {and} \bibinfo{person}{Julian Shun}.} \bibinfo{year}{2022}\natexlab{b}.
\newblock \showarticletitle{Parallel Batch-Dynamic Algorithms for $k$-Core Decomposition and Related Graph Problems}. In \bibinfo{booktitle}{\emph{34th {ACM} Symposium on Parallelism in Algorithms and Architectures}}. \bibinfo{pages}{191--204}.
\newblock


\bibitem[Liu et~al\mbox{.}(2023a)]%
        {liu2023cargo}
\bibfield{author}{\bibinfo{person}{Shang Liu}, \bibinfo{person}{Yang Cao}, \bibinfo{person}{Takao Murakami}, \bibinfo{person}{Jinfei Liu}, {and} \bibinfo{person}{Masatoshi Yoshikawa}.} \bibinfo{year}{2023}\natexlab{a}.
\newblock \showarticletitle{CARGO: Crypto-Assisted Differentially Private Triangle Counting without Trusted Servers}.
\newblock \bibinfo{journal}{\emph{arXiv preprint arXiv:2312.12938}} (\bibinfo{year}{2023}).
\newblock


\bibitem[Liu et~al\mbox{.}(2022a)]%
        {liu2022crypto}
\bibfield{author}{\bibinfo{person}{Shang Liu}, \bibinfo{person}{Yang Cao}, \bibinfo{person}{Takao Murakami}, {and} \bibinfo{person}{Masatoshi Yoshikawa}.} \bibinfo{year}{2022}\natexlab{a}.
\newblock \showarticletitle{A crypto-assisted approach for publishing graph statistics with node local differential privacy}. In \bibinfo{booktitle}{\emph{2022 IEEE International Conference on Big Data (Big Data)}}. IEEE, \bibinfo{pages}{5765--5774}.
\newblock


\bibitem[Liu et~al\mbox{.}(2024)]%
        {liu2024edge}
\bibfield{author}{\bibinfo{person}{Yuhan Liu}, \bibinfo{person}{Tianhao Wang}, \bibinfo{person}{Yixuan Liu}, \bibinfo{person}{Hong Chen}, {and} \bibinfo{person}{Cuiping Li}.} \bibinfo{year}{2024}\natexlab{}.
\newblock \showarticletitle{Edge-Protected Triangle Count Estimation under Relationship Local Differential Privacy}.
\newblock \bibinfo{journal}{\emph{IEEE Transactions on Knowledge and Data Engineering}} (\bibinfo{year}{2024}).
\newblock


\bibitem[Liu et~al\mbox{.}(2022c)]%
        {liu2022collecting}
\bibfield{author}{\bibinfo{person}{Yuhan Liu}, \bibinfo{person}{Suyun Zhao}, \bibinfo{person}{Yixuan Liu}, \bibinfo{person}{Dan Zhao}, \bibinfo{person}{Hong Chen}, {and} \bibinfo{person}{Cuiping Li}.} \bibinfo{year}{2022}\natexlab{c}.
\newblock \showarticletitle{Collecting triangle counts with edge relationship local differential privacy}. In \bibinfo{booktitle}{\emph{2022 IEEE 38th International Conference on Data Engineering (ICDE)}}. IEEE, \bibinfo{pages}{2008--2020}.
\newblock


\bibitem[Liu et~al\mbox{.}(2018)]%
        {LZZY18}
\bibfield{author}{\bibinfo{person}{Zemin Liu}, \bibinfo{person}{Vincent~W. Zheng}, \bibinfo{person}{Zhou Zhao}, \bibinfo{person}{Hongxia Yang}, \bibinfo{person}{Kevin Chen-Chuan Chang}, \bibinfo{person}{Minghui Wu}, {and} \bibinfo{person}{Jing Ying}.} \bibinfo{year}{2018}\natexlab{}.
\newblock \showarticletitle{Subgraph-Augmented Path Embedding for Semantic User Search on Heterogeneous Social Network}. In \bibinfo{booktitle}{\emph{Proceedings of the 2018 World Wide Web Conference}} (Lyon, France) \emph{(\bibinfo{series}{WWW '18})}. \bibinfo{publisher}{International World Wide Web Conferences Steering Committee}, \bibinfo{address}{Republic and Canton of Geneva, CHE}, \bibinfo{pages}{1613–1622}.
\newblock
\showISBNx{9781450356398}
\urldef\tempurl%
\url{https://doi.org/10.1145/3178876.3186073}
\showDOI{\tempurl}


\bibitem[Mahawaga~Arachchige et~al\mbox{.}(2022)]%
        {mahawaga2022local}
\bibfield{author}{\bibinfo{person}{Pathum~Chamikara Mahawaga~Arachchige}, \bibinfo{person}{Dongxi Liu}, \bibinfo{person}{Seyit Camtepe}, \bibinfo{person}{Surya Nepal}, \bibinfo{person}{Marthie Grobler}, \bibinfo{person}{Peter Bertok}, {and} \bibinfo{person}{Ibrahim Khalil}.} \bibinfo{year}{2022}\natexlab{}.
\newblock \showarticletitle{Local Differential Privacy for Federated Learning}. In \bibinfo{booktitle}{\emph{European Symposium on Research in Computer Security}}. Springer, \bibinfo{pages}{195--216}.
\newblock


\bibitem[Masuda et~al\mbox{.}(2018)]%
        {MSEW18}
\bibfield{author}{\bibinfo{person}{Naoki Masuda}, \bibinfo{person}{Michiko Sakaki}, \bibinfo{person}{Takahiro Ezaki}, {and} \bibinfo{person}{Takamitsu Watanabe}.} \bibinfo{year}{2018}\natexlab{}.
\newblock \showarticletitle{Clustering Coefficients for Correlation Networks}.
\newblock \bibinfo{journal}{\emph{Frontiers in Neuroinformatics}}  \bibinfo{volume}{12} (\bibinfo{year}{2018}).
\newblock
\showISSN{1662-5196}
\urldef\tempurl%
\url{https://doi.org/10.3389/fninf.2018.00007}
\showDOI{\tempurl}


\bibitem[Mei et~al\mbox{.}(2021)]%
        {MTXP21}
\bibfield{author}{\bibinfo{person}{Gang Mei}, \bibinfo{person}{Jingzhi Tu}, \bibinfo{person}{Lei Xiao}, {and} \bibinfo{person}{Francesco Piccialli}.} \bibinfo{year}{2021}\natexlab{}.
\newblock \showarticletitle{An efficient graph clustering algorithm by exploiting k-core decomposition and motifs}.
\newblock \bibinfo{journal}{\emph{Computers \& Electrical Engineering}}  \bibinfo{volume}{96} (\bibinfo{year}{2021}), \bibinfo{pages}{107564}.
\newblock
\showISSN{0045-7906}
\urldef\tempurl%
\url{https://doi.org/10.1016/j.compeleceng.2021.107564}
\showDOI{\tempurl}


\bibitem[Mueller et~al\mbox{.}(2022)]%
        {mueller2022sok}
\bibfield{author}{\bibinfo{person}{Tamara~T Mueller}, \bibinfo{person}{Dmitrii Usynin}, \bibinfo{person}{Johannes~C Paetzold}, \bibinfo{person}{Daniel Rueckert}, {and} \bibinfo{person}{Georgios Kaissis}.} \bibinfo{year}{2022}\natexlab{}.
\newblock \showarticletitle{SoK: Differential privacy on graph-structured data}.
\newblock \bibinfo{journal}{\emph{arXiv preprint arXiv:2203.09205}} (\bibinfo{year}{2022}).
\newblock


\bibitem[Murakami and Sei(2023)]%
        {murakami2023automatic}
\bibfield{author}{\bibinfo{person}{Takao Murakami} {and} \bibinfo{person}{Yuichi Sei}.} \bibinfo{year}{2023}\natexlab{}.
\newblock \showarticletitle{Automatic Tuning of Privacy Budgets in Input-Discriminative Local Differential Privacy}.
\newblock \bibinfo{journal}{\emph{IEEE Internet of Things Journal}} (\bibinfo{year}{2023}).
\newblock


\bibitem[Naseri et~al\mbox{.}(2022)]%
        {NHC22}
\bibfield{author}{\bibinfo{person}{Mohammad Naseri}, \bibinfo{person}{Jamie Hayes}, {and} \bibinfo{person}{Emiliano~De Cristofaro}.} \bibinfo{year}{2022}\natexlab{}.
\newblock \showarticletitle{Local and Central Differential Privacy for Robustness and Privacy in Federated Learning}. In \bibinfo{booktitle}{\emph{29th Annual Network and Distributed System Security Symposium, {NDSS} 2022, San Diego, California, USA, April 24-28, 2022}}. \bibinfo{publisher}{The Internet Society}.
\newblock
\urldef\tempurl%
\url{https://www.ndss-symposium.org/ndss-paper/auto-draft-204/}
\showURL{%
\tempurl}


\bibitem[Neo4j(2012)]%
        {noauthororeditorneo4j}
\bibfield{author}{\bibinfo{person}{Neo4j}.} \bibinfo{year}{2012}\natexlab{}.
\newblock \bibinfo{title}{Neo4j - The World’s Leading Graph Database}.
\newblock
\newblock
\urldef\tempurl%
\url{http://neo4j.org/}
\showURL{%
\tempurl}


\bibitem[Nguyen et~al\mbox{.}(2024)]%
        {nguyen2024faster}
\bibfield{author}{\bibinfo{person}{Dung Nguyen}, \bibinfo{person}{Mahantesh Halappanavar}, \bibinfo{person}{Venkatesh Srinivasan}, {and} \bibinfo{person}{Anil Vullikanti}.} \bibinfo{year}{2024}\natexlab{}.
\newblock \showarticletitle{Faster approximate subgraph counts with privacy}.
\newblock \bibinfo{journal}{\emph{Advances in Neural Information Processing Systems}}  \bibinfo{volume}{36} (\bibinfo{year}{2024}).
\newblock


\bibitem[Nissim et~al\mbox{.}(2007)]%
        {NRS07}
\bibfield{author}{\bibinfo{person}{Kobbi Nissim}, \bibinfo{person}{Sofya Raskhodnikova}, {and} \bibinfo{person}{Adam Smith}.} \bibinfo{year}{2007}\natexlab{}.
\newblock \showarticletitle{Smooth Sensitivity and Sampling in Private Data Analysis}. In \bibinfo{booktitle}{\emph{Proceedings of the Thirty-Ninth Annual ACM Symposium on Theory of Computing}}. \bibinfo{pages}{75–84}.
\newblock


\bibitem[Palla et~al\mbox{.}(2005)]%
        {palla2005uncovering}
\bibfield{author}{\bibinfo{person}{Gergely Palla}, \bibinfo{person}{Imre Der{\'e}nyi}, \bibinfo{person}{Ill{\'e}s Farkas}, {and} \bibinfo{person}{Tam{\'a}s Vicsek}.} \bibinfo{year}{2005}\natexlab{}.
\newblock \showarticletitle{Uncovering the overlapping community structure of complex networks in nature and society}.
\newblock \bibinfo{journal}{\emph{Nature}} \bibinfo{volume}{435}, \bibinfo{number}{7043} (\bibinfo{year}{2005}), \bibinfo{pages}{814--818}.
\newblock


\bibitem[Prat-P{\'e}rez et~al\mbox{.}(2012)]%
        {prat2012shaping}
\bibfield{author}{\bibinfo{person}{Arnau Prat-P{\'e}rez}, \bibinfo{person}{David Dominguez-Sal}, \bibinfo{person}{Josep~M Brunat}, {and} \bibinfo{person}{Josep-Lluis Larriba-Pey}.} \bibinfo{year}{2012}\natexlab{}.
\newblock \showarticletitle{Shaping communities out of triangles}. In \bibinfo{booktitle}{\emph{Proceedings of the ACM international Conference on Information and Knowledge Management}}. \bibinfo{pages}{1677--1681}.
\newblock


\bibitem[Qin et~al\mbox{.}(2020)]%
        {qin2020analysis}
\bibfield{author}{\bibinfo{person}{Lei Qin}, \bibinfo{person}{Yidan Wang}, \bibinfo{person}{Qiang Sun}, \bibinfo{person}{Xiaomei Zhang}, \bibinfo{person}{Ben-Chang Shia}, \bibinfo{person}{Chengcheng Liu}, {et~al\mbox{.}}} \bibinfo{year}{2020}\natexlab{}.
\newblock \showarticletitle{Analysis of the covid-19 epidemic transmission network in mainland china: K-core decomposition study}.
\newblock \bibinfo{journal}{\emph{JMIR public health and surveillance}} \bibinfo{volume}{6}, \bibinfo{number}{4} (\bibinfo{year}{2020}), \bibinfo{pages}{e24291}.
\newblock


\bibitem[Qin et~al\mbox{.}(2017)]%
        {Qin17}
\bibfield{author}{\bibinfo{person}{Zhan Qin}, \bibinfo{person}{Ting Yu}, \bibinfo{person}{Yin Yang}, \bibinfo{person}{Issa Khalil}, \bibinfo{person}{Xiaokui Xiao}, {and} \bibinfo{person}{Kui Ren}.} \bibinfo{year}{2017}\natexlab{}.
\newblock \showarticletitle{Generating Synthetic Decentralized Social Graphs with Local Differential Privacy}. In \bibinfo{booktitle}{\emph{Proceedings of the 2017 ACM SIGSAC Conference on Computer and Communications Security}} (Dallas, Texas, USA) \emph{(\bibinfo{series}{CCS '17})}. \bibinfo{publisher}{Association for Computing Machinery}, \bibinfo{address}{New York, NY, USA}, \bibinfo{pages}{425–438}.
\newblock
\showISBNx{9781450349468}
\urldef\tempurl%
\url{https://doi.org/10.1145/3133956.3134086}
\showDOI{\tempurl}


\bibitem[Rossi and Ahmed(2015)]%
        {nr}
\bibfield{author}{\bibinfo{person}{Ryan Rossi} {and} \bibinfo{person}{Nesreen Ahmed}.} \bibinfo{year}{2015}\natexlab{}.
\newblock \showarticletitle{The network data repository with interactive graph analytics and visualization}. In \bibinfo{booktitle}{\emph{Proceedings of the AAAI conference on artificial intelligence}}, Vol.~\bibinfo{volume}{29}.
\newblock


\bibitem[Roth et~al\mbox{.}(2021)]%
        {roth2021mycelium}
\bibfield{author}{\bibinfo{person}{Edo Roth}, \bibinfo{person}{Karan Newatia}, \bibinfo{person}{Yiping Ma}, \bibinfo{person}{Ke Zhong}, \bibinfo{person}{Sebastian Angel}, {and} \bibinfo{person}{Andreas Haeberlen}.} \bibinfo{year}{2021}\natexlab{}.
\newblock \showarticletitle{Mycelium: Large-scale distributed graph queries with differential privacy}. In \bibinfo{booktitle}{\emph{Proceedings of the ACM SIGOPS 28th Symposium on Operating Systems Principles}}. \bibinfo{pages}{327--343}.
\newblock


\bibitem[S~Monteiro et~al\mbox{.}(2021)]%
        {s2021superspreading}
\bibfield{author}{\bibinfo{person}{Higor S~Monteiro}, \bibinfo{person}{Shaojun Luo}, \bibinfo{person}{Saulo~DS Reis}, \bibinfo{person}{Carles Igual}, \bibinfo{person}{Antonio S~Lima~Neto}, \bibinfo{person}{Matias Travizan}, \bibinfo{person}{Jose Soares De Andrade~Jr}, \bibinfo{person}{Hernan Makse}, {et~al\mbox{.}}} \bibinfo{year}{2021}\natexlab{}.
\newblock \showarticletitle{Superspreading k-cores at the center of Covid-19 pandemic persistence}.
\newblock \bibinfo{journal}{\emph{Bulletin of the American Physical Society}}  \bibinfo{volume}{66} (\bibinfo{year}{2021}).
\newblock


\bibitem[Seidman(1983)]%
        {seidman1983network}
\bibfield{author}{\bibinfo{person}{Stephen~B Seidman}.} \bibinfo{year}{1983}\natexlab{}.
\newblock \showarticletitle{Network structure and minimum degree}.
\newblock \bibinfo{journal}{\emph{Social networks}} \bibinfo{volume}{5}, \bibinfo{number}{3} (\bibinfo{year}{1983}), \bibinfo{pages}{269--287}.
\newblock


\bibitem[Seif et~al\mbox{.}(2022)]%
        {seif2022differentially}
\bibfield{author}{\bibinfo{person}{Mohamed Seif}, \bibinfo{person}{Dung Nguyen}, \bibinfo{person}{Anil Vullikanti}, {and} \bibinfo{person}{Ravi Tandon}.} \bibinfo{year}{2022}\natexlab{}.
\newblock \showarticletitle{Differentially private community detection for stochastic block models}.
\newblock \bibinfo{journal}{\emph{arXiv preprint arXiv:2202.00636}} (\bibinfo{year}{2022}).
\newblock


\bibitem[Serafino et~al\mbox{.}(2020)]%
        {covid19kcore}
\bibfield{author}{\bibinfo{person}{Matteo Serafino}, \bibinfo{person}{Higor~S. Monteiro}, \bibinfo{person}{Shaojun Luo}, {and} \bibinfo{person}{Hernán~A. Makse}.} \bibinfo{year}{2020}\natexlab{}.
\newblock \bibinfo{booktitle}{\emph{{Project COVID19 K-core tracker}}}.
\newblock
\urldef\tempurl%
\url{https://github.com/makselab/COVID19}
\showURL{%
\tempurl}


\bibitem[Serafino et~al\mbox{.}(2022)]%
        {serafino2022digital}
\bibfield{author}{\bibinfo{person}{Matteo Serafino}, \bibinfo{person}{Higor~S Monteiro}, \bibinfo{person}{Shaojun Luo}, \bibinfo{person}{Saulo~DS Reis}, \bibinfo{person}{Carles Igual}, \bibinfo{person}{Antonio~S Lima~Neto}, \bibinfo{person}{Mat{\'\i}as Travizano}, \bibinfo{person}{Jos{\'e}~S Andrade~Jr}, {and} \bibinfo{person}{Hern{\'a}n~A Makse}.} \bibinfo{year}{2022}\natexlab{}.
\newblock \showarticletitle{Digital contact tracing and network theory to stop the spread of COVID-19 using big-data on human mobility geolocalization}.
\newblock \bibinfo{journal}{\emph{PLOS Computational Biology}} \bibinfo{volume}{18}, \bibinfo{number}{4} (\bibinfo{year}{2022}), \bibinfo{pages}{e1009865}.
\newblock


\bibitem[Shi et~al\mbox{.}(2011)]%
        {SCRCS11}
\bibfield{author}{\bibinfo{person}{Elaine Shi}, \bibinfo{person}{T.{-}H.~Hubert Chan}, \bibinfo{person}{Eleanor~Gilbert Rieffel}, \bibinfo{person}{Richard Chow}, {and} \bibinfo{person}{Dawn Song}.} \bibinfo{year}{2011}\natexlab{}.
\newblock \showarticletitle{Privacy-Preserving Aggregation of Time-Series Data}. In \bibinfo{booktitle}{\emph{Proceedings of the Network and Distributed System Security Symposium}}.
\newblock


\bibitem[Sotiropoulos and Tsourakakis(2021)]%
        {ST21}
\bibfield{author}{\bibinfo{person}{Konstantinos Sotiropoulos} {and} \bibinfo{person}{Charalampos~E. Tsourakakis}.} \bibinfo{year}{2021}\natexlab{}.
\newblock \showarticletitle{Triangle-Aware Spectral Sparsifiers and Community Detection}. In \bibinfo{booktitle}{\emph{Proceedings of the 27th ACM SIGKDD Conference on Knowledge Discovery \& Data Mining}} (Virtual Event, Singapore) \emph{(\bibinfo{series}{KDD '21})}. \bibinfo{publisher}{Association for Computing Machinery}, \bibinfo{address}{New York, NY, USA}, \bibinfo{pages}{1501–1509}.
\newblock
\showISBNx{9781450383325}
\urldef\tempurl%
\url{https://doi.org/10.1145/3447548.3467260}
\showDOI{\tempurl}


\bibitem[Srihari and Leong(2013)]%
        {SL13}
\bibfield{author}{\bibinfo{person}{Sriganesh Srihari} {and} \bibinfo{person}{Hon Leong}.} \bibinfo{year}{2013}\natexlab{}.
\newblock \showarticletitle{A survey of computational methods for protein complex prediction from protein interaction networks}.
\newblock \bibinfo{journal}{\emph{Journal of bioinformatics and computational biology}}  \bibinfo{volume}{11} (\bibinfo{date}{04} \bibinfo{year}{2013}), \bibinfo{pages}{1230002}.
\newblock
\urldef\tempurl%
\url{https://doi.org/10.1142/S021972001230002X}
\showDOI{\tempurl}


\bibitem[Sun et~al\mbox{.}(2019)]%
        {sun2019analyzing}
\bibfield{author}{\bibinfo{person}{Haipei Sun}, \bibinfo{person}{Xiaokui Xiao}, \bibinfo{person}{Issa Khalil}, \bibinfo{person}{Yin Yang}, \bibinfo{person}{Zhan Qin}, \bibinfo{person}{Hui Wang}, {and} \bibinfo{person}{Ting Yu}.} \bibinfo{year}{2019}\natexlab{}.
\newblock \showarticletitle{Analyzing subgraph statistics from extended local views with decentralized differential privacy}. In \bibinfo{booktitle}{\emph{Proceedings of the 2019 ACM SIGSAC conference on computer and communications security}}. \bibinfo{pages}{703--717}.
\newblock


\bibitem[Sánchez-Aguayo et~al\mbox{.}(2021)]%
        {computers10100121}
\bibfield{author}{\bibinfo{person}{Marco Sánchez-Aguayo}, \bibinfo{person}{Luis Urquiza-Aguiar}, {and} \bibinfo{person}{José Estrada-Jiménez}.} \bibinfo{year}{2021}\natexlab{}.
\newblock \showarticletitle{Fraud Detection Using the Fraud Triangle Theory and Data Mining Techniques: A Literature Review}.
\newblock \bibinfo{journal}{\emph{Computers}} \bibinfo{volume}{10}, \bibinfo{number}{10} (\bibinfo{year}{2021}).
\newblock
\showISSN{2073-431X}
\urldef\tempurl%
\url{https://doi.org/10.3390/computers10100121}
\showDOI{\tempurl}


\bibitem[Team(2017)]%
        {appleLearningWith}
\bibfield{author}{\bibinfo{person}{Differential~Privacy Team}.} \bibinfo{year}{2017}\natexlab{}.
\newblock \bibinfo{title}{{L}earning with {P}rivacy at {S}cale --- machinelearning.apple.com}.
\newblock \bibinfo{howpublished}{\url{https://machinelearning.apple.com/research/learning-with-privacy-at-scale}}.
\newblock
\newblock
\shownote{[Accessed 10-04-2024]}.


\bibitem[Team({[n.\,d.]})]%
        {githubGitHubGoogledifferentialprivacy}
\bibfield{author}{\bibinfo{person}{Google Differential~Privacy Team}.} \bibinfo{year}{[n.\,d.]}\natexlab{}.
\newblock \bibinfo{title}{{G}it{H}ub - google/differential-privacy: {G}oogle's differential privacy libraries. --- github.com}.
\newblock \bibinfo{howpublished}{\url{https://github.com/google/differential-privacy}}.
\newblock
\newblock
\shownote{[Accessed 12-04-2024]}.


\bibitem[Tire and Gursoy(2024)]%
        {10415644}
\bibfield{author}{\bibinfo{person}{Ekin Tire} {and} \bibinfo{person}{M.~Emre Gursoy}.} \bibinfo{year}{2024}\natexlab{}.
\newblock \showarticletitle{Answering Spatial Density Queries Under Local Differential Privacy}.
\newblock \bibinfo{journal}{\emph{IEEE Internet of Things Journal}} \bibinfo{volume}{11}, \bibinfo{number}{10} (\bibinfo{year}{2024}), \bibinfo{pages}{17419--17436}.
\newblock
\urldef\tempurl%
\url{https://doi.org/10.1109/JIOT.2024.3357570}
\showDOI{\tempurl}


\bibitem[Tseng et~al\mbox{.}(2021)]%
        {tseng2021parallel}
\bibfield{author}{\bibinfo{person}{Tom Tseng}, \bibinfo{person}{Laxman Dhulipala}, {and} \bibinfo{person}{Julian Shun}.} \bibinfo{year}{2021}\natexlab{}.
\newblock \showarticletitle{Parallel index-based structural graph clustering and its approximation}. In \bibinfo{booktitle}{\emph{Proceedings of the International Conference on Management of Data}}. \bibinfo{pages}{1851--1864}.
\newblock


\bibitem[Wang et~al\mbox{.}(2023b)]%
        {wang2023mago}
\bibfield{author}{\bibinfo{person}{Songlei Wang}, \bibinfo{person}{Yifeng Zheng}, \bibinfo{person}{Xiaohua Jia}, \bibinfo{person}{Qian Wang}, {and} \bibinfo{person}{Cong Wang}.} \bibinfo{year}{2023}\natexlab{b}.
\newblock \showarticletitle{MAGO: Maliciously Secure Subgraph Counting on Decentralized Social Graphs}.
\newblock \bibinfo{journal}{\emph{IEEE Transactions on Information Forensics and Security}} (\bibinfo{year}{2023}).
\newblock


\bibitem[Wang et~al\mbox{.}(2019)]%
        {10.1145/3299869.3319891}
\bibfield{author}{\bibinfo{person}{Tianhao Wang}, \bibinfo{person}{Bolin Ding}, \bibinfo{person}{Jingren Zhou}, \bibinfo{person}{Cheng Hong}, \bibinfo{person}{Zhicong Huang}, \bibinfo{person}{Ninghui Li}, {and} \bibinfo{person}{Somesh Jha}.} \bibinfo{year}{2019}\natexlab{}.
\newblock \showarticletitle{Answering Multi-Dimensional Analytical Queries under Local Differential Privacy}. In \bibinfo{booktitle}{\emph{Proceedings of the 2019 International Conference on Management of Data}} (Amsterdam, Netherlands) \emph{(\bibinfo{series}{SIGMOD '19})}. \bibinfo{publisher}{Association for Computing Machinery}, \bibinfo{address}{New York, NY, USA}, \bibinfo{pages}{159–176}.
\newblock
\showISBNx{9781450356435}
\urldef\tempurl%
\url{https://doi.org/10.1145/3299869.3319891}
\showDOI{\tempurl}


\bibitem[Wang et~al\mbox{.}(2023a)]%
        {wang2023differential}
\bibfield{author}{\bibinfo{person}{Yanling Wang}, \bibinfo{person}{Qian Wang}, \bibinfo{person}{Lingchen Zhao}, {and} \bibinfo{person}{Cong Wang}.} \bibinfo{year}{2023}\natexlab{a}.
\newblock \showarticletitle{Differential privacy in deep learning: Privacy and beyond}.
\newblock \bibinfo{journal}{\emph{Future Generation Computer Systems}} (\bibinfo{year}{2023}).
\newblock


\bibitem[Wang et~al\mbox{.}(2025)]%
        {WLT25}
\bibfield{author}{\bibinfo{person}{Zhuo Wang}, \bibinfo{person}{Zhixiong Li}, \bibinfo{person}{Jinxing Tu}, {and} \bibinfo{person}{Jianqiang Huang}.} \bibinfo{year}{2025}\natexlab{}.
\newblock \showarticletitle{DYTC:Dynamic Graph Triangle Counting on GPU}. In \bibinfo{booktitle}{\emph{Proceedings of the 2024 8th International Conference on Algorithms, Computing and Systems}} \emph{(\bibinfo{series}{ICACS '24})}. \bibinfo{publisher}{Association for Computing Machinery}, \bibinfo{address}{New York, NY, USA}, \bibinfo{pages}{83–87}.
\newblock
\showISBNx{9798400718304}
\urldef\tempurl%
\url{https://doi.org/10.1145/3708597.3708610}
\showDOI{\tempurl}


\bibitem[Warner(1965)]%
        {warner1965randomized}
\bibfield{author}{\bibinfo{person}{Stanley~L Warner}.} \bibinfo{year}{1965}\natexlab{}.
\newblock \showarticletitle{Randomized response: A survey technique for eliminating evasive answer bias}.
\newblock \bibinfo{journal}{\emph{J. Amer. Statist. Assoc.}} \bibinfo{volume}{60}, \bibinfo{number}{309} (\bibinfo{year}{1965}), \bibinfo{pages}{63--69}.
\newblock


\bibitem[Xiang et~al\mbox{.}(2023)]%
        {xiang2023preserving}
\bibfield{author}{\bibinfo{person}{Zihang Xiang}, \bibinfo{person}{Tianhao Wang}, {and} \bibinfo{person}{Di Wang}.} \bibinfo{year}{2023}\natexlab{}.
\newblock \showarticletitle{Preserving Node-level Privacy in Graph Neural Networks}.
\newblock \bibinfo{journal}{\emph{arXiv preprint arXiv:2311.06888}} (\bibinfo{year}{2023}).
\newblock


\bibitem[Ye et~al\mbox{.}(2020a)]%
        {ye2020lf}
\bibfield{author}{\bibinfo{person}{Qingqing Ye}, \bibinfo{person}{Haibo Hu}, \bibinfo{person}{Man~Ho Au}, \bibinfo{person}{Xiaofeng Meng}, {and} \bibinfo{person}{Xiaokui Xiao}.} \bibinfo{year}{2020}\natexlab{a}.
\newblock \showarticletitle{LF-GDPR: A framework for estimating graph metrics with local differential privacy}.
\newblock \bibinfo{journal}{\emph{IEEE Transactions on Knowledge and Data Engineering}} \bibinfo{volume}{34}, \bibinfo{number}{10} (\bibinfo{year}{2020}), \bibinfo{pages}{4905--4920}.
\newblock


\bibitem[Ye et~al\mbox{.}(2020b)]%
        {ye2020towards}
\bibfield{author}{\bibinfo{person}{Qingqing Ye}, \bibinfo{person}{Haibo Hu}, \bibinfo{person}{Man~Ho Au}, \bibinfo{person}{Xiaofeng Meng}, {and} \bibinfo{person}{Xiaokui Xiao}.} \bibinfo{year}{2020}\natexlab{b}.
\newblock \showarticletitle{Towards locally differentially private generic graph metric estimation}. In \bibinfo{booktitle}{\emph{2020 IEEE 36th International Conference on Data Engineering (ICDE)}}. IEEE, \bibinfo{pages}{1922--1925}.
\newblock


\bibitem[Zeng et~al\mbox{.}(2014)]%
        {ZZLSZ14}
\bibfield{author}{\bibinfo{person}{Wei Zeng}, \bibinfo{person}{An Zeng}, \bibinfo{person}{Hao Liu}, \bibinfo{person}{Ming-Sheng Shang}, {and} \bibinfo{person}{Tao Zhou}.} \bibinfo{year}{2014}\natexlab{}.
\newblock \showarticletitle{Uncovering the information core in recommender systems}.
\newblock \bibinfo{journal}{\emph{Scientific reports}}  \bibinfo{volume}{4} (\bibinfo{date}{August} \bibinfo{year}{2014}), \bibinfo{pages}{6140}.
\newblock
\showISSN{2045-2322}
\urldef\tempurl%
\url{https://doi.org/10.1038/srep06140}
\showDOI{\tempurl}


\bibitem[Zhong et~al\mbox{.}(2023)]%
        {zhong2023disparate}
\bibfield{author}{\bibinfo{person}{Da Zhong}, \bibinfo{person}{Ruotong Yu}, \bibinfo{person}{Kun Wu}, \bibinfo{person}{Xiuling Wang}, \bibinfo{person}{Jun Xu}, {and} \bibinfo{person}{Wendy~Hui Wang}.} \bibinfo{year}{2023}\natexlab{}.
\newblock \showarticletitle{Disparate Vulnerability in Link Inference Attacks against Graph Neural Networks}.
\newblock \bibinfo{journal}{\emph{Proceedings on Privacy Enhancing Technologies}} (\bibinfo{year}{2023}).
\newblock


\bibitem[Zhu et~al\mbox{.}(2023)]%
        {zhu2023blink}
\bibfield{author}{\bibinfo{person}{Xiaochen Zhu}, \bibinfo{person}{Vincent~YF Tan}, {and} \bibinfo{person}{Xiaokui Xiao}.} \bibinfo{year}{2023}\natexlab{}.
\newblock \showarticletitle{Blink: Link Local Differential Privacy in Graph Neural Networks via Bayesian Estimation}. In \bibinfo{booktitle}{\emph{Proceedings of the 2023 ACM SIGSAC Conference on Computer and Communications Security}}. \bibinfo{pages}{2651--2664}.
\newblock


\end{thebibliography}






\end{document}